\documentclass[11pt]{article}

\usepackage[protrusion=false]{microtype}
\usepackage{amsmath,amssymb,amsthm,mathtools,accents,setspace,bm,bbm}
\newcommand{\slightspacing}{\setstretch{1.15}}

\DeclareMathOperator*{\argmin}{arg\,min}

\usepackage[margin=1in]{geometry}
\usepackage{xcolor}
\definecolor{Bleu}{RGB}{0,0,204}
\definecolor{CB5red}{HTML}{DA0000}
\definecolor{CB5blue}{HTML}{0066FF}
\definecolor{CB5purple}{HTML}{8000B3}
\definecolor{customblue}{HTML}{2166AC}
\definecolor{customred}{HTML}{B2182B}

\usepackage{subcaption}                 
\usepackage{multirow,makecell,booktabs}
\usepackage{tikz,adjustbox}
\usetikzlibrary{decorations.pathreplacing,arrows.meta}

\usepackage{algorithm,algpseudocode}

\usepackage{enumitem,comment,appendix}

\usepackage[
  backend=biber,
  style=apa,
  natbib=true,
  giveninits=true,
  uniquename=init,
  maxnames=2, minnames=1,
  doi=false, url=false
]{biblatex}

\AtEveryBibitem{%
  \clearfield{note}%
  \clearfield{eprint}%
  \clearfield{url}%
  \clearfield{doi}%
  \clearfield{issn}%
  \clearfield{isbn}%
  \clearfield{urldate}%
  \clearfield{month}%
  \clearfield{day}%
  \clearname{editor}%
}

\usepackage{hyperref}
\hypersetup{
  colorlinks = true,
  linkcolor  = Bleu,
  citecolor  = Bleu,
  urlcolor   = Bleu
}
\usepackage{cleveref}
\Crefformat{equation}{(#2#1#3)}

\usepackage{ulem}
\makeatletter
\toks@\expandafter{\@endtheorem\@endpetrue}
\edef\@endtheorem{\the\toks@}
\makeatother

\usepackage{titletoc}
\newcommand\DoToC{%
  \startcontents
  \printcontents{}{1}{\vskip 1.5em\hrule\vskip .75em}
  \vskip .75em\hrule\vskip 2em
}

\usetikzlibrary{decorations.pathreplacing,arrows.meta}

\newtheorem{theorem}{Theorem}
\newtheorem{lemma}{Lemma}

\theoremstyle{remark}
\newtheorem{condition}{Condition}

\newlist{enumcond}{enumerate}{1}
\setlist[enumcond,1]{label=\textup{(\alph*)},
                     ref=\thecondition\textup{(\alph*)}}
\crefname{enumcondi}{condition}{conditions}
\Crefname{enumcondi}{Condition}{Conditions}

\crefname{condition}{condition}{conditions}
\Crefname{condition}{Condition}{Conditions}

\newcommand{\E}{\mathbb{E}}

\DeclareRobustCommand{\uQ}{\underaccent{\bar}{Q}}

\DeclareMathOperator*{\essinf}{ess\,inf}

\date{}
\begin{document}
\allowdisplaybreaks
\def\spacingset#1{\renewcommand{\baselinestretch}%
{#1}\small\normalsize} \spacingset{1}
  
\title{Estimation of causal dose-response functions under data fusion}
\author{Jaewon Lim$^1$, and Alex Luedtke$^{2,3}$ \\
\\$^1$Department of Biostatistics, University of Washington
\\$^2$Department of Health Care Policy, Harvard University
\\$^3$Department of Statistics, University of Washington}
\maketitle
	
\bigskip
\begin{abstract}\slightspacing
Estimating the causal dose-response function is challenging, particularly when data from a single source are insufficient to estimate responses precisely across all exposure levels. To overcome this limitation, we propose a data fusion framework that leverages multiple data sources that are partially aligned with the target distribution. Specifically, we derive a Neyman-orthogonal loss function tailored for estimating the dose-response function within data fusion settings. To improve computational efficiency, we propose a stochastic approximation that retains orthogonality. We apply kernel ridge regression with this approximation, which provides closed-form estimators. Our theoretical analysis demonstrates that incorporating additional data sources yields tighter finite-sample regret bounds and improved worst-case performance, as confirmed via minimax lower bound comparison. Simulation studies validate the practical advantages of our approach, showing improved estimation accuracy when employing data fusion. This study highlights the potential of data fusion for estimating non-smooth parameters such as causal dose-response functions.
\end{abstract}
	

\slightspacing

\section{Introduction}
The causal dose-response function (CDRF) maps a level of continuous or multi-valued exposure to a corresponding expected counterfactual outcome. The exposure can be a potentially beneficial treatment, such as the duration of training \citep{kluve_evaluating_2012} or the frequency of therapy sessions \citep{robinson_dose-response_2020}, or a harmful contaminant, such as the concentration of air pollution \citep{dominici_air_2002}. CDRFs are also referred to as causal dose-response curves \citep{diaz_targeted_2013}, dose-response functions \citep{bonvini_fast_2022}, average dose-response functions \citep{galvao_uniformly_2015}, or average structural functions \citep{blundell_endogeneity_2004}. Estimating a CDRF is challenging because, typically, few, if any, individuals in a dataset have received most dose levels \citep{bahadori_end--end_2022}.

Extensive work has been devoted to estimating CDRFs. Many existing approaches rely on parametric assumptions. An early formal approach, introduced by \citet{robins_marginal_2000}, was based on a marginal structural model, which fits a parametric model using inverse-probability weighting. This approach requires the parametric model to contain the true CDRF. Later, the assumption of correct model specification was relaxed by projecting CDRFs onto the function space implied by a particular parametric model for more robust estimation \citep{neugebauer_nonparametric_2007}.
\citet{ai_unified_2021} introduced a weighted optimization scheme for estimating CDRFs, employing weights suggested by \citet{robins_marginal_2000}.
\citet{yiu_covariate_2018} constructed balancing weights that render pretreatment variables unassociated with treatment assignment after weighting.
\citet{hirano_propensity_2004,imai_causal_2004} used parametric generalized propensity scores to estimate CDRFs, and \citet{christian_fong_covariate_2018} proposed a covariate-balancing generalized propensity score methodology.

More recently, a growing body of literature has focused on estimating CDRFs using nonparametric methods and orthogonal statistical learning \citep{bonvini_fast_2022, foster_orthogonal_2023}.
\citet{flores_estimation_2007} considered a plug-in, kernel-based estimator of CDRFs.
\citet{singh_kernel_2023} explored kernel ridge regression to develop simple, closed-form estimators of CDRFs.
\citet{galvao_uniformly_2015} proposed a semiparametric, two-step estimator that first estimates a ratio of conditional densities and then estimates the CDRF.
\citet{kennedy_non-parametric_2016,bonvini_fast_2022} introduced a two-step estimation procedure in which nuisance functions are estimated and then a pseudo-outcome is regressed on the exposure variable.
\citet{westling_causal_2020} proposed a nonparametric estimator for monotone CDRFs.
Given multiple candidate CDRF estimators, early work demonstrated how to use cross-validation to select among them, with corresponding oracle guarantees \citep{van_der_laan_unified_2003, diaz_targeted_2013}.

Recent approaches include the estimation of CDRFs in nonstandard settings. \citet{huang_nonparametric_2023} studied settings in which the continuous treatment variable is measured with error. They proposed a method based on weighting and generalized empirical likelihood techniques. \citet{zeng_continuous_2024} examined scenarios in which the primary outcome is missing for some proportion of the data. They developed a two-step estimation procedure similar to that of \citet{bonvini_fast_2022}, but their method constructs pseudo-outcomes by leveraging both labeled and unlabeled data through a surrogate outcome.

Previous approaches have considered only samples originating from a single set of trial or observational data. This limitation makes the estimation of CDRFs practically difficult because data from a single source may not be sufficient for precise estimation across all exposure levels. To address this issue, we consider data fusion, which enables the combination of information from distinct sources that are partially aligned with the target distribution \citep{bareinboim_causal_2016}. Data fusion has been applied to the transportability of causal effects \citep{hernan_compound_2011, pearl_transportability_2011, bareinboim_transportability_2014, stuart_assessing_2015, rudolph_robust_2016, dahabreh_extending_2019, dahabreh_efficient_2022}. \citet{qiu_efficient_2024} combined information from an auxiliary population to address dataset shift and to estimate the target population risk. \citet{sun_semiparametric_2022} combined information from two data sources—one containing treatment information and the other outcome information—to estimate average treatment effects via an instrumental variable. \citet{li_efficient_2023,graham_towards_2025,li_data_2025} developed general frameworks for the efficient estimation of smooth, finite-dimensional parameters under data fusion.

We apply a general form of data fusion, leveraging one or more sources. Our first contribution is as follows:
\begin{enumerate}
\item We derive a Neyman-orthogonal loss for estimating the CDRF in data fusion settings.
\end{enumerate}
Evaluating this loss can be computationally expensive.
To overcome this challenge:
\begin{enumerate}[resume*]
\item We propose a stochastic approximation of the loss that retains Neyman orthogonality. When used with kernel ridge regression, it yields a closed-form estimator.
\item We establish finite-sample regret upper bounds for both kernel ridge regression and empirical risk minimizers. These bounds become tighter as additional sources are incorporated into the analysis.
\end{enumerate}
Since upper bounds can be loose, the above results do not guarantee that incorporating additional data sources will necessarily improve performance.
Our final contribution is to investigate this question:
\begin{enumerate}[resume*]
\item We characterize settings in which using data fusion necessarily reduces the risk of estimating CDRFs with high probability. In doing so, we establish a high-probability worst-case lower bound for the risk of estimating CDRFs without data fusion.
\end{enumerate}

\section{Data Fusion Framework and Estimation Algorithm}
\subsection{Statistical Setting and Parameters of Interest}

We begin by defining the parameter of interest and its associated risk.
Let $(X, A, Y) \sim Q^{0}$, where $X \in \mathcal{X} \subset \mathbb{R}^r$ denotes the covariates, $A \in \mathcal{A} = [0, 1]^d$ the exposure, $Y \in \mathcal{Y} \subset \mathbb{R}$ the outcome of an experiment, and $Q^{0}$ the target distribution.
Our parameter of interest, the CDRF, maps a realized exposure $a$ to the expected value of the potential outcome $Y^{a}$.
Under the assumptions discussed in \citet{westling_causal_2020}, the CDRF can be identified as
\begin{align*}
\theta_{Q^{0}}(a) = \E\!\left[ Y^{a} \right] &= \int \E\!\left[ Y \mid X = x, A = a \right] \, Q^{0}_{X}(dx).
\end{align*}
The target distribution is assumed to belong to a model $\mathcal{Q}$ of distributions for the random vector $(X, A, Y)$—some of our later results will impose moment and smoothness restrictions on this model, but for now we leave it unspecified.
For any $Q \in \mathcal{Q}$, we denote by $Q_{X}$ the marginal distribution of $X$, by $Q_{A \mid X}$ the conditional distribution of $A$ given $X$, and by $Q_{Y \mid X, A}$ the conditional distribution of $Y$ given $(X, A)$.

We evaluate performance by integrating a squared error loss with respect to $\mu$, a prespecified measure on $\mathcal{A}$ that dominates $Q^{0}_{A \mid X}$ with $Q_X^0$-probability one. Hence, the risk at $\theta: \mathcal{A} \to \mathbb{R}$ is
\begin{align*}
  \mathcal{L}_{Q^{0}}(\theta)
  &:= \int \!\bigl[\theta(a) - \theta_{Q^{0}}(a)\bigr]^{2}\, \mu(da) = 
  \int \!\E_{Q^{0}_{X}}\!\Bigl[
        \E_{Q^{0}_{Y \mid X, A}}\!\bigl[\theta(A) - Y \,\big|\, X, A = a\bigr]
      \Bigr]^{2}\! \mu(da).
\end{align*}
Our data do not consist of draws directly from the target distribution $Q^{0}$.  
Instead, they contain samples from $k$ data sources, formalized as observing $n$ independent copies of $(X, A, Y, S) \sim P^{0}$. Here, $S$ is a categorical random variable with support $[k]$ indicating the data source. For any distribution $P$, we define $P_{X}(\cdot \mid s)$ as the conditional distribution of $X \mid S = s$,  
$P_{A \mid X}(\cdot \mid x, s)$ as that of $A \mid X = x, S = s$,  
and $P_{Y \mid X, A}(\cdot \mid x, a, s)$ as that of $Y \mid X = x, A = a, S = s$.  
We assume that all such conditional distributions are well defined on common measurable spaces.

Under the following condition, we establish a connection between the conditional distributions of $P^{0}$ and $Q^{0}$.  
Let $\mathcal{S}_{X}$ and $\mathcal{S}_{Y}$ denote known collections of data sources whose distributions of $X$ and $Y \mid X, A$ align with $Q^{0}_{X}$ and $Q^{0}_{Y \mid X, A}$, respectively, in the sense defined below.

\begin{condition}[Data Fusion Conditions]\label{cond:identifiability}
All of the following hold:
\begin{enumcond}
    \item \label{cond:source_probabilities}
    (Positive correctly aligned sources) 
     There exist constants $M_{\xi}, M_{\eta} \ge 1$ such that
    \[
      P^{0}(S \in \mathcal{S}_{X}) \;\ge\; M_{\xi}^{-1}, 
      \quad P^{0}(S \in \mathcal{S}_{Y}) \;\ge\; M_{\eta}^{-1}.
    \]
    
    \item \label{cond:overlap}(Sufficient overlap)  
    $Q^{0}_{X,A}(\cdot)$ is absolutely continuous with respect to $P^{0}_{X,A}(\,\cdot \mid S = s)$ for all $s \in \mathcal{S}_{Y}$.  
    Moreover, there exists $c_{1} \geq 1$ such that 
    \[
      c_{1}^{-1} \;\leq\; \frac{dQ_{X, A}^{0}}{dP_{X, A}^{0}(\,\cdot \mid S \in \mathcal{S}_{Y})}(x,a) \;\leq\; c_{1},
      \quad Q^{0}\text{-a.e. } (x,a).
    \]

    \item \label{cond:common}(Exchangeability of conditionals) 
    For all $s \in \mathcal{S}_{X}$, $P^{0}_{X}(\,\cdot \mid S = s) = Q^{0}_{X}(\cdot)$. For all $s \in \mathcal{S}_{Y}$,
    \[
      P^{0}_{Y \mid X,A}(\,\cdot \mid X=x, A=a, S=s) 
      \;=\; Q^{0}_{Y \mid X,A}(\,\cdot \mid X=x, A=a),
      \quad Q^{0}_{X,A}\text{-a.e.}
    \]
\end{enumcond}
\end{condition}

The statistical model $\mathcal{P}$ is defined as the set of all distributions $P^{0}$ on some measurable space such that $(P^{0}, Q^{0})$ satisfy \Cref{cond:identifiability} for some $Q^{0} \in \mathcal{Q}$.  
\Cref{cond:identifiability} can be relaxed in the sense that it suffices for there to exist at least one $\uQ \in \mathcal{Q}$ satisfying the condition, even if the true $Q^{0}$ fails to meet \Cref{cond:overlap}, as discussed in \citet{li_efficient_2023}.  
This relaxation is justified because the risk does not depend on the distribution of $A \mid X$, which is not identifiable.  
\Cref{cond:source_probabilities} guarantees that the probabilities of correctly and partially aligned sources are uniformly bounded away from zero across $\mathcal{P}$.  
\Cref{cond:overlap,cond:common} are used to identify the risk through the observed distribution $P^{0}$.  
Specifically, \Cref{cond:overlap} ensures that the conditional distribution $P^{0}_{Y \mid X, A}(\cdot \mid x, a, S = s)$ is uniquely defined up to $Q^{0}$-null sets, while \Cref{cond:common} enables us to link the risk defined by the conditional distributions of $Q^{0}$ to those of $P^{0}$.  
Moreover, we require \Cref{cond:overlap} to define a Neyman–orthogonal loss for the risk in the next section.

By Theorem~1 in \citet{li_efficient_2023}, \Cref{cond:overlap,cond:common} allow us to identify $\theta_{Q^{0}}$ and $\mathcal{L}_{Q^{0}}(\theta)$ with the following functions indexed by $P^{0}$:
\begin{align*}
  \mathcal{L}_{Q^{0}}(\theta)
  &= L_{P^{0}}(\theta)
  := \int \!
     \E_{P^{0}}\!\Bigl[
       \E_{P^{0}}\!\bigl[\theta(A) - Y \,\big|\, X, A = a, S \in \mathcal{S}_{Y}\bigr]
       \Big|\, S \in \mathcal{S}_{X}
     \Bigr]^{2}\! \mu(da).
\end{align*}
In view of this identification result, $\mathcal{L}_{Q^{0}}(\theta)$ can be estimated using data drawn from $P^0$. The same applies to the global risk minimizer $\theta_{Q^{0}}$, which is identified with
\begin{align*}
  \theta_{Q^{0}}(a)
  &= \theta_{0}(a)
  := \E_{P^{0}}\!\Bigl[
       \E_{P^{0}}\!\bigl[Y \,\big|\, X = x, A = a, S \in \mathcal{S}_{Y}\bigr]
       \Big|\, S \in \mathcal{S}_{X}
     \Bigr].
\end{align*}

\subsection{Derivation of Neyman-orthogonal Loss}

Our goal is to find a minimizer of the population risk, denoted by $L_{P^{0}}(\theta)$, within a function class $\Theta \subseteq L^{2}(\mu)$.  
Because the population risk involves nuisance parameters, we aim to construct a Neyman–orthogonal loss function such that the estimation error of the nuisance parameters has a fourth-order impact on the oracle population risk \citep{foster_orthogonal_2023, curth_estimating_2021}.  
This property is particularly beneficial because estimators of some nuisance parameters on which the problem relies typically do not converge at the $n^{-1/2}$ rate when estimated using machine-learning methods \citep{sugiyama_density_2012}.

A candidate for a Neyman-orthogonal loss is identified as the loss function whose empirical mean corresponds to a one-step estimator of $L_{P^0}(\theta)$, $L_{OS, \widehat{P}_{n}}(\theta):= L_{\widehat{P}_{n}}(\theta) + \mathbb{P}_{n}D^{*}_{\widehat{P}_{n}}(\theta)$, based on an initial estimate $\widehat{P}_{n}$ of $P^0$ \citep{bickel_adaptive_1982}. This estimator adjusts for the bias induced by plug-in estimation of estimator $\widehat{P}_{n}$ of $P$, where $D^{*}_{P}(\theta)$ denotes a gradient of $P\mapsto L_P(\theta)$ at $P$. By Corollary~1 in \citet{li_efficient_2023}, $D^{*}_{P}(\theta)$ can be obtained as the projection of $D^{*}_{Q}(\theta)$ onto the tangent space of $\mathcal{P}$ at $P$. 

The resulting one-step estimator of $L_{P^{0}}(\theta)$ is given by the empirical mean
\[
\mathbb{P}_{n}\!\left[\ell^{\text{OS}}_{P^{0}}(\theta, g_{0}; \cdot)\right],
\]
where $\ell^{\text{OS}}_{P^{0}}(\theta, g_{0}; (x, a, y, s))
:= \E_{B \sim \mu}\!\left[\ell_{P^{0}}(\theta, g_{0}; (x, a, y, s, B))\right]$ and
\begin{align}
\ell_{P^{0}}(\theta, g_{0}; (x, a, y, s, b))
&:= \bigl(\theta(b) - \theta_{0}(b)\bigr)^{2}
+ 2\xi_{0}\,\mathbf{1}(s \in \mathcal{S}_{X})
   \bigl(\theta(b) - \theta_{0}(b)\bigr)
   \bigl(\theta_{0}(b) - m_{0}(x, b)\bigr)  \nonumber\\
&\quad
+ 2\eta_{0}\,w_{0}(x, a)\,\mathbf{1}(s \in \mathcal{S}_{Y})
   \bigl(\theta(a) - \theta_{0}(a)\bigr)
   \bigl(m_{0}(x, a) - y\bigr),
\label{eq:loss}
\end{align}

with nuisance parameters defined as
\begin{align*}
&g_{0}: (a, x) \mapsto (\xi_{0}, \eta_{0}, w_{0}(a, x), m_{0}(a, x), \theta_{0}(a)), \\
&\xi_0 := \frac{1}{P^0(S \in \mathcal{S}_X)}, \quad \eta_0 := \frac{1}{P^0(S \in \mathcal{S}_Y)}, \quad w_0(x, a) := \frac{p^0(x \mid S \in \mathcal{S}_X)}{p^0(x \mid S \in \mathcal{S}_Y)} \cdot \frac{1}{p^0(a \mid x, S \in \mathcal{S}_Y)} \\
&m_0(x, a) := \E_{P^0}[Y \mid x, a, S \in \mathcal{S}_Y].
\end{align*}
We define \( p^{0}(x \mid S \in \mathcal{S}_X) \) and \( p^{0}(x \mid S \in \mathcal{S}_Y) \) as densities with respect to a common dominating measure on \(\mathcal{X}\). Similarly, \( p(a \mid x, S \in \mathcal{S}_Y) \) denotes the conditional density of \( A \) given \( X \) and \( S \in \mathcal{S}_Y \) with respect to the known measure \( \mu \) on \(\mathcal{A}\).

We employ a stochastic approximation of the one-step estimator because it involves an expectation with respect to the known measure $\mu$, which makes direct optimization computationally challenging in practice.  
To approximate this expectation, we draw $n$ independent samples $B_{1}, \dots, B_{n} \sim \mu$ and denote a realization of the full data vector by $z := (x, a, y, s, b)$.  
For notational convenience, we denote the (oracle) loss function by $\ell_{P^{0}}(\theta, g_{0}; z)$, as defined in~\eqref{eq:loss}, for the remainder of the paper.

Distributions in $\mathcal{Q}$ are required to satisfy the following moment conditions:
\begin{condition}[Constraints on Model for Target Distribution]\label{cond:target_model}
For each $Q \in \mathcal{Q}$:
\begin{enumcond}
    \item \text{(Bounded conditional mean)} 
    \label{cond:bounded-m} $|m_{0}(x, a)| \leq 1$ $Q_{X} \times \mu$-almost everywhere (a.e.). 
In particular, this implies that $\big|\theta_{0}(a)\big| = \big|\E_{Q_X}\big[m_{0}(x, a)\big] \big| \leq 1$ $\mu$-a.e.
    \item \text{(Sub-exponential tails)} 
    \label{cond:subexp-y}
    There exist constants $\sigma, L > 0$ such that, for all integers $m \geq 2$,
    \[
    \int |y - \theta_{Q}(a)|^{m}\, Q(dy \mid x, a) \leq \tfrac{1}{2}m!\,\sigma^{2}L^{m-2}.
    \]
    \item \text{(Bounded density ratio between $\mu$ and $Q_{A\mid X}$)} 
    \label{cond:dmu-dqax-bdd}
    There exists a constant $c_{\mu} \ge 1$ such that for $Q_{X}$-almost every $x$,
    \begin{align*}
    \mu \ll Q_{A \mid X}(\cdot \mid x)
    \quad\text{and}\quad
    c_{\mu}^{-1} \;\le\;
    \frac{d\mu}{dQ_{A \mid X}(\cdot \mid x)}(a)
    \;\le\; c_{\mu},
    \quad Q_{A \mid X}(\cdot \mid x)\text{-a.e.\ } a.
    \end{align*}
\end{enumcond}
\end{condition}

The bounded conditional mean assumption in \Cref{cond:bounded-m} can be enforced, without loss of generality, by rescaling the outcome $Y$ whenever the conditional mean is uniformly bounded by some finite constant. The sub-exponential tail condition in \Cref{cond:subexp-y} allows the outcome to be unbounded while ensuring that its tail probabilities are appropriately controlled. More precisely, for each $(X, A)$ it implies that 
\[
P_{Q}\big(|Y - \theta_{Q}(A)| \geq t \,\big|\, X = x, A = a\big) \leq 2\exp(-t/L), \quad t \gg 0,
\]
by Proposition~2.7.1 in \citet{vershynin_high-dimensional_2018}. \Cref{cond:dmu-dqax-bdd} states that the Radon--Nikodym derivative of $\mu$ with respect to $Q_{A \mid X}(\cdot \mid x)$ is uniformly bounded away from zero and infinity over $\mathcal{Q}$. 
Combining \Cref{cond:overlap,cond:dmu-dqax-bdd}, we also obtain that, for any $P^{0} \in \mathcal{P}$, $M_{w}^{-1} \leq w_{0}(x,a) \;\leq\; M_{w}$ for $M_{w} := c_{1} c_{\mu} \ge 1$. Hence, $w_{0}$ is also uniformly bounded away from zero and infinity.

We define the nuisance parameter space $\mathcal{G}$ (with $g_{0} \in \mathcal{G}$) associated with $\mathcal{P}$ as the set of tuples 
$g = (\xi, \eta, w, m, \tau)$, where $\xi, \eta \in (0, 1]$, 
$w, m\!:\! \mathcal{X} \times \mathcal{A} \to \mathbb{R}$, 
and $\tau\!:\! \mathcal{A} \to \mathbb{R}$, satisfying the following condition.

\begin{condition}[Uniform Boundedness of Nuisances]\label{cond:nuisances}
For each $g \in \mathcal{G}$, the constants $M_{\xi}, M_{\eta}, M_{w}$ serve as universal bounds for any $P^{0} \in \mathcal{P}$ such that:
\mbox{}
\begin{enumcond}
    \item \text{(Positive probability for relevant sources)} 
    \label{cond:bounded-xi-eta}
    $\xi \le M_{\xi}$ and $\eta \le M_{\eta}$.
    
    \item \text{(Bounded density ratio and conditional mean)} 
    \label{cond:bounded-w-m}
    \[
      M_{w}^{-1} \le w(x, a) \le M_{w}, 
      \qquad |m(x, a)| \le 1 \quad Q_X^{0}\times \mu\text{-a.e.}, 
      \qquad |\tau(a)|\le 1 \quad \mu\text{-a.e.}.
    \]
\end{enumcond}
\end{condition}

Condition~\ref{cond:bounded-xi-eta} requires that the source probability estimators are uniformly bounded away from zero over $\mathcal{P}$.
Condition~\ref{cond:bounded-w-m} requires that the density ratio estimators of $w_{0}$ are uniformly bounded away from zero and infinity, and that the estimators of the conditional mean $m_{0}$ and the CDRF $\theta_{0}$ are uniformly bounded by one. We also define a constant $M_{\lambda}:= M_{w}^{-1}(M_{\eta} + M_{w} + 2)$. 
The loss function defined in \Cref{eq:loss} can be extended to any $g \in \mathcal{G}$ by replacing $g_{0}$ with $g$.  
Building on this extension, we define the expected loss as a function of both the parameter of interest and the nuisance parameters:
\[
  L_{P^{0}}(\theta, g)
  := \E_{\widetilde{P}^{0}}\!\big[\ell_{P^{0}}(\theta, g; Z)\big],
\]
where $\widetilde{P}^{0} := P^{0} \times \mu$ denotes the product measure on the extended space $Z := (X, A, Y, S, B)$.  
In particular, when evaluated at the true nuisance parameters, the loss reduces to the target functional:
\[
  L_{P^{0}}(\theta, g_{0}) = L_{P^{0}}(\theta).
\]

\subsection{Estimation Algorithms}
Using the loss function defined above, we propose two estimation algorithms-a generic two-stage estimation strategy and one based on kernel ridge regression-both of which employ sample splitting \citep{chernozhukov_doubledebiased_2018, foster_orthogonal_2023}. We divide the sample into two non-overlapping subsets, estimate nuisance parameters using the first subset, and then estimate the target parameter based on the estimated nuisance parameters using the remaining subset. 

\Cref{alg:gen_est} implements empirical risk minimization over a function class 
$\Theta \subseteq L^{2}(\mu)$ subject to the constraint 
$\sup_{\theta \in \Theta}\|\theta\|_{L^{\infty}(\mu)} \le 1$.  
The sampling weights $\xi_{0}$ and $\eta_{0}$ are estimated as the inverses of the empirical probabilities of $S \in \mathcal{S}_{X}$ and $S \in \mathcal{S}_{Y}$, respectively.  
The density ratio $w_{0}(x, a)$ is estimated using direct density-ratio methods such as Unconstrained Least-Squares Importance Fitting (uLSIF) or the Kullback–Leibler Importance Estimation Procedure (KLIEP) \citep{sugiyama_density_2012}.  
The outcome regression $m_{0}(x, a)$ is fit using a flexible learning algorithm, 
for example, Super Learner (via the \texttt{SuperLearner} package).

\begin{algorithm}[tb]
\caption{Estimation of CDRF over a general function class $\Theta$}
\label{alg:gen_est}
\begin{algorithmic}[1]
\State \textbf{Input:} Partition the sample $z^{n}$ into two disjoint subsamples $z^{n}_{1}$ and $z^{n}_{2}$.
\State \textbf{Nuisance estimation:} Estimate nuisance parameters $\widehat{g}$ using only $z^{n}_{1}$:
\begin{align*}
  &\widehat{\xi} := \frac{1}{\widehat{P}_{n}(S \in \mathcal{S}_{X})}, 
  \qquad
  \widehat{\eta} := \frac{1}{\widehat{P}_{n}(S \in \mathcal{S}_{Y})}, 
  \qquad
  \widehat{w}(x, a) := \text{direct estimator of } w_{0}(x, a), \\[4pt]
  &\widehat{m}(x, a) := \E_{\widehat{P}_{n}}\![Y \mid X = x, A = a, S \in \mathcal{S}_{Y}], 
  \qquad
  \widehat{\tau}_{0}(a) := \E_{\widehat{P}_{n}}\![\widehat{m}(X, a) \mid S \in \mathcal{S}_{X}].
\end{align*}
\State \textbf{Target estimation:} Find the empirical-risk minimizer $\widehat{\theta}$ using only $z^{n}_{2}$:
\begin{align*}
  \widehat{\theta}
  &= \underset{\theta \in \Theta}{\arg\min} \;
     \mathbb{P}^{(2)}_{n}\!\big[\ell_{P^{0}}(\theta, \widehat{g}; \cdot)\big],
\end{align*}
where $\mathbb{P}^{(2)}_{n}$ denotes the empirical distribution based on $z^{n}_{2}$.
\end{algorithmic}
\end{algorithm}

\Cref{alg:rkhs_est} uses kernel ridge regression.  
Let $\mathcal{H}$ denote a reproducing kernel Hilbert space (RKHS) with norm $\|\cdot\|_{\mathcal{H}}$ and kernel $\mathcal{K}$.  
We consider the constrained RKHS ball 
\begin{align*}
\mathcal{H}_{c}
  = \bigl\{ \theta \in \mathcal{H} \,\big|\, 
    \|\theta\|_{\mathcal{H}} \le c,\;
    \|\theta\|_{L^{\infty}(\mu)} \le 1
  \bigr\},    
\end{align*}
where $c$ is a hyperparameter selected by cross-validation.  
Given $z$, the kernel ridge regression problem becomes an empirical loss minimization over the constrained class $\mathcal{H}_{c}$.

\begin{algorithm}[tb]
\caption{Estimation of the CDRF using kernel ridge regression}
\label{alg:rkhs_est}
\begin{algorithmic}[1]
\State \textbf{Input:} Two disjoint subsamples $z^{n}_{1}$ and $z^{n}_{2}$ obtained by splitting the full sample $z^{n}$.
\State \textbf{Nuisance estimation:} Estimate nuisance parameters using only $z^{n}_{1}$ as described in \Cref{alg:gen_est}, yielding $\widehat{g}$.
\State \textbf{Target estimation:}
  \begin{enumerate}
    \item For each $i = 1, \ldots, |z^{n}_{2}|$, define
    \begin{align*}
      u_{i} &= \mathbf{1}(s_{i} \in \mathcal{S}_{Y})\, \widehat{\eta}\, \widehat{w}(x_{i}, a_{i}) \bigl(\widehat{m}(x_{i}, a_{i}) - y_{i}\bigr),\\
      v_{i} &= \mathbf{1}(s_{i} \in \mathcal{S}_{X})\, \widehat{\xi}\,
              \bigl(\widehat{\tau}(b_{i}) - \widehat{m}(x_{i}, b_{i})\bigr)
              - \widehat{\tau}(b_{i}).
    \end{align*}
    \item Define the Gram matrix
      $\mathbf{K} = (K_{ij}) \in \mathbb{R}^{2|z^{n}_{2}| \times 2|z^{n}_{2}|}$,
      where
      $K_{ij} = \mathcal{K}(a'_{i}, a'_{j}) / |z^{n}_{2}|$
      and $a' = (a_{i}) \cup (b_{i})$ for $i = 1, \ldots, |z^{n}_{2}|$.
    \item Use cross-validation (\Cref{alg:cross_val}) to select $\lambda_{n}$ and compute the kernel weights:
    \begin{align*}
      \widehat{\beta} &= -\frac{\mathbf{u}}{\lambda_{n} \sqrt{|z^{n}_{2}|}}, \qquad
      \widehat{\gamma}
      = -(\mathbf{K}_{22} + \lambda_{n} \mathbf{I})^{-1}
        \!\left(
          \frac{\mathbf{v}}{\sqrt{|z^{n}_{2}|}}
          + \mathbf{K}_{21} \widehat{\beta}
        \right).
    \end{align*}
    \item Construct the closed-form kernel-ridge estimator $\widehat{\theta}$ as
    \begin{align*}
      \widehat{\theta}(\cdot)
      = \frac{1}{\sqrt{|z^{n}_{2}|}}
        \sum_{j=1}^{|z^{n}_{2}|}
        \bigl[
          \widehat{\beta}_{j} \mathcal{K}(\cdot, a_{j})
          + \widehat{\gamma}_{j} \mathcal{K}(\cdot, b_{j})
        \bigr].
    \end{align*}
  \end{enumerate}
\end{algorithmic}
\end{algorithm}

\section{Excess Risk Guarantees and Worst-Case Efficiency Gains from Using Data Fusion}

\subsection{Overview}

This section summarizes and formalizes our theoretical findings regarding the impact of data fusion on estimator performance. We first derive oracle excess risk bounds for empirical risk minimizers, both in general settings and in kernel-ridge regression. These results demonstrate that, when nuisance functions are estimated with sufficient accuracy, the excess risk is dominated by the target estimation error. We then show that incorporating data fusion reduces the Lipschitz constant of the loss function, which in turn tightens the upper bound on the excess risk.  
Lastly, we establish that the high-probability worst-case upper bound on the risk under data fusion can fall strictly below the high-probability worst-case lower bound attainable without data fusion. Together, these results provide both predictive and worst-case efficiency guarantees for leveraging partially aligned data sources. \Cref{fig:maxrisk_line} summarizes these relationships between excess risk bounds and worst-case guarantees under data fusion and no fusion.

\begin{figure}[tb]
    \centering
    \begin{adjustbox}{margin=2cm 0cm 2cm 0cm,center}
        \begin{tikzpicture}
            \draw[thick,{Stealth[length=2.5mm]}-{Stealth[length=2.5mm]}] (0,0) -- (15,0);

            \fill[CB5blue] (3.5,0) circle (2.5pt);
            \fill[CB5red] (11.5,0) circle (2.5pt);
            \node[below, CB5blue] at (3.5,1.2) {\shortstack{Risk under \\[.25em]\Cref*{alg:rkhs_est}}};
            \node[below, CB5red] at (11.5,1.2) {\shortstack{Risk of\\[.25em]optimal estimator}};

            \draw[line width = 1.25pt, CB5blue] (5.9, 0.2) -- (6.02164, 0.2);
            \draw[line width = 1.25pt, CB5blue] (6, 0.2) -- (6, -0.2); 
            \draw[line width = 1.25pt, CB5blue] (5.9, -0.2) -- (6.02164, -0.2);

            \draw[line width = 1.25pt, CB5red] (8.47836, 0.2) -- (8.6, 0.2);
            \draw[line width = 1.25pt, CB5red] (8.5, 0.2) -- (8.5, -0.2); 
            \draw[line width = 1.25pt, CB5red] (8.47836, -0.2) -- (8.6, -0.2);
            

            
            \node[below, CB5blue] at (6, -0.5) {\shortstack{UB,\\ \Cref*{thm:rkhs_excess}}};
            \node[below, CB5red] at (8.5, -0.5) {\shortstack{LB,\\ \Cref*{thm:minimax_bound}}};
            
            \node[below, CB5blue] at (3.5, 2) {\bf With data fusion};
            \node[below, CB5red] at (11.5, 2) {\bf Without data fusion};
        \end{tikzpicture}
    \end{adjustbox}
\caption{Number line demonstrating approach for establishing data fusion necessarily improves worst-case performance. The key idea is to show that a \textbf{\color{CB5blue}high-probability worst-case upper bound (UB)} of the risk for kernel ridge regression with data fusion lies below a \textbf{\color{CB5red}high-probability worst-case lower bound (LB)} for any estimator without it.}
\label{fig:maxrisk_line}
\end{figure}
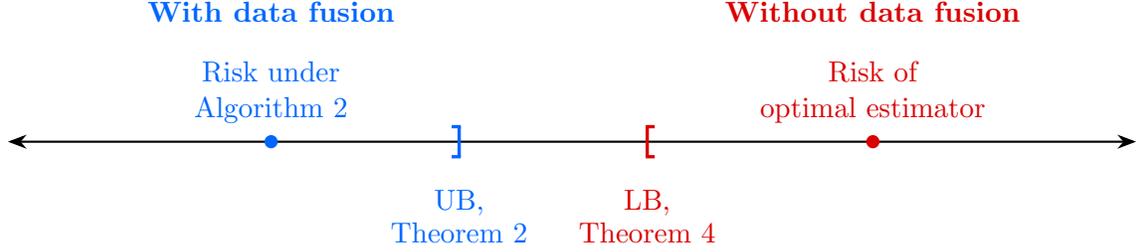

\subsection{Excess Risk Bound}
This subsection establishes oracle excess risk bounds for the estimators obtained from \Cref{alg:gen_est,alg:rkhs_est}. We show that, for a general function class $\Theta \subseteq L^{2}(\mu)$, the nuisance parameter estimation error affects the excess risk only at the fourth order. As a result, when the nuisance estimation error is sufficiently small, the oracle estimation error is dominated by the target estimation error.

Let $\theta^{*} \in \arg\min_{\theta \in \Theta} L_{P^{0}}(\theta, g_{0})$ denote a minimizer of the oracle risk. Throughout, we equip $\Theta$ with the $L^{2}(\mu)$ norm and, depending on the context, $\mathcal{G}$ with either the $L^{2}(Q^{0}_{X} \times \mu; \ell^{2})$ or $L^{4}(Q^{0}_{X} \times \mu; \ell^{2})$ norm.  
For $g \in \mathcal{G}$ and $1 \le p < \infty$, we define
\[
  \|g\|_{L^{p}(Q^{0}_{X} \times \mu; \ell^{2})}
  := \left( \E_{(X, A) \sim Q^{0}_{X} \times \mu}\!\big[\|g(X, A)\|_{2}^{p}\big] \right)^{1/p}.
\]
The nuisance estimation error of $\widehat{g}$ is then quantified as 
$\|\widehat{g} - g_{0}\|_{L^{2}(Q^{0}_{X} \times \mu; \ell^{2})}$ 
or 
$\|\widehat{g} - g_{0}\|_{L^{4}(Q^{0}_{X} \times \mu; \ell^{2})}$.

We next show that the oracle excess risk, $L_{P^{0}}(\widehat{\theta}, g_{0}) - L_{P^{0}}(\theta^{*}, g_{0})$, is bounded by the sum of the excess risk evaluated at $\widehat{g}$, $L_{P^{0}}(\widehat{\theta}, \widehat{g}) - L_{P^{0}}(\theta^{*}, \widehat{g})$, and a fourth-order term quantifying the nuisance estimation error, 
where $\widehat{\theta}$ and $\widehat{g}$ are the estimators produced by \Cref{alg:gen_est}.

\begin{lemma}[Oracle Excess Risk Bound in Terms of Target and Nuisance Errors]\label{lem:general}
Assume \Cref{cond:identifiability,cond:target_model,cond:nuisances} hold. If $\widehat{\theta}$ and $\widehat{g}$ are the target estimator and nuisance estimators from \Cref{alg:gen_est}, then
\begin{align*}
L_{P^{0}}(\widehat{\theta}, g_{0}) - L_{P^{0}}(\theta^{*}, g_{0}) \leq 2 \left( L_{P^{0}}(\widehat{\theta}, \widehat{g}) - L_{P^{0}}(\theta^{*}, \widehat{g}) \right) + 2M_{\lambda}^{2} \| \widehat{g} - g_{0} \|_{L^{4}(Q^{0}_{X} \times \mu; \ell^{2})}^{4}.
\end{align*}
\end{lemma}

The proof is given in the appendix. It relies on the Neyman orthogonality of the loss, as well as second-order smoothness, strong convexity, and higher-order smoothness conditions, as stated in \Cref{lem:neyman}. This lemma shows that the oracle excess risk is bounded by the sum of two components: the target estimation error and the fourth power of the nuisance estimation error. Thus, when the nuisance estimation error is smaller than the one-fourth power of the target estimation error, the excess risk bound is dominated by the target estimation error alone.

We now analyze the oracle excess risk for the estimators produced by \Cref{alg:gen_est,alg:rkhs_est}. The difficulty of estimating the CDRF depends on the complexity of the function class $\Theta$, which we quantify via the local Rademacher complexity. For a measure $\mathcal{D}$, defined as
\begin{align*}
\mathcal{R}_{n, \mathcal{D}}(\Theta, \delta) := \E_{\epsilon_{1:n}, a_{1:n} \sim \mathcal{D}} \left[ \sup_{\theta \in \Theta: \|\theta\|_{L^{2}(\mu)} \leq \delta} \left| \frac{1}{n} \sum_{i=1}^{n} \epsilon_{i} \theta(a_{i}) \right| \right],
\end{align*}
where $\epsilon_1, \ldots, \epsilon_n$ are independent Rademacher random variables. 

For any $\theta' \in \Theta$, define the star hull of $\Theta$ around $\theta'$ as
\begin{align*}
\text{star}(\Theta, \theta') := \left\{ t\theta + (1 - t)\theta' \mid \theta \in \Theta,\, t \in [0, 1] \right\}.
\end{align*}

We now present the oracle excess risk bound for the estimator obtained from \Cref{alg:gen_est}. To control deviations of the empirical loss from its expectation, we apply Talagrand’s concentration inequality, leveraging the fact that the loss is Lipschitz in its first argument, with a Lipschitz constant that depends on the true nuisance parameters with high probability. 
We define
\begin{align*}
  B_{P^{0}}(\delta)
  := 4(1 + \delta)^{2}\!\bigl(1 + \xi_{0} + C_{\sigma,\delta}\, \eta_{0}\|w_{0}\|_{\infty}\bigr),  
\end{align*}
where
\begin{align*}
  C_{\sigma,\delta}
  := 1 + \max\!\left\{\sigma\sqrt{2\log(8/\delta)},\, 2L\log(8/\delta)\right\},  
\end{align*}
and $\|\cdot\|_{\infty}$ denotes the essential supremum with respect to $Q^{0}_{X} \times \mu$.  
Throughout, we write $A \lesssim B$ to mean that $A \le C B$ for some universal constant $C > 0$.

\begin{theorem}[Excess Risk Bound for \Cref{alg:gen_est}]\label{thm:excess_risk}
Suppose \Cref{cond:identifiability,cond:target_model,cond:nuisances} hold. Fix $\delta \in (0,1)$, and let $(\widehat{\theta}, \widehat{g})$ be as in \Cref{alg:gen_est}, with 
$\|\widehat{g} - g_{0}\|_{L^{4}(Q^{0}_{X} \times \mu; \ell^{2})} = o_{p}(1)$ and $\|\widehat{w} - w_{0}\|_{L^{\infty}(Q^{0}_{X} \times \mu; \ell^{2})} = o_{p}(1)$.
Suppose $\delta_n^{2} = \Omega\!\left(n^{-1}\log\log n\right)$, and let $\delta_n$ be any solution to the inequality
\[
\max\Big\{ \mathcal{R}_{n, Q^{0}}\big(\mathrm{star}(\Theta - \theta^{*}, 0), r\big),\ 
           \mathcal{R}_{n, \mu}\big(\mathrm{star}(\Theta - \theta^{*}, 0), r\big) \Big\} \le r^{2}.
\]
Then there exists $N(\delta) \in \mathbb{N}$ such that, for all $n \ge N(\delta)$, the following holds with probability at least $1 - \delta/2$: for any realization $z := (x,a,b,y,s)$ of $Z \sim \widetilde{P}^{0}$ and any $\theta_{1}, \theta_{2} \in \Theta$,
\begin{align}\label{eq:lipschitz_bound}
\big| \ell_{P^{0}}(\theta_{1}, \widehat{g}; z) - \ell_{P^{0}}(\theta_{2}, \widehat{g}; z) \big|
\le B_{P^{0}}(\delta)\, \big\| \big( \theta_{1}(b) - \theta_{2}(b),\ \theta_{1}(a) - \theta_{2}(a) \big) \big\|_{2}.
\end{align}
Furthermore, with probability at least $1 - \delta$, the oracle excess risk satisfies
\[
L_{P^{0}}(\widehat{\theta}, g_{0}) - L_{P^{0}}(\theta^{*}, g_{0})
\;\lesssim\; B_{P^{0}}(\delta)^{2} \left( \delta_{n}^{2} + \frac{\log(1/\delta)}{n} \right)
\;+\; 2 M_{\lambda}^{2} \big\| \widehat{g} - g_{0} \big\|_{L^{4}(Q^{0}_{X} \times \mu; \ell^{2})}^{4}.
\]
\end{theorem}

Next, we consider the excess risk bound for kernel ridge regression.  
We follow the setup introduced in \citet{nie_quasi-oracle_2021, zhang_optimality_2023}. Assume a separable RKHS $\mathcal{H}$ on the compact set $\mathcal{A}$ with a strictly positive-definite, continuous kernel $\mathcal{K}\!:\! \mathcal{A} \times \mathcal{A} \to \mathbb{R}$.  
Define the integral operator $T_{\mathcal{K}}\!:\! L^{2}(\mu) \to L^{2}(\mu)$ by
\begin{align*}
  (T_{\mathcal{K}}\theta)(a)
  := \int_{\mathcal{A}} \mathcal{K}(a, t)\, \theta(t)\, d\mu(t).
\end{align*}

This operator is self-adjoint, positive, and compact \citep{steinwart_mercers_2012}.
By the spectral theorem for self-adjoint compact operators and Mercer's decomposition \citep{cucker_mathematical_2001},
\begin{align*}
k(a, t) = \sum_{j=1}^{\infty} \sigma_{j} e_{j}(a) e_{j}(t), \hspace{3em}
T_{\mathcal{K}}(\cdot) = \sum_{j=1}^{\infty} \sigma_{j} \langle \cdot, e_{j} \rangle_{L^{2}(\mu)} e_{j},
\end{align*}
where $\{ \sigma_{j} \}_{j=1}^{\infty}$ is a non-increasing summable sequence of eigenvalues and $\{ e_{j} \}_{j=1}^{\infty}$ is a corresponding $L^{2}(\mu)$-orthonormal system of eigenfunctions. For $s \geq 0$, we denote the fractional integral operator $T_{\mathcal{K}}^{s}: L^{2}(\mu) \rightarrow L^{2}(\mu)$ and its image, called an $s$-power space $\mathcal{H}^{s}$ of $\mathcal{H}$, as
\begin{align*}
T_{\mathcal{K}}^{s}(\cdot) := \sum_{j=1}^{\infty} \sigma_{j}^{s} \langle \cdot, e_{j} \rangle_{L^{2}(\mu)} e_{j}, \hspace{1cm}
\mathcal{H}^{s} := \left\{ \sum_{j=1}^{\infty} \sigma_{j}^{s/2} a_{j} e_{j} \, : \, \sum_{j=1}^{\infty} a_{j}^{2} < \infty \right\}.
\end{align*}

It holds that $\mathcal{H}^{1} \subseteq \{[h] : h\in \mathcal{H}\}$, with $[h]$ a $\mu$-equivalence class of functions. Moreover, $0 < s_{1} < s_{2}$, $\mathcal{H}^{s_{2}} \subseteq \mathcal{H}^{s_{1}} \subseteq \mathcal{H}^{0}\subseteq L^2(\mu)$. Functions in $\mathcal{H}^{s}$ with larger $s$ are naturally interpreted as having higher smoothness with respect to the RKHS $\mathcal{H}$.

Moreover, we impose the following assumptions on the $\mathcal{H}$ and its associated kernel $\mathcal{K}$.

\begin{condition}[RKHS Conditions]\label{cond:rkhs}
The following conditions hold:
\begin{enumcond}
     \item \text{(Kernel function)}\label{cond:kernel}
    The reproducing kernel $\mathcal{K}:\mathcal{A}\times\mathcal{A} \to \mathbb{R}$ is continuous and strictly positive definite, so that the Gram matrix is invertible for distinct points.
    
    \item \text{(Eigenvalue decay at least polynomial)}\label{cond:evd}  
    There exist constants $p \in (0, 1)$ and $G < \infty$ such that the eigenvalues $\{\sigma_j\}_{j=1}^\infty$ of the kernel integral operator satisfy $\sigma_{j} \leq G j^{-1/p}$ for all $j\ge 1$. 
    
    \item \text{(Bounded eigenfunctions)}\label{cond:bounded_eigenf} The eigenfunctions $e_{j}$ corresponding to $\sigma_j$ are uniformly bounded by a constant $M_{e} < \infty$ not depending on $j$.
\end{enumcond}
\end{condition}

Condition~\ref{cond:kernel} ensures the strict positive-definiteness of the kernel, and hence the invertibility of the Gram matrix for distinct inputs, guaranteeing the closed-form estimator in \Cref{alg:rkhs_est}.  
By the spectral theorem (e.g., Theorem 2 of \citealp{cucker_mathematical_2001}),  
the eigenvalues of a compact, self-adjoint, positive operator $T_{\mathcal{K}}$ converge to zero.  
Condition~\ref{cond:evd} further requires that this decay is at least polynomial,  
$\sigma_{j} \lesssim j^{-1/p}$ with $p \in (0,1)$.  
The eigenvalue-decay rate reflects the smoothness of the kernel and of the associated RKHS (see, e.g., \citealp{rasmussen_gaussian_2005}).  
Finally, Condition~\ref{cond:bounded_eigenf} assumes that the eigenfunctions are uniformly bounded to control the local Rademacher complexity (cf. Lemma 5 in \citet{nie_quasi-oracle_2021}).

The following condition imposes both that $\mathcal{Q}$ is not too large-satisfying both the moment conditions from \Cref{cond:target_model} and an additional smoothness condition-and also is large enough so that estimating $\theta_{0}$ is challenging.
\begin{condition}[Characterization of $\mathcal{Q}$]\label{cond:src}  
    $\mathcal{Q}$ consists of all distributions satisfying \Cref{cond:target_model} and the following source condition: there exist constants $\alpha \in (0, 1/2)$ and $R < \infty$ such that, for any $Q^{0} \in \mathcal{Q}$, the CDRF $\theta_{0}$ satisfies
    \[
    \| T_{\mathcal{K}}^{\alpha} \theta_{0} \|_{\mathcal{H}} \leq R. 
    \]
\end{condition}

\Cref{cond:src} is referred to as a \textit{source condition} and requires that the true CDRF be sufficiently smooth, belonging to a $(1 - 2\alpha)$-power space where $0 < 1 - 2\alpha < 1$ \citep{fischer_sobolev_2020}.

Given a nuisance estimator $\widehat{g}$, the use of kernel ridge regression in \Cref{alg:rkhs_est} produces an estimator $\widehat{\theta}$. Since the penalty parameter $\lambda$ is chosen by cross-validation, it is random, and consequently the corresponding ball radius $c$ is also random. Nonetheless, as shown by \citet{van_der_laan_unified_2003,mitchell_general_2009}, the cross-validation selector is asymptotically equivalent to the oracle selector. At the population level, Algorithm~\ref{alg:rkhs_est} is equivalent to empirical risk minimization over the constrained RKHS ball $\mathcal{H}_{c}$ centered at the origin with radius $c \ge 1$, which depends only on $\lambda$. Combining these facts, the oracle constrained RKHS ball minimizer is asymptotically equivalent to the estimator from \Cref{alg:rkhs_est}. For clarity, we focus on this constrained formulation throughout the remainder of the paper.

\begin{theorem}[Excess Risk Bound for \Cref{alg:rkhs_est}]\label{thm:rkhs_excess}
Assume that \Cref{cond:identifiability,cond:target_model,cond:nuisances,cond:rkhs,cond:src} hold.  
Fix $\delta \in (0,1)$, and let $(\widehat{\theta}, \widehat{g})$ be as in \Cref{alg:rkhs_est}, such that 
$\|\widehat{g} - g_{0}\|_{L^{2}(Q^{0}_{X} \times \mu; \ell^{2})} = o_{p}(1)$ and 
$\|\widehat{w} - w_{0}\|_{L^{\infty}(Q^{0}_{X} \times \mu; \ell^{2})} = o_{p}(1)$.  
Then there exists $N(\delta) \in \mathbb{N}$ such that, for all $n \ge N(\delta)$, with probability at least $1 - \delta$,
\begin{align*}
  &L_{P^{0}}(\widehat{\theta}, g_{0}) - L_{P^{0}}(\theta^{*}, g_{0}) \\
  &\qquad\lesssim B_{P^{0}}(\delta)^{2}
    \!\left(
      c^{\frac{2p}{1+p}}(n\log(n)^{-2})^{-\frac{1}{1+p}}
      + \frac{\log(1/\delta)}{n}
    \right)
    + M_{\lambda}^{2}c^{\frac{2p}{1+p}}
      \|\widehat{g} - g_{0}\|_{L^{2}(Q^{0}_{X} \times \mu; \ell^{2})}^{\frac{4}{1+p}}.
\end{align*}
Moreover, suppose $\widehat{g}$ satisfies 
$\|\widehat{g} - g_{0}\|_{L^{2}(Q^{0}_{X} \times \mu; \ell^{2})}
   = o_{p}\!\big((n\log(n)^{-2})^{-1/4}\big)$,
and the regularization parameter is chosen as
\[
  c \propto
  \bigl(B_{P^{0}}(\delta)^{-2(1+p)}n\log(n)^{-2}\bigr)^{\frac{\alpha}{p + (1 - 2\alpha)}}.
\]
Then there exists $N'(\delta)\in\mathbb{N}$ such that, for all $n \ge N'(\delta)$, with probability at least $1 - \delta$,
\[
  L_{P^{0}}(\widehat{\theta}, g_{0}) = L_{P^{0}}(\widehat{\theta}, g_{0})
  - L_{P^{0}}(\theta_{0}, g_{0})
  \;\lesssim\;
  \bigl(
    B_{P^{0}}(\delta)^{-2(1+p)}n\log(n)^{-2}
  \bigr)^{-\frac{1 - 2\alpha}{p + (1 - 2\alpha)}}.
\]
\end{theorem}

\Cref{thm:rkhs_excess} is a special case of \Cref{thm:excess_risk} with $\Theta$ an RKHS ball $\mathcal{H}_c$. The proof derives a valid critical radius via Lemma~5 of \citet{nie_quasi-oracle_2021}, which bounds the local Rademacher complexity in terms of $c$, $\delta$, and $n$. It also requires a weaker condition on the nuisance rates ($L^{2}$ rather than $L^{4}$ in \Cref{thm:excess_risk}). For the second part, the approximation error is absorbed into the excess risk term, yielding an estimation error bound relative to the global minimizer $\theta_{0}$. Moreover, \Cref{thm:rkhs_excess} provides not only a high-probability upper bound on the risk but—under a stronger uniform rate for the nuisance estimators as in \Cref{thm:minimax_bound}—also yields a bound uniform over $\mathcal{P}$, obtained by taking the supremum over $P \in \mathcal{P}$.

\subsection{Efficiency Gain and Prediction Guarantee}

We now show that performing data fusion reduces the excess risk bound. Specifically, we compare the excess risk bounds under two settings: (i) using data fusion to leverage the full sample, and (ii) not using data fusion and instead only using target samples---that is, those from sources in $\mathcal{S}_{X} \cap \mathcal{S}_{Y}$.
When not using data fusion, the (oracle) Neyman-orthogonal loss reduces to:
\begin{align*}
\uline{\ell}_{P^0}(\theta, g_{0}; z) 
&:= \left( \theta(b) - \underline{\tau}_{0}(b) \right)^{2} 
+ 2\underline{\xi}_{0} \mathbf{1}(s \in \mathcal{S}_{X} \cap \mathcal{S}_{Y}) 
\left( \theta(b) - \underline{\tau}_{0}(b) \right) 
\left( \underline{\tau}_{0}(b) - \underline{m}_{0}(x, b) \right) \nonumber \\
&\quad + 2\underline{\eta}_{0} \underline{w}_{0}(x, a) \mathbf{1}(s \in \mathcal{S}_{X} \cap \mathcal{S}_{Y}) 
\left( \theta(a) - \underline{\tau}_{0}(b) \right) 
\left( \underline{m}_{0}(x, a) - y \right),
\end{align*}
where $z := (x,a,b,y,s)$ and the nuisance components are defined as
\begin{align*}
&\underline{\xi}_{0} := \frac{1}{P^{0}(S \in \mathcal{S}_{X} \cap \mathcal{S}_{Y})}, \quad 
\underline{\eta}_{0} := \frac{1}{P^{0}(S \in \mathcal{S}_{X} \cap \mathcal{S}_{Y})}, \quad
\underline{w}_{0}(x, a) := \frac{1}{p^{0}(a \mid x, S \in \mathcal{S}_{X} \cap \mathcal{S}_{Y})}\,, \\
&\underline{m}_{0}(x, a) := \E_{P^0} \left[ Y \mid X = x, A = a, S \in \mathcal{S}_{X} \cap \mathcal{S}_{Y} \right], \quad 
\underline{\tau}_{0}(a) := \E_{P^0} \left[ \underline{m}_{0}(X, a) \mid S \in \mathcal{S}_{X} \cap \mathcal{S}_{Y} \right],
\end{align*}
where $p^{0}(a \mid x, S \in \mathcal{S}_{X} \cap \mathcal{S}_{Y})$ denotes the conditional density of $A$ given $X$ and $S \in \mathcal{S}_{X} \cap \mathcal{S}_{Y}$ with respect to the known measure $\mu$ on $\mathcal{A}$. 

In this restricted setting, we analogously define the nuisance estimators 
$\widehat{\underline{\xi}}, \widehat{\underline{\eta}}, \widehat{\underline{w}}, \widehat{\underline{m}},$ and $\widehat{\underline{\tau}}$.  
Let $(\widehat{\theta}, \widehat{g})$ denote the estimators obtained from \Cref{alg:rkhs_est} with data fusion,  
and $(\underline{\theta}, \underline{g})$ the corresponding estimators without data fusion,  
where $\underline{g}\!:\! (a, x) \mapsto (\widehat{\underline{\xi}}, \widehat{\underline{\eta}}, \widehat{\underline{w}}(a, x), \widehat{\underline{m}}(a, x), \widehat{\underline{\tau}}(a))$.

For \Cref{thm:bound_fusion} below, we impose the same nuisance estimation and cross-validation conditions as in the second statement of \Cref{thm:rkhs_excess}.  
Specifically, suppose the nuisance estimators satisfy
\begin{align*}
  \|\widehat{g} - g_{0}\|_{L^{2}(Q^{0}_{X}\times\mu; \ell^{2})}
  &= o_{p}\!\big((n\log(n)^{-2})^{-1/4}\big),
  &\|\widehat{w} - w_{0}\|_{L^{\infty}(Q^{0}_{X}\times\mu; \ell^{2})}
  &= o_{p}(1), \\
  \|\underline{g} - g_{0}\|_{L^{2}(Q^{0}_{X}\times\mu; \ell^{2})}
  &= o_{p}\!\big((n\log(n)^{-2})^{-1/4}\big),
  &\|\underline{\widehat{w}} - w_{0}\|_{L^{\infty}(Q^{0}_{X}\times\mu; \ell^{2})}
  &= o_{p}(1),
\end{align*}
and the regularization parameter is chosen as
\[
  c \propto
  \Big(
    B_{P^{0}}(\delta)^{-2(1+p)} n\log(n)^{-2}
  \Big)^{\frac{\alpha}{p + (1 - 2\alpha)}}.
\]

\begin{theorem}[Data Fusion Improves the Excess Risk Upper Bound]\label{thm:bound_fusion}
Suppose the nuisance estimators and the regularization parameter satisfy the conditions listed directly above the theorem.  
Further assume that \Cref{cond:identifiability,cond:target_model,cond:nuisances,cond:rkhs,cond:src} hold,  
and that \Cref{cond:identifiability,cond:nuisances} also hold with fusion sets 
$\underline{\mathcal{S}}_{X} = \underline{\mathcal{S}}_{Y} = \mathcal{S}_{X} \cap \mathcal{S}_{Y}$.  

Fix $\delta \in (0,1)$.  
Then there exists $N(\delta)$ such that, for all $n \ge N(\delta)$, the $(1-\delta)$-probability excess risk bounds under data fusion and without data fusion are both of order 
$(n\log(n)^{-2})^{-\frac{1 - 2\alpha}{p + (1 - 2\alpha)}}$.  
Moreover, they satisfy
\[
  \frac{\textnormal{Excess-risk bound under data fusion}}
       {\textnormal{Excess-risk bound without data fusion}}
  =
  \left(
    \frac{1 + \xi_{0} + C_{\sigma,\delta}\,\eta_{0}\|w_{0}\|_{\infty}}
         {1 + \underline{\xi}_{0} + C_{\sigma,\delta}\,\underline{\eta}_{0}\|\underline{w}_{0}\|_{\infty}}
  \right)^{\!\frac{2(1+p)(1 - 2\alpha)}{p + (1 - 2\alpha)}}
  \le 1,
\]
where $C_{\sigma,\delta} := 1 + \max\!\left\{  \sigma\sqrt{2\log(8/\delta)},\, 2L\log(8/\delta) \right\}$.
The inequality is strict if 
\[
  P^{0}\!\big(S \in \mathcal{S}_{X}\setminus\mathcal{S}_{Y}\big)
  + \!\operatorname*{ess\,inf}_{(X,A,Y)\sim P^{0}}
      P^{0}\!\big(S \in \mathcal{S}_{Y}\setminus\mathcal{S}_{X}
      \mid A,X,S\in\mathcal{S}_{Y}\big)
  > 0.
\]
\end{theorem}
The condition for the strict inequality is slightly stronger than merely requiring the existence of some partially aligned sources; that is, $P^{0}\!\big(S \in \mathcal{S}_{X} \triangle \mathcal{S}_{Y}\big) = P^{0}\!\big(S \in \mathcal{S}_{X} \setminus \mathcal{S}_{Y}\big) + P^{0}\!\big(S \in \mathcal{S}_{Y} \setminus \mathcal{S}_{X}\big) > 0$. The key to establishing the above result is showing that using data fusion necessarily reduces the loss’s Lipschitz constant. We have shown that the excess-risk bound decreases under data fusion by applying \Cref{thm:rkhs_excess}. 

Beyond excess risk upper bounds, we show that using data fusion necessarily improves performance in a minimax sense, under appropriate conditions. Concretely, we show a high-probability worst-case upper bound of the risk under data fusion is smaller than a worst-case lower bound of the risk without data fusion, provided there are sufficiently more data that are only partially aligned with the target distribution than are perfectly aligned. 
To formalize this result, we impose an additional assumption on the eigenvalue decay of the kernel integral operator, requiring it to be bounded both above and below by a polynomial rate.

\begin{condition}[Eigenvalue Decay at Exactly Polynomial Rate] \label{cond:minimax_conditions}
There exists $p \in (0,1)$ and constants $G_{1}, G_{2} > 0$ such that, for all $j \ge 1$,
\[
  G_{1} j^{-1/p} \le \sigma_{j} \le G_{2} j^{-1/p}.
\]
\end{condition}

In the following lemma, we provide minimax lower bounds for an estimator $\underline{\theta}$ that does not use data fusion. Formally, any such $\underline{\theta}$ takes as input a sample $\{(X_{i}, A_{i}, Y_{i}, S_{i})\}_{i=1}^{n}$ and is measurable with respect to the $\sigma$-algebra generated by $\{(X_{i}, A_{i}, Y_{i})\}_{i : S_{i} \in \mathcal{S}_{X} \cap \mathcal{S}_{Y}}$.

\begin{lemma}[Minimax Lower Bound without Data Fusion]\label{lem:minimax_bound}
Assume that \Cref{cond:identifiability,cond:target_model,cond:nuisances,cond:rkhs,cond:src,cond:minimax_conditions} hold. Suppose $1-2\alpha \ge  p$. Fix an estimator $\underline{\theta}$ that does not use data fusion and $\delta\in (0,1)$. Let $N_\delta<\infty$ be as defined in \eqref{eq:Ndelta} in the appendix. There exists a constant $c_{\mathcal{Q}} > 0$ depending on $\mathcal{Q}$ only such that, for all $(n,n_{XY})$ such that $n\ge n_{XY}\ge N_\delta$, there exists $P\in\mathcal{P}$ such that $P(S \in \mathcal{S}_X \cap \mathcal{S}_Y)=n_{XY}/n$ and, with probability at least $1 - 2\delta$,
\begin{align}
    L_P(\underline{\theta}, g_0) \geq c_{\mathcal{Q}} \delta 
    \left(1+\sqrt{3\log(1/\delta)/n_{XY}}\right)^{-\frac{1 - 2\alpha}{p + (1 - 2\alpha)}} 
    n_{XY}^{-\frac{1 - 2\alpha}{p + (1 - 2\alpha)}}. \label{eq:minimaxLB}
    \end{align}
\end{lemma}

In the above lemma, $n_{XY}$ denotes the expected number of perfectly aligned samples among $n$ i.i.d.\ observations drawn from $P$. The proof proceeds by lower bounding a minimax risk over $\mathcal{P}$ by a worst-case Bayes risk over $\mathcal{Q}$ via a maximin argument. We then construct a finite family $\{Q_0,\ldots,Q_M\}\subset\mathcal{Q}$ whose CDRFs are well separated, while the corresponding product measures are close in Kullback--Leibler divergence; applying Fano’s inequality to this packing yields the desired bound.

We combine the minimax lower bound for the no–fusion estimator from \Cref{lem:minimax_bound} with the uniform excess risk upper bound for the data-fusion estimator from \Cref{thm:rkhs_excess}. The next theorem formalizes the cases when data fusion strictly improves the worst-case risk. 

Let $(\widehat{\theta},\widehat{g})$ be the estimators in \Cref{alg:rkhs_est} and $\underline{\theta}$ be any estimator that does not use data fusion.
In the following result, we assume the nuisance estimator $\widehat{g}$ satisfies the following uniform condition (stronger than \Cref{thm:rkhs_excess}): 
for any $\varepsilon>0$, there exists $M_g<\infty$ such that
\begin{align}
\limsup_{n\to\infty}\ \sup_{P\in\mathcal P}\ 
P^{n}\!\left(
  \|\widehat{g}-g_0\|_{L^{2}(Q_X^0\times \mu;\,\ell^{2})} 
  > M_g\,(n\log(n)^{-2})^{-1/4}
\right) \le \varepsilon. \label{eq:ghatUniform}
\end{align}
We further suppose the regularization parameter is chosen as
\begin{align}
c \propto (n \log(n)^{-2})^{\frac{\alpha}{p + (1 - 2\alpha)}}. \label{eq:regParam}
\end{align}

\begin{theorem}[Sufficient Conditions for Data Fusion to Improve Worst-Case Performance] \label{thm:minimax_bound}
Suppose $(\widehat{\theta},\widehat{g})$ satisfy the conditions above, $c$ satisfies \eqref{eq:regParam}, \Cref{cond:identifiability,cond:target_model,cond:nuisances,cond:rkhs,cond:src,cond:minimax_conditions} hold, and $1-2\alpha \ge  p$. Fix $\delta \in (0,1/2]$ and $h>0$. Suppose the \textbf{\color{CB5red}non-data fusion sample size} is sufficiently large and the \textbf{\color{CB5blue}data fusion sample size} is sufficiently larger, in that, for $N_\delta$ as defined in \eqref{eq:large_n_minimax},
\[
N_\delta\le {\color{CB5red}\bm{n_{XY}}}\le \tfrac{n}{(\log n)^{2+h}}\le {\color{CB5blue}\bm{n}}.
\]
Then, for $\mathcal{P}_{XY}:=\{P\in\mathcal{P} : P\!\left(S \in \mathcal{S}_{X} \cap \mathcal{S}_{Y}\right)=n_{XY}/n\}$ and $r(\delta,n):= c_{\mathcal{Q}} \delta n^{-\frac{1-2\alpha}{p+(1-2\alpha)}}/4$,
\begin{align}
    &\sup_{P\in\mathcal{P}_{XY}} P^n\!\Bigl\{ L_P(\underline{\theta}, g_0) \ge r(\delta,n)\Bigr\}- \sup_{P\in\mathcal{P}_{XY}} P^n\!\Bigl\{ L_P(\widehat{\theta}, g_0) > r(\delta,n)/\log(n)^\frac{h(1-2\alpha)}{2[p+(1-2\alpha)]} \Bigr\}
    \;\ge\; 1 - 2\delta. \label{eq:minimaxDominance}
\end{align}
\end{theorem}

To our knowledge, this is the first work to formally demonstrate that data fusion improves performance in a statistical learning problem. Notably, this improvement holds regardless of the estimation algorithm used in the non-fusion setting.  
To establish this result, we use that the first term on the left-hand side is approximately one by \Cref{lem:minimax_bound},  
and that, by taking the supremum of the excess-risk bound in \Cref{thm:rkhs_excess}, the second term is approximately zero.

The suprema in the theorem may be attained at different distributions, which may appear to complicate interpretation.  
However, this also makes it possible to exhibit a particular distribution under which using data fusion improves performance.  
Concretely, letting $P_{\star}$ denote a (near) maximizer of the first supremum in \eqref{eq:minimaxDominance}, it follows that
\[
  P_{\star}^{n}\!\left\{
    \frac{L_{P_{\star}}(\widehat{\theta}, g_{0})}
         {L_{P_{\star}}(\underline{\theta}, g_{0})}
    \le \log(n)^{-\frac{h(1 - 2\alpha)}{2[p + (1 - 2\alpha)]}}
  \right\}
  \ge 1 - 3\delta.
\]
In other words, for sufficiently large $n$, there exists a distribution under which using data fusion substantially outperforms not using data fusion.

\section{Numerical Experiments}\label{sec:simulations}
We evaluate the empirical performance of our proposed data fusion kernel ridge CDRF estimator. Our goal is to assess whether incorporating additional partially aligning datasets via data fusion improves predictive performance, even when the reference measure is different than the data-generating measure or true CDRF does not belong to the chosen RKHS. Throughout we focus on the RKHS with a Laplace kernel and bandwidth chosen according to the median heuristic \citep{garreau_large_2018}.

We conducted Monte Carlo simulations with 500 runs under various configurations. The data are generated as follows: $S_X \sim \text{Bernoulli}(0.5)$, $S_Y \sim \text{Bernoulli}(0.5)$, $S_X \perp\!\!\!\perp S_Y$, $X \mid S_X = 1 \sim \mathcal{N}( \mu_0 = \tfrac{1}{3}(1,1,1)^\top, \; \Sigma_0 = 0.09 \cdot I_3 )$,
\begin{align*}
    &X \mid S_X = 0 \sim \mathcal{N}\left( \mu_1 = \tfrac{1}{6}(1,1,1)^\top, \;
    \Sigma_1 = 
    \begin{bmatrix}
    0.25 & 0.1 & 0.1 \\
    0.1 & 0.25 & 0.1 \\
    0.1 & 0.1 & 0.25
    \end{bmatrix}
    \right), 
\end{align*}
$A \mid X \sim \text{Beta}\left(1 + 1/[1 + \exp(\langle X, 1 \rangle)], \;1 + 1/[1 + \exp(\langle X, 1 \rangle)] \right)$, $Y \mid X, A, S_Y = 1 \sim \mathcal{N}( \theta_0(A) \cdot \langle X, 1 \rangle, \; 0.1^2 )$, and $Y \mid X, A, S_Y = 0 \sim \mathcal{N}( \widetilde{\theta}_0(A) \cdot \langle X, 1 \rangle, \; 0.1^2 )$, where $\langle X, 1 \rangle$ denotes the row sum of covariates $X$, and $\theta_0(a)$ is the true CDRF.

We consider three distinct functional forms for $\theta_0(a)$ across simulation sets, along with corresponding misspecified counterparts $\tilde{\theta}_0(a)$ used to generate conditional outcomes. The first setting has a Gaussian-shaped CDRFs, with $\theta_0(a) = \phi(a; 0.5, 0.25^2) - 1$ and $\widetilde{\theta}_0(a) = 0.5 \left( \phi(a; 1, 0.25^2) - 1 \right)$ for $\phi(a; \mu, \sigma^2)$ the Gaussian density with mean $\mu$ and variance $\sigma^2$. The second has CDRFs defined by trigonometric functions, with $\theta_0(a) = 5 \sin(3a) + 3 \cos(10a)$ and $\widetilde{\theta}_0(a) = 0.5 [\cos(3a) + \sin(10a)]$. The third has discontinuous CDRFs defined as
\begin{align*}
\theta_0(a) &=
\begin{cases}
\left( a^{0.5} + 0.1 \right), & \text{if } a < 0.5, \\
0.5 \left( a^4 + 1 \right), & \text{if } a \geq 0.5,
\end{cases}\hspace{4em} 
\widetilde{\theta}_0(a) =
\begin{cases}
\left( \log(1 + a) \right)^{0.5}, & \text{if } a < 0.3, \\
0.1 \left( \cos(3a)+1 \right), & \text{if } a \geq 0.3.
\end{cases}
\end{align*}
The discontinuous CDRFs are of interest since---unlike the other two CDRFs considered \citep[][Corollary 10.48]{wendland_scattered_2004}---they are not contained in the RKHS induced by the Laplacian kernel.

For all three true and misspecified CDRFs, we evaluated performance under three different reference measures for risk assessment. The first is the uniform distribution over $[0,1]$, which assigns equal importance to all parts of the dose-response curve. The second is $\text{Beta}(5,5)$, a bell-shaped distribution emphasizing the center of the range, thereby making estimation accuracy near $a = 0.5$ more important. The third is $\text{Beta}(0.5, 0.5)$, a U-shaped distribution that places more weight on boundary regions, which may be useful when treatment effects at the boundaries are of particular interest.

For each scenario described above and a range of sample sizes $n$, we computed the median empirical risk of the estimated CDRFs with and without data fusion. Following the procedure in \Cref{alg:rkhs_est}, we randomly split the dataset into two halves. The first half was used to estimate nuisance parameters, while the second half was used for target estimation and risk evaluation.

Specifically, we began by estimating the source probabilities $P^{0}(S \in \mathcal{S}_{X})$ and $P^{0}(S \in \mathcal{S}_{Y})$ using empirical proportions. We then estimated the density ratio $(x, a) \mapsto w_{0}(x, a)$ via the `densratio` package in R, implementing the unconstrained Least-Squares Importance Fitting (uLSIF) method \citep{kanamori_least-squares_2009}. For the conditional outcome regression $(x, a) \mapsto m_{0}(x, a)$, we used a SuperLearner package \citep{van_der_laan_super_2007}, combining mean regression, random forests, XGBoost, and elastic net. To choose the regularization parameter $\lambda$ for kernel ridge regression, we performed cross-validation over a grid $\{0.0001, 0.0051, \ldots, 0.0301\}$, selecting the value minimizing estimated risk, as detailed in \Cref{alg:cross_val}.

Given the estimated nuisance components and the selected $\lambda$, we computed the closed-form estimator of the CDRF as in \Cref{alg:rkhs_est}. We further evaluated performance by computing the empirical risk, defined as the mean squared error between the estimated and true CDRFs, where the evaluation points $a$ were sampled from the designated reference measure in each setting. We also report the percentage reduction in median risk achieved through data fusion. 

\begin{table}[tb]\singlespacing\small
\centering
\caption{Median risks and percent reduction from using data fusion across sample sizes.}\label{tab:median-risk}
\begin{subtable}{\textwidth}\centering
\caption{Gaussian CDRF, with risks scaled by \( \times 10^{-3} \)}\label{tab:median-risk-gaussian}
\begin{tabular}{c ccc | ccc | ccc}
\toprule
\multirow{2}{*}{\makecell[c]{Sample \\ Size}} 
  & \multicolumn{3}{c|}{Uniform(0, 1)} 
  & \multicolumn{3}{c|}{Beta(5, 5)} 
  & \multicolumn{3}{c}{Beta(0.5, 0.5)} \\
\cmidrule(lr){2-4} \cmidrule(lr){5-7} \cmidrule(lr){8-10}
  & Fusion & No Fusion & (\%) 
  & Fusion & No Fusion & (\%) 
  & Fusion & No Fusion & (\%) \\
\midrule
100  & 90.2 & 145.3 & 38 & 259.8 & 265.3 & 2  & 169.4 & 251.7 & 33 \\
200  & 48.9 &  82.8 & 41 & 128.9 & 189.4 & 32 &  98.2 & 158.0 & 38 \\
400  & 16.9 &  37.6 & 55 &  23.9 &  81.3 & 71 &  30.1 &  82.4 & 64 \\
800  &  6.7 &  10.3 & 35 &   7.4 &  15.4 & 52 &  10.1 &  17.9 & 44 \\
1600 &  3.1 &   4.4 & 30 &   3.6 &   5.3 & 31 &   4.7 &   6.6 & 30 \\
3200 &  1.8 &   2.3 & 19 &   2.4 &   3.2 & 25 &   2.7 &   3.4 & 21 \\
\bottomrule
\end{tabular}\vspace{1em}
\end{subtable}
\begin{subtable}{\textwidth}\centering
    \caption{Trigonometric CDRF, with risks scaled by by \( \times 10^{-2} \)}\label{tab:median-risk-trigonometric}
    \begin{tabular}{c ccc | ccc | ccc}
    \toprule
    \multirow{2}{*}{\makecell[c]{Sample \\ Size}} 
      & \multicolumn{3}{c|}{Uniform(0, 1)} 
      & \multicolumn{3}{c|}{Beta(5, 5)} 
      & \multicolumn{3}{c}{Beta(0.5, 0.5)} \\
    \cmidrule(lr){2-4} \cmidrule(lr){5-7} \cmidrule(lr){8-10}
      & Fusion & No Fusion & (\%) 
      & Fusion & No Fusion & (\%) 
      & Fusion & No Fusion & (\%) \\
    \midrule
    100  & 367.3 & 549.4 & 33 & 804.1 & 907.6 & 11 & 519.6 & 855.2 & 39 \\
    200  & 213.6 & 306.2 & 30 & 439.1 & 689.0 & 36 & 334.7 & 499.2 & 33 \\
    400  &  78.8 & 149.2 & 47 & 164.5 & 350.5 & 53 & 124.7 & 294.4 & 58 \\
    800  &  36.4 &  51.2 & 29 &  51.1 &  61.7 & 17 &  39.5 &  70.7 & 44 \\
    1600 &  15.8 &  23.1 & 32 &  21.5 &  29.0 & 26 &  15.3 &  26.6 & 43 \\
    3200 &   8.3 &  12.6 & 34 &  11.4 &  17.4 & 34 &   7.5 &  11.6 & 35 \\
    \bottomrule
    \end{tabular}
\end{subtable}\vspace{1em}
\begin{subtable}{\textwidth}\centering
\caption{Discontinuous CDRF, with risks scaled by \( \times 10^{-3} \)}\label{tab:median-risk-discontinuous}
\begin{tabular}{c ccc | ccc | ccc}
\toprule
\multirow{2}{*}{\makecell[c]{Sample \\ Size}} 
  & \multicolumn{3}{c|}{Uniform(0, 1)} 
  & \multicolumn{3}{c|}{Beta(5, 5)} 
  & \multicolumn{3}{c}{Beta(0.5, 0.5)} \\
\cmidrule(lr){2-4} \cmidrule(lr){5-7} \cmidrule(lr){8-10}
  & Fusion & No Fusion & (\%) 
  & Fusion & No Fusion & (\%) 
  & Fusion & No Fusion & (\%) \\
\midrule
100  & 143.7 & 198.8 & 28 & 190.7 & 214.4 & 11 & 208.4 & 275.2 & 24 \\
200  &  63.4 & 111.9 & 43 &  80.3 & 136.2 & 41 &  95.5 & 169.2 & 44 \\
400  &  31.6 &  45.8 & 31 &  32.5 &  56.8 & 43 &  49.6 &  67.8 & 27 \\
800  &  19.5 &  26.1 & 26 &  19.7 &  25.1 & 21 &  29.8 &  32.4 &  8 \\
1600 &  10.8 &  16.1 & 33 &  13.3 &  15.8 & 16 &  17.5 &  19.2 &  8 \\
3200 &   6.7 &   9.2 & 27 &   8.5 &  10.0 & 15 &   9.9 &  12.6 & 21 \\
\bottomrule
\end{tabular}
\end{subtable}
\end{table}

The results are summarized in \Cref{tab:median-risk}.  Across all settings, incorporating data fusion improves performance. This includes the discontinuous CDRF case in \Cref{tab:median-risk-discontinuous}, where the true CDRF does not lie in the RKHS of the estimator. Nonetheless, the proposed method still achieved good performance even at moderate sample sizes, demonstrating robustness to model misspecification. The tables also reveal robustness to the choice of reference measure. Importantly, our method benefits from data fusion even in cases where the reference measure differs from the data-generating distribution.

\section{Discussion}

Our data fusion framework can also be applied to the estimation of non-pathwise differentiable function-valued parameters such as CDRFs. As an extension, it can be used to estimate any smooth summary functional of such parameters—even when the associated risk is not of the mean integrated squared error type. This broadens the applicability of data fusion beyond traditional risk measures.

By employing a stochastic approximation, we successfully construct a Neyman-orthogonal loss function for such risks that avoids taking expectations over the parameter’s indexing variable. The resulting orthogonal loss is almost same as the loss derived directly from the nonparametric efficient influence function, up to an $O_p(1)$ term. While similar constructions may be possible for other parameters, a general theoretical discussion remains to be developed.

Although CDRF estimation under data fusion shows both strong theoretical guarantees and favorable empirical performance, several caveats remain. First, the estimation algorithms assume correct knowledge of which part of each data source aligns with the target distribution. Second, the methods are tailored for prediction, not for statistical inference. Finally, while the regularization parameter $\lambda$ is data-dependent, the radius $c$ of the RKHS ball is fixed. In practice, kernel ridge regression behaves similarly to RKHS-ball-constrained optimization when the sample size is sufficiently large.

There are several promising directions for future work. 
One is to address federated learning settings where individual-level data from each clinical trial are private and cannot be shared, but summary statistics can be exchanged \citep{han_federated_2025}. 
Another is to extend statistical inference to non-pathwise differentiable parameters under data fusion \citep{luedtke_one-step_2024}. 
Lastly, it would be interesting to explore the use of neural networks or other function classes for estimating CDRFs, as in the orthogonal statistical learning framework \citep{foster_orthogonal_2023}.

\section*{Acknowledgements}
This work was supported by the National Science Foundation and National Institutes of Health under award numbers DMS-2210216 and DP2-LM013340.

\section*{Supplementary Materials}
R code for the CDRF estimation function \Cref{alg:rkhs_est} and simulation codes for \Cref{sec:simulations} are available at
\url{https://github.com/imjaewon07/CDRF_data_fusion}.

\printbibliography

\appendix

\setcounter{equation}{0}
\renewcommand{\theequation}{S\arabic{equation}}
\setcounter{theorem}{0}
\setcounter{figure}{0}
\setcounter{table}{0}
\setcounter{lemma}{0}
\setcounter{corollary}{0}
\renewcommand{\thetheorem}{S\arabic{theorem}}
\renewcommand{\thecorollary}{S\arabic{corollary}}
\renewcommand{\thelemma}{S\arabic{lemma}}
\renewcommand{\thefigure}{S\arabic{figure}}
\renewcommand{\thetable}{S\arabic{table}}
\renewcommand{\thealgorithm}{S\arabic{algorithm}}

\makeatletter
\renewcommand{\theHequation}{S\arabic{equation}}
\renewcommand{\theHtheorem}{S\arabic{theorem}}
\renewcommand{\theHlemma}{S\arabic{lemma}}
\renewcommand{\theHcorollary}{S\arabic{corollary}}
\renewcommand{\theHfigure}{S\arabic{figure}}
\renewcommand{\theHtable}{S\arabic{table}}
\makeatother

\section*{\LARGE Appendices}

\DoToC

\section{Cross Validation Algorithm in Kernel Ridge Regression}
\begin{algorithm}[tb]
\caption{Cross validation in kernel ridge regression}
\label{alg:cross_val}
\begin{algorithmic}[1]
    \State \textbf{Input}: The sample $z_{n}$, number of folds $K$, values of $\lambda$.
    \State \textbf{Cross validation}: 
    \begin{enumerate}
        \item Partition the sample $z^{n}$ into training set $z_{k}^{n}$ and test set $z_{-k}^{n}$ for $k = 1, \ldots, K$.
        \item For each value of $\lambda$, do the following steps.
        \item For each $k$, construct nuisance estimators only using $z_{k}^{n}$ as in \Cref{alg:gen_est}.
        \item With the same $k$, estimate the following parameters by plugging in test set $(z_{-k}^{n})$ into the nuisance parameters in the previous step.
        \begin{align*}
        &u_{i} = \mathbf{1}(s_{i} \in \mathcal{S}_{Y}) \; \widehat{\eta} \; \widehat{w}(a_{i}, x_{i}) \left( \widehat{m}(a_{i}, x_{i}) - y_{i} \right),\;  \\
        &v_{i} = \mathbf{1}(s_{i} \in \mathcal{S}_{X}) \; \widehat{\xi} \left(\E_{X}\left[ \widehat{m}(b_{i}, x_{i}) \right] - \widehat{m}(b_{i}, x_{i})\right) - \E_{X}\left[\widehat{m}(b_{i}, x_{i})\right] \\
        &\widehat{\beta} = -\frac{\mathbf{u}}{\lambda_{n}\sqrt{|z_{k}^{n}|}}, \; \widehat{\gamma} = -\left( \mathbf{K}_{22} + \lambda_{n} \mathbf{I}_{|z_{k}^{n}|} \right)^{-1} \left( \frac{\mathbf{v}}{\sqrt{|z_{k}^{n}|}} + \mathbf{K}_{21} \widehat{\beta} \right) \\
        &\widehat{\theta}(\cdot) = \frac{1}{\sqrt{|z_{k}^{n}|}} \sum_{j=1}^{|z_{k}^{n}|} \left[ \widehat{\beta}_{j}\mathcal{K}(\cdot, a_{j}) + \widehat{\gamma}_{j} \mathcal{K}(\cdot, b_{j}) \right].
        \end{align*}
        \item Calculate the risk for the specific $k$ as follows.
        \begin{align*}
        \widehat{R}_{k}(\widehat{\theta}) = \frac{1}{|z_{-k}^{n}|}\sum_{i=1}^{|z_{-k}^{n}|} \left[ \widehat{\theta}(b_{i})^{2} + 2v_{i}\widehat{\theta}(b_{i}) + 2u_{i}\widehat{\theta}(a_{i}) \right] + \lambda_{n} \| \widehat{\theta} \|_{\mathcal{H}}.
        \end{align*}
        \item Calculate the cross-validated risk by taking average over $\widehat{R}_{k}$ over $k = 1, \ldots, K$.
        \item Going back to step 2, choose the $\lambda_{n}$ showing the smallest cross-validated risk.     
    \end{enumerate}
\end{algorithmic}
\end{algorithm}

\section{Proofs}
\subsection{Derivation of Neyman-Orthogonal Loss and a Stochastic Approximation Thereof (\Cref*{lem:onestep_loss})}
We now derive a Neyman-orthogonal loss for CDRFs. We proceed in several steps. First, we provide the form of the nonparametric canonical gradient of the risk at $Q^{0}$ with respect to the tangent space of the target distributions. Next, we project it onto the corresponding tangent space of $P^{0}$. Finally, we construct a one-step estimator of the risk, along with an approximate version that facilitates its minimization. We provide the nonparametric canonical gradients and construction of one-step estimators under two scenarios: (i) when data fusion is fully utilized, and (ii) when data fusion is not used and estimation is based solely on samples in the intersection $S \in \mathcal{S}_{X} \cap \mathcal{S}_{Y}$.

If $\mathcal{Q}$ is not locally nonparametric, standard calculations or a chain rule \citep{kennedy_semiparametric_2023,luedtke_simplifying_2025} show $Q\mapsto \mathcal{L}_{Q}(\theta)$ at $Q^0$ is pathwise differentiable with canonical gradient
\begin{align*}
D^{*}_{Q^{0}}(\theta): (a,x,y) &\mapsto 2 \int \left( \theta(a) - \theta_{0}(a) \right) \left( \theta_{0}(a) - m_{0}(x, a) \right) d\mu(a) \\
&\quad + 2 \frac{d\mu_{A}}{dQ^{0}_{A \mid X}}(a \mid x) \left( \theta(a) - \theta_{0}(a) \right) \bigl( m_{0}(x, a) - y \bigr).
\end{align*}
By Corollary~1 in \citet{li_efficient_2023}, the canonical gradient of $P\mapsto L_P(\theta)$ at $P^0$ under the data fusion setting defined by fusion sets $(\mathcal{S}_{X},\mathcal{S}_{Y})$ is
\begin{align*}
D_{P^{0}}^{*}(\theta): (a,x,y,s) &\mapsto \mathbf{1}(s \in \mathcal{S}_{X}) \cdot 2 \xi_{0} 
\int \left( \theta(a) - \theta_{0}(a) \right) 
\left( \theta_{0}(a) - m_{0}(x, a) \right) d\mu(a) \\[8pt]
&\quad + \mathbf{1}(s \in \mathcal{S}_{Y}) \cdot 2 \eta_{0}w_{0}(x, a) \left( \theta(a) - \theta_{0}(a) \right)  \left( m_{0}(x, a) - y \right),
\end{align*}
where 
\begin{align*}
&\xi_{0} := \frac{1}{P^{0}(S \in \mathcal{S}_{X})}, \quad \eta_{0} := \frac{1}{P^{0}(S \in \mathcal{S}_{Y})}, \quad w_{0}(x, a) := \frac{p^{0}(x \mid S \in \mathcal{S}_{X})}{p^{0}(x \mid S \in \mathcal{S}_{Y})} \frac{1}{p^{0}(a \mid x, S \in \mathcal{S}_{Y})} \\
&m_{0}(x, a) := \E_{P^0} \left[ Y \mid x, a, S \in \mathcal{S}_{Y} \right]
\end{align*}
are the true nuisance parameters.
Similarly, the canonical gradient of $\mathcal{L}_{P}(\theta)$ at $P^{0}$ without data fusion (i.e., using only the intersection $\mathcal{S}_{X}\cap\mathcal{S}_{Y}$) is given by
\begin{align*}
\underline{D}_{P^{0}}^{*}(\theta): (a,x,y,s) &\mapsto \mathbf{1}(s \in \mathcal{S}_{X} \cap \mathcal{S}_{Y}) \cdot 2 \underline{\xi}_{0} 
\int \left( \theta(a) - \underline{\tau}_{0}(a) \right) 
\left( \underline{\tau}_{0}(a) - \underline{m}_{0}(x, a) \right) d\mu(a) \\[8pt]
&\quad + \mathbf{1}(s \in \mathcal{S}_{X} \cap \mathcal{S}_{Y}) \cdot 2 \underline{\eta}_{0}\underline{w}_{0}(x, a) \left( \theta(a) - \underline{\tau}_{0}(a) \right)  \left( \underline{m}_{0}(x, a) - y \right),
\end{align*}
where 
\begin{align*}
&\underline{\xi}_{0} = \underline{\eta}_{0} := \frac{1}{P^{0}(S \in \mathcal{S}_{X} \cap \mathcal{S}_{Y})}, \quad \underline{w}_{0}(x, a) := \frac{1}{p^{0}(a \mid x, S \in \mathcal{S}_{X} \cap \mathcal{S}_{Y})} \\
&\underline{m}_{0}(x, a) := \E_{P^0} \left[ Y \mid x, a, S \in \mathcal{S}_{X} \cap \mathcal{S}_{Y} \right].
\end{align*}
If $\mathcal{Q}$ is not locally nonparametric, then the above quantities are still gradients that can be used to construct one-step estimators, but they may not be canonical gradients.

\begin{lemma}[Stochastically approximated loss] \label{lem:onestep_loss}
The (oracle) pointwise loss function associated with the (nonparametric) one-step estimator of \( L_{P^0}(\theta) \) under data fusion is defined as
\[
\ell_{P^0}^{\text{OS}}(\theta, g_0; x, a, y, s)
:= \E_{B \sim \mu}\big[\ell_{P^0}(\theta, g_0; x, a, y, s, B)\big],
\]
where the inner loss function \( \ell_{P^0}(\theta, g_0; x, a, y, s, b) \) given by
\begin{align}
\ell_{P^0}(\theta, g_0; x, a, y, s, b) 
&:= \left(\theta(b) - \theta_{0}(b)\right)^2 \nonumber \\
&\quad + 2\xi_0 \mathbf{1}(s \in \mathcal{S}_X)
    \left(\theta(b) - \theta_{0}(b)\right)
    \left(\theta_{0}(b) - m_0(b, x)\right) \nonumber \\
&\quad + 2\eta_0 w_0(x, a) \mathbf{1}(s \in \mathcal{S}_Y)
    \left(\theta(a) - \theta_{0}(a)\right)
    \left(m_0(x, a) - y\right).
\end{align}
The empirical mean of this loss, \( \mathbb{P}_n\big[\ell_{P^0}(\theta, g_0; \cdot\,) \big] \), is a stochastic approximation of the oracle one-step estimator of \( L_{P^0}(\theta) \). Similarly, the (oracle) loss function without data fusion is defined as
\[
\uline{\ell}_{P^0}^{\text{OS}}(\theta, g_0; x, a, y, s)
:= \E_{B \sim \mu}\big[\uline{\ell}_{P^0}(\theta, g_0; x, a, y, s, B)\big],
\]
where
\begin{align}
\uline{\ell}_{P^0}(\theta, g_0; x, a, y, s, b)
&:= \left(\theta(b) - \theta_{0}(b)\right)^2 \nonumber \\
&\quad + 2\underline{\xi}_0 \mathbf{1}(s \in \mathcal{S}_X \cap \mathcal{S}_Y)
    \left(\theta(b) - \underline{\tau}_0(b)\right)
    \left(\underline{\tau}_0(b) - \underline{m}_0(b, x)\right) \nonumber \\
&\quad + 2\underline{\eta}_0 \underline{w}_0(x, a) \mathbf{1}(s \in \mathcal{S}_X \cap \mathcal{S}_Y)
    \left(\theta(a) - \underline{\tau}_0(a)\right)
    \left(\underline{m}_0(x, a) - y\right).
\label{eq:loss2}
\end{align}
\end{lemma}

\begin{proof}[Proof of \Cref{lem:onestep_loss}]
We construct a one-step estimator of $L_{P^0}(\theta)$ using $D_{P^{0}}^{*}(\theta)$. Given an estimator $P$ of $P^0$, the one-step expansion of $L_{P^0}$ under data fusion is
\begin{align*}
&L_{P^0} + \mathbb{P}_{n}D_{P^0}^{*}(\theta) \\
&= \int \left( \theta(a) - \theta_{0}(a) \right)^2 d\mu(a)  + \frac{2\xi_0}{n} \sum_{i=1}^{n} \mathbf{1}(S_i \in \mathcal{S}_X) \int \left( \theta(a) - \theta_{0}(a) \right)\left( \theta_{0}(a) - m_0(X_{i}, a) \right) d\mu(a) \\
&\quad + \frac{2\eta_0}{n} \sum_{i=1}^{n} \mathbf{1}(S_i \in \mathcal{S}_Y) w_0(X_i, A_i) \left( \theta(A_i) - \theta_{0}(A_i) \right) \left( m_0(X_i, A_i) - Y_i \right).
\end{align*}

Since the first two terms involve an expectation with respect to the known measure \( \mu \), we can approximate them using an artificial sample \( B_i \sim \mu \) and the corresponding empirical average:
\begin{align*}
&\approx \frac{1}{n} \sum_{i=1}^{n} \left( \theta(B_i) - \theta_{0}(B_{i}) \right)^2 + \frac{2\xi_0}{n} \sum_{i=1}^{n} \mathbf{1}(S_i \in \mathcal{S}_X) \left( \theta(B_i) - \theta_{0}(B_i) \right)\left( \theta_{0}(B_i) - m_0(X_i, B_i) \right) \\
&\quad + \frac{2\eta_0}{n} \sum_{i=1}^{n} \mathbf{1}(S_i \in \mathcal{S}_Y) w_0(A_i, X_i) \left( \theta(A_i) - \theta_{0}(A_i) \right) \left( m_0(X_i, A_i) - Y_i \right) \\
&= \mathbb{P}_n\big[\ell_{P^0}(\theta, g_0; \cdot) \big],
\end{align*}
where $Z_{i} = (X_{i}, A_{i}, Y_{i}, S_{i}, B_{i})$. Hence, the approximated loss under data fusion is defined as
\begin{align*}
\ell_{P^0}(\theta, g_0; (x, a, y, s, b)) := 
&\left(\theta(b) - \theta_{0}(b)\right)^2 \nonumber \\
&\quad + 2\xi_0 \mathbf{1}(s \in \mathcal{S}_X)\left(\theta(b) - \theta_{0}(b)\right)\left(\theta_{0}(b) - m_0(x, b)\right) \nonumber \\
&\quad + 2\eta_0 w_0(x, a) \mathbf{1}(s \in \mathcal{S}_Y)\left(\theta(a) - \theta_{0}(a)\right)\left(m_0(x, a) - y\right),
\end{align*}
with nuisance parameters defined as
\begin{align*}
&\xi_0 := \frac{1}{P^0(S \in \mathcal{S}_X)}, \quad \eta_0 := \frac{1}{P^0(S \in \mathcal{S}_Y)}, \quad w_0(x, a) := \frac{p^0(x \mid S \in \mathcal{S}_X)}{p^0(x \mid S \in \mathcal{S}_Y)} \cdot \frac{1}{p^0(a \mid x, S \in \mathcal{S}_Y)} \\
&m_0(x, a) := \E_{P^0}[Y \mid x, a, S \in \mathcal{S}_Y].
\end{align*}

A similar procedure can be applied to define the one-step loss function without data fusion.
\end{proof}

\subsection{Closed-Form Solution of Kernel Ridge Regression (\Cref*{lem:rkhs_alg})}
\begin{lemma}[Closed-form solution of kernel ridge regression]\label{lem:rkhs_alg}
The closed-form estimator \( \widehat{\theta} \) obtained in \Cref{alg:rkhs_est} corresponds to the solution of the kernel ridge regression problem in the RKHS \( \mathcal{H} \) induced by the kernel \( \mathcal{K} \), with the objective function given by
\[
\mathbb{P}_n\big[\ell_{P^0}(\theta, \widehat{g}; \cdot) \big] + \lambda_{n} \|\theta\|^{2}_{\mathcal{H}},
\]
where \( \{u_i, v_i\}_{i=1}^{n} \) are defined as in \Cref{alg:rkhs_est}. The resulting estimator admits the following closed-form expression:
\[
\widehat{\theta}(\cdot) = \frac{1}{\sqrt{n}} \sum_{j=1}^{n} \left[ \widehat{\beta}_{j} \mathcal{K}(\cdot, a_{j}) + \widehat{\gamma}_{j} \mathcal{K}(\cdot, b_{j}) \right],
\]
where the weights \( \widehat{\beta} \) and \( \widehat{\gamma} \) are computed as described in \Cref{alg:rkhs_est}.
\end{lemma}

\begin{proof}[Proof of \Cref{lem:rkhs_alg}]
The empirical objective function for kernel ridge regression is given by
\begin{align*}
\widehat{R}(\theta) 
&:= \mathbb{P}_n\left[\ell_{P^0}(\theta, \widehat{g}; \cdot)\right] + \lambda_n \|\theta\|_{\mathcal{H}}^2 \\
&= \frac{1}{n}\sum_{i=1}^n \left( \theta(b_i) - \widehat{\tau}(b_i) \right)^2 
+ \frac{2}{n}\sum_{i=1}^n \widehat{\xi} \cdot \mathbf{1}(s_i \in \mathcal{S}_X) \left( \theta(b_i) - \widehat{\tau}(b_i) \right) \left( \widehat{\tau}(b_i) - \widehat{m}(x_i, b_i) \right) \\
&\quad + \frac{2}{n}\sum_{i=1}^n \widehat{\eta} \cdot \widehat{w}(x_i, a_i) \cdot \mathbf{1}(s_i \in \mathcal{S}_Y) \left( \theta(a_i) - \widehat{\tau}(a_i) \right) \left( \widehat{m}(x_i, a_i) - y_i \right) + \lambda_n \|\theta\|_{\mathcal{H}}^2.
\end{align*}
We reorganize the expression by expanding the quadratic term and regrouping linear and constant components:
\begin{align*}
\widehat{R}(\theta) 
&= \frac{1}{n}\sum_{i=1}^n \theta(b_i)^2 
+ \frac{2}{n}\sum_{i=1}^n v_i \cdot \theta(b_i) 
+ \frac{2}{n}\sum_{i=1}^n u_i \cdot \theta(a_i) 
+ \frac{1}{n}\sum_{i=1}^n C_i 
+ \lambda_n \|\theta\|_{\mathcal{H}}^2,
\end{align*}
where
\begin{align*}
u_i &:= \widehat{\eta} \cdot \widehat{w}(a_i, x_i) \cdot \mathbf{1}(s_i \in \mathcal{S}_Y) \cdot (\widehat{m}(x_i, a_i) - y_i), \\
v_i &:= \widehat{\xi} \cdot \mathbf{1}(s_i \in \mathcal{S}_X) \cdot \left( \widehat{\tau}(b_i) - \widehat{m}(x_i, b_i) \right) - \widehat{\tau}(b_i), \\
C_i &:= \widehat{\tau}(b_i)^2 
- 2\widehat{\xi} \cdot \mathbf{1}(s_i \in \mathcal{S}_X) \cdot \left( \widehat{\tau}(b_i) - \widehat{m}(x_i, b_i) \right) \cdot \widehat{\tau}(b_i) \\
&\quad - 2\widehat{\eta} \cdot \widehat{w}(x_i, a_i) \cdot \mathbf{1}(s_i \in \mathcal{S}_Y) \cdot (\widehat{m}(x_i, a_i) - y_i) \cdot \widehat{\tau}(a_i).
\end{align*}

Denote by \( \theta \in \mathcal{H} \) of the form
\[
\theta(\cdot) = \frac{1}{n} \sum_{j=1}^n [\beta_j \mathcal{K}(\cdot, a_j) + \gamma_j \mathcal{K}(\cdot, b_j)].
\]
Then the empirical objective becomes
\begin{align*}
\widehat{R}(\theta) 
&= \frac{1}{n} \sum_{i=1}^n \left( \sum_{j=1}^n \frac{1}{n} \beta_j \mathcal{K}(b_i, a_j) + \gamma_j \mathcal{K}(b_i, b_j) \right)^2 + \frac{2}{n} \sum_{i=1}^n v_i \left( \sum_{j=1}^n \frac{1}{n} \beta_j \mathcal{K}(b_i, a_j) + \gamma_j \mathcal{K}(b_i, b_j) \right) \\
&\quad + \frac{2}{n} \sum_{i=1}^n u_i \left( \sum_{j=1}^n \left( \frac{1}{n} \beta_j \mathcal{K}(a_i, a_j) + \gamma_j \mathcal{K}(a_i, b_j) \right) \right) + \lambda_n \langle \theta, \theta \rangle_{\mathcal{H}} + \sum_{i=1}^n C_i .
\end{align*}
Let \( u := (u_1, \dots, u_n)^\top \), \( v := (v_1, \dots, v_n)^\top \), and define the block Gram matrix
\[
\mathbf{K} = 
\begin{pmatrix}
\mathbf{K}_{11} & \mathbf{K}_{12} \\
\mathbf{K}_{21} & \mathbf{K}_{22}
\end{pmatrix}
:= \frac{1}{n} \begin{pmatrix}
\mathcal{K}(a_i, a_j) & \mathcal{K}(a_i, b_j) \\
\mathcal{K}(b_i, a_j) & \mathcal{K}(b_i, b_j)
\end{pmatrix}.
\]

Then the objective simplifies to:
\begin{align*}
\widehat{R}(\theta) 
&= \frac{1}{n} \left( 
\beta^\top \mathbf{K}_{12} \mathbf{K}_{21} \beta 
+ \beta^\top \mathbf{K}_{12} \mathbf{K}_{22} \gamma 
+ \gamma^\top \mathbf{K}_{22} \mathbf{K}_{21} \beta 
+ \gamma^\top \mathbf{K}_{22} \mathbf{K}_{22} \gamma \right) \\
&\quad + \frac{2}{n} \left(
\beta^\top \mathbf{K}_{12} v + \gamma^\top \mathbf{K}_{22} v + \beta^\top \mathbf{K}_{11} u + \gamma^\top \mathbf{K}_{21} u \right) \\
&\quad + \frac{\lambda_n}{n} \left(
\beta^\top \mathbf{K}_{11} \beta + 2\beta^\top \mathbf{K}_{12} \gamma + \gamma^\top \mathbf{K}_{22} \gamma \right) + \sum_{i=1}^{n}C_{i}.
\end{align*}
We compute the gradients of the empirical risk with respect to \( \beta \) and \( \gamma \) as follows:
\begin{align*}
    \begin{bmatrix}
      \frac{\partial}{\partial \beta} \widehat{R}(\theta) \\[0.5em] \frac{\partial}{\partial \gamma} \widehat{R}(\theta)
    \end{bmatrix}
    &= \frac{2}{n}
    \begin{pmatrix}
    \mathbf{K}_{11} & \mathbf{K}_{12} \\
    \mathbf{K}_{21} & \mathbf{K}_{22}
    \end{pmatrix}
    \begin{pmatrix}
    u + \lambda_n \beta \\
    \mathbf{K}_{21} \beta + \mathbf{K}_{22} \gamma + v + \lambda_n \gamma
    \end{pmatrix}.
\end{align*}

Setting the gradient to zero and using the fact that, by \Cref{cond:kernel}, 
the block matrix $\mathbf{K}$ is invertible, we obtain that the minimizer 
$(\widehat{\beta}, \widehat{\gamma})$ satisfies
\begin{align*}
    \left(\widehat{\beta},\widehat{\gamma}\right) 
    &= \left(
        -\frac{u}{\lambda_n},\;
        -\left( \mathbf{K}_{22} + \lambda_n I_n \right)^{-1} 
          \left( \mathbf{K}_{21} \widehat{\beta} + v \right)
      \right),
\end{align*}
where $I_n$ denotes the identity matrix $n \times n$.
\end{proof}

\paragraph{Notation for functional derivatives.}
Let $L(\theta,g)$ be a functional defined on 
$\Theta \times \mathcal{G}$, where $\Theta \subset L^2(\mu)$ 
and $\mathcal{G} \subset L^2(\mu\times Q_X^0; \ell^2)$.

\begin{itemize}
\item The \emph{(Gâteaux) derivative} of $L$ with respect to $\theta$ at 
$(\theta^*, g_0)$ in the direction $h \in L^2(\mu)$ is
\[
\mathcal{D}_{\theta} L(\theta^*, g_0)[h]
:= \left.\frac{d}{dt} L(\theta^* + t h, g_0)\right|_{t=0}.
\]

\item The derivative of $L$ with respect to $g$ at $(\theta^*, g_0)$ 
in the direction $k \in L^2(\mu\times Q_X^0; \ell^2)$ is
\[
\mathcal{D}_{g} L(\theta^*, g_0)[k]
:= \left.\frac{d}{dt} L(\theta^*, g_0 + t k)\right|_{t=0}.
\]

\item The \emph{second derivative} with respect to $\theta$ in directions 
$h_1, h_2 \in L^2(\mu)$ is
\[
\mathcal{D}^2_{\theta} L(\theta^*, g_0)[h_1,h_2]
:= \left.\frac{\partial^2}{\partial s\,\partial t}
L(\theta^* + t h_1 + s h_2, g_0)\right|_{t=0,\,s=0}.
\]

\item The \emph{mixed derivative} with respect to $\theta$ and $g$ in directions 
$h \in L^2(\mu)$ and $k \in L^2(\mu\times Q_X^0; \ell^2)$ is
\[
\mathcal{D}_{\theta} \mathcal{D}_{g} L(\theta^*, g_0)[h,k]
:= \left.\frac{\partial^2}{\partial t\,\partial s}
L(\theta^* + t h,\, g_0 + s k)\right|_{t=0,\,s=0}.
\]

\end{itemize}

\subsection{Properties of the Population Loss (\Cref*{lem:neyman})}
\begin{lemma}[Properties of the population loss]\label{lem:neyman}
The loss function in \Cref{eq:loss} satisfies the following properties. Assume that \( \Theta \) is convex, and recall that \( \theta^* \) denotes the minimizer of the oracle risk over \( \Theta \).

\begin{enumcond}\label{cond:loss_properties}
    \item \textbf{(First-order optimality)}\label{cond:first_order}
    The minimizer \( \theta^* \) satisfies the first-order optimality condition:
    \begin{align*}
        \mathcal{D}_{\theta} L_{P^0}(\theta^*, g_0)[\theta - \theta^*] 
        &= 2 \, \E_{\mu} \left[ \big( \theta(A) - \theta^*(A) \big) \big( \theta^*(A) - \theta_{0}(A) \big) \right] \geq 0,
    \end{align*}
    for all \( \theta \in L^2(\mu) \).

    \item \textbf{(Neyman orthogonality)}\label{cond:neyman}
    The population risk \( L_{P^0} \) is Neyman orthogonal:
    \begin{align*}
        \mathcal{D}_g \mathcal{D}_{\theta} L_{P^0}(\theta^*, g_0)[\theta - \theta^*, g - g_0] = 0,
    \end{align*}
    for all \( \theta \in L^2(\mu) \) and \( g \in \mathcal{G} \).

    \item \textbf{(Second-order smoothness and strong convexity)}\label{cond:second_order}
    The population risk is second-order smooth and strongly convex in \( \theta \):
    \begin{align*}
        \mathcal{D}^2_{\theta} L_{P}(\overline{\theta}, g)[\theta - \theta^*, \theta - \theta^*] 
        &= 2 \, \E_{\mu} \left[ \big( \theta(A) - \theta^*(A) \big)^2 \right] 
        = 2 \| \theta - \theta^* \|_{L^2(\mu)}^2,
    \end{align*}
    for all \( \theta \in \Theta \), \( \overline{\theta} \in L^2(\mu) \), and \( g \in \mathcal{G} \).

    \item \textbf{(Higher-order smoothness)}\label{cond:higher_order}
    The population risk is higher-order smooth in the nuisance parameter:
    \begin{align*}
        \left| \mathcal{D}^2_{g} \mathcal{D}_{\theta} L_{P}(\theta^*, \overline{g})[\theta - \theta^*, g - g_0, g - g_0] \right| 
        \leq 2 M_{\lambda} \| \theta - \theta^* \|_{L^2(\mu)} \| g - g_0 \|_{L^4(Q^{0}_{X} \times \mu; \ell^2)}^2,
    \end{align*}
    for all \( \theta \in L^2(\mu) \), \( g \in \mathcal{G} \), and \( \overline{g} \in \operatorname{star}(\mathcal{G}, g_0) \).
\end{enumcond}
\end{lemma}
\begin{proof}[Proof of \Cref{lem:neyman}]
We will use the fact that $\Theta$ is convex to prove first-order optimality. Since $\Theta$ is convex, $t\theta - (1-t)\theta^{*} \in \Theta$ whenever $\theta, \theta^{*} \in \Theta$ and $t \in [0, 1]$. By the definition of $\theta^{*}$,
\begin{align*}
&\|\theta_{0} - \theta^{*} \|_{L^{2}(\mu)} = \inf_{\theta \in \Theta} \|\theta_{0} - \theta \|_{L^{2}(\mu)} \leq \|\theta_{0} - (t\theta + (1-t)\theta^{*}) \|_{L^{2}(\mu)}
\end{align*}
implies
\begin{align*}
\|\theta_{0} - \theta^{*} \|_{L^{2}(\mu)}^{2} &\leq \|\theta_{0} - (t\theta + (1-t)\theta^{*}) \|_{L^{2}(\mu)}^{2} \\
&= \|\theta_{0} - \theta^{*} - t(\theta - \theta^{*}) \|_{L^{2}(\mu)}^{2} \\
&= \|\theta_{0} - \theta^{*} \|_{L^{2}(\mu)}^{2} - 2t \langle \theta_{0} - \theta^{*}, \theta - \theta^{*} \rangle_{L^{2}(\mu)} + t^{2}\|\theta - \theta^{*} \|_{L^{2}(\mu)}^{2}
\end{align*}
Canceling out the first term and for any $t \in (0, 1]$, we can rearrange terms and divide both sides by $t$ to find
\begin{align*}
2 \langle \theta_{0} - \theta^{*}, \theta - \theta^{*} \rangle_{L^{2}(\mu)} \leq t\|\theta - \theta^{*} \|_{L^{2}(\mu)}^{2} .
\end{align*}
Since above inequality holds for any $t \in [0, 1)$, we can take right limits as $t \rightarrow 0+$. Then the right handside goes to 0. This implies that
\begin{align*}
2 \langle \theta_{0} - \theta^{*}, \theta - \theta^{*} \rangle_{L^{2}(\mu_{A})} \leq \lim_{t \rightarrow 0+} t\|\theta - \theta^{*} \|_{L^{2}(\mu_{A})}^{2} = 0.
\end{align*}
Then for any $\theta \in \Theta$, the first order derivative of the population risk with respect to $\theta$,
\begin{align*}
&\mathcal{D}_{\theta} L_{\mathcal{P}^{0}}(\theta^{*}, g_{0})\left[ \theta - \theta^{*} \right] \\
&= \left. \frac{d}{dt} \, \E_{\mu} \left[ \left( \theta^{*}(A) + t(\theta(A) - \theta^{*}(A)) - \theta_{0}(A) \right)^{2} \right] \right|_{t=0} \\
&\quad + \left. \frac{d}{dt} \, \E_{P^{0}} \left[ \mathbf{1}(S \in \mathcal{S}_{X}) \cdot 2\xi_{0} \cdot \E_{\mu} \left[ \left( \theta^{*}(A) + t(\theta(A) - \theta^{*}(A)) \right) \left( \theta_{0}(A) - m_{0}(X, A) \right) \right] \right] \right|_{t=0} \\
&\quad + \left. \frac{d}{dt} \, \E_{P^{0}} \left[ \mathbf{1}(S \in \mathcal{S}_{Y}) \cdot 2\eta_{0} w_{0}(X, A) \cdot \left( \theta^{*}(A) + t(\theta(A) - \theta^{*}(A)) \right) \left( m_{0}(X, A) - Y \right) \right] \right|_{t=0} \\
&= 2 \E_{\mu} \left[ \left( \theta(A) - \theta^{*}(A) \right) \left( \theta^{*}(A) - \theta_{0}(A) \right) \right] \\
&\quad + 2 \E_{Q_{X}^{0}} \left[ \E_{\mu} \left[ \left( \theta(A) - \theta^{*}(A) \right) \left( \theta_{0}(A) - m_{0}(X, A) \right) \right] \right] \\
&\quad + 2 \E_{Q_{X}^{0}} \left[ \E_{\mu} \left[ \E_{Q_{Y \mid X, A}^{0}} \left[ \left( \theta(A) - \theta^{*}(A) \right) \left( m_{0}(X, A) - Y \right) \right] \right] \right] \\
&= 2 \E_{\mu} \left[ \left( \theta(A) - \theta^{*}(A) \right) \left( \theta^{*}(A) - \theta_{0}(A) \right) \right] \geq 0,
\end{align*}
where last two terms in the derivative disappear using the change of order of integration and the last inequality comes from the argument above.

Next, we prove Neyman orthogonality with respect to the nuisance parameters.  
For any \( \theta \in L^{2}(\mu) \) and nuisance parameter \( g \in \mathcal{G} \),
\begin{align*}
& \mathcal{D}_{g} \mathcal{D}_{\theta} L_{P^{0}}(\theta^{*}, g_{0})[\theta - \theta^{*}, g - g_{0}] \\
&= \E_{\mu} \left[ 
    \left. \frac{d^{2}}{dt_{1} dt_{2}} 
    \left( \theta^{*} + t_{2}(\theta - \theta^{*}) - \theta_{0} - t_{1}(\tau - \theta_{0}) \right)^{2} 
    \right|_{(t_{1}, t_{2}) = (0, 0)} 
\right] \\
&\quad + 2 \, \E_{Q^{0}_{X} \times \mu} \left[ 
    \left. \frac{d^{2}}{dt_{1} dt_{2}} 
    \left( \frac{\xi_{0} + t_{1}(\xi - \xi_{0})}{\xi_{0}} \right)
    \left( \theta^{*} + t_{2}(\theta - \theta^{*}) - \theta_{0} - t_{1}(\tau - \theta_{0}) \right) \right. \right. \\
&\quad \left. \left. \times 
    \left( \theta_{0} + t_{1}(\tau - \theta_{0}) - (m + t_{1}(m - m_{0})) \right)
    \right|_{(t_{1}, t_{2}) = (0, 0)} 
\right] \\
&\quad + 2 \, \E_{Q^{0}_{X} \times \mu \times Q_{Y \mid X, A}^{0}} \left[
    \left. \frac{d^{2}}{dt_{1} dt_{2}} 
    \left( \frac{\eta_{0} + t_{1}(\eta - \eta_{0})}{\eta_{0}} \right)
    \left( \frac{w_{0} + t_{1}(w - w_{0})}{w_{0}} \right) \right. \right. \\
&\quad \left. \left. \times
    \left( \theta^{*} + t_{2}(\theta - \theta^{*}) - \theta_{0} - t_{1}(\tau - \theta_{0}) \right)
    \left( m_{0} + t_{1}(m - m_{0}) - Y \right)
    \right|_{(t_{1}, t_{2}) = (0, 0)} 
\right] \\
&= -2 \, \E_{\mu} \left[ (\theta - \theta^{*})(\tau - \theta_{0}) \right] \\
&\quad + 2 \, \E_{Q^{0}_{X} \times \mu} \left[ 
    (\theta - \theta^{*}) \left( 
        \frac{\xi - \xi_{0}}{\xi_{0}} (\theta_{0} - m_{0}) 
        + (\tau - \theta_{0}) - (m - m_{0}) 
    \right) 
\right] \\
&\quad + 2 \, \E_{Q^{0}_{X} \times \mu \times Q_{Y \mid X, A}^{0}} \left[ 
    (\theta - \theta^{*}) \left( 
        \frac{\eta - \eta_{0}}{\eta_{0}} (m_{0} - Y)
        + \frac{w - w_{0}}{w_{0}} (m_{0} - Y)
        + (m - m_{0})
    \right) 
\right] \\
&= 0.
\end{align*}

Next, we derive the second-order derivative with respect to the target parameter and prove second-order smoothness and strong convexity.
\begin{align*}
\mathcal{D}_{\theta}^{2} L_{P^{0}}(\overline{\theta}, g_{0})[\theta - \theta^{*}, \theta - \theta^{*}]
&= \E_{\mu} \left[
    \left. \frac{d^2}{dt_{1} dt_{2}} 
    \left( \overline{\theta} + t_{1}(\theta - \theta^{*}) + t_{2}(\theta - \theta^{*}) - \theta_{0} \right)^{2}
    \right|_{(t_{1}, t_{2}) = (0, 0)}
\right] \\
&= \E_{\mu} \left[
    \left. \frac{d}{dt_{1}} 
    \left( 2 (\theta - \theta^{*}) \left( \overline{\theta} + t_{1}(\theta - \theta^{*}) - \theta_{0} \right) \right)
    \right|_{t_{1} = 0}
\right] \\
&= 2 \, \E_{\mu} \left[ (\theta(A) - \theta^{*}(A))^{2} \right].
\end{align*}

Lastly, we prove the higher-order smoothness of the risk:
\begin{align*}
&\mathcal{D}_{g}^{2}\mathcal{D}_{\theta}L_{P^{0}}(\theta^{*}, \overline{g})[\theta - \theta^{*}, g - g_{0}, g - g_{0}] \\
&= \E_{\mu}\!\left[ 
    \frac{d^{3}}{dt_{1}dt_{2}dt_{3}}
    \big( \theta^{*} + t_{3}(\theta - \theta^{*}) - \overline{\tau} - t_{1}(\tau - \theta_{0}) - t_{2}(\tau - \theta_{0}) \big)^{2}
\right]_{\!(0,0,0)} \\
&\quad + 2\E_{Q^{0}_{X} \times \mu}\!\Biggl[
    \frac{d^{3}}{dt_{1}dt_{2}dt_{3}}
    \left( \frac{\overline{\xi} + t_{1}(\xi - \xi_{0}) + t_{2}(\xi - \xi_{0})}{\xi_{0}} \right) \\
&\qquad \times \bigl(\theta^{*} + t_{3}(\theta - \theta^{*}) - \overline{\tau}
    - t_{1}(\tau - \theta_{0}) - t_{2}(\tau - \theta_{0})\bigr) \\
&\qquad \times \bigl(\overline{\tau} + t_{1}(\tau - \theta_{0}) + t_{2}(\tau - \theta_{0})
    - (\overline{m} + t_{1}(m - m_{0}) + t_{2}(m - m_{0}))\bigr)
\Biggr]_{\!(0,0,0)} \\
&\quad+ 2\E_{Q^{0}_{X} \times \mu \times Q^{0}_{|A,X}}\!\left[
    \frac{d^{3}}{dt_{1}dt_{2}dt_{3}}
    \left( \frac{\overline{\eta} + t_{1}(\eta - \eta_{0}) + t_{2}(\eta - \eta_{0})}{\eta_{0}} \right)
    \left( \frac{\overline{w} + t_{1}(w - w_{0}) + t_{2}(w - w_{0})}{w_{0}} \right) \right.\\
&\qquad \left. \times \big( \theta^{*} + t_{3}(\theta - \theta^{*}) - \overline{\tau} - t_{1}(\tau - \theta_{0}) - t_{2}(\tau - \theta_{0}) \big)
       \big( \overline{m} + t_{1}(m - m_{0}) + t_{2}(m - m_{0}) - Y \big)
\right]_{\!(0,0,0)} \\
&= 4\E_{Q^{0}_{X} \times \mu}\!\left[
    (\theta - \theta^{*})\left( \frac{\xi - \xi_{0}}{\xi_{0}} \right)
    \big( (\tau - \theta_{0}) - (m - m_{0}) \big)
\right] \\
&\quad+ 4\E_{Q^{0}_{X} \times \mu }\!\left[
    (\theta - \theta^{*}) \left(
        \frac{\eta - \eta_{0}}{\eta_{0}}\frac{w - w_{0}}{w_{0}}(\overline{m} - m_{0})
        + \frac{\overline{\eta}}{\eta_{0}} \frac{w - w_{0}}{w_{0}} (m - m_{0})
        + \frac{\eta - \eta_{0}}{\eta_{0}} \frac{\overline{w}}{w_{0}} (m - m_{0})
    \right)
\right] \\
&= 2\E_{Q^{0}_{X} \times \mu}\!\left[
    (\theta - \theta^{*})
    \begin{bmatrix}
        \xi - \xi_{0} \\[0.3em]
        \eta - \eta_{0} \\[0.3em]
        w - w_{0} \\[0.3em]
        m - m_{0} \\[0.3em]
        \tau - \theta_{0}
    \end{bmatrix}^{\!T}
    \underbrace{
    \begin{bmatrix}
        0 & 0 & 0 & 0 & \frac{1}{\xi_{0}} \\
        0 & 0 & \frac{\overline{m} - m_{0}}{\eta_{0}w_{0}} & \frac{\overline{w}}{\eta_{0}w_{0}} & 0 \\
        0 & \frac{\overline{m} - m_{0}}{\eta_{0}w_{0}} & 0 & \frac{\overline{\eta}}{\eta_{0}w_{0}} & 0 \\
        0 & \frac{\overline{w}}{\eta_{0}w_{0}} & \frac{\overline{\eta}}{\eta_{0}w_{0}} & 0 & 0 \\
        \frac{1}{\xi_{0}} & 0 & 0 & 0 & 0
    \end{bmatrix}
    }_{\text{:= A}}
    \begin{bmatrix}
        \xi - \xi_{0} \\[0.3em]
        \eta - \eta_{0} \\[0.3em]
        w - w_{0} \\[0.3em]
        m - m_{0} \\[0.3em]
        \tau - \theta_{0}
    \end{bmatrix}
\right].
\end{align*}

Here, we analyze the characteristic polynomial of \( A \) to identify its maximum eigenvalue:
\begin{align*}
\det(\lambda I - A) 
&= \left(\lambda^{2} - \frac{1}{\xi_{0}^{2}}\right) 
\cdot
\begin{vmatrix} 
\lambda & -\frac{\overline{m} - m_{0}}{\eta_{0}w_{0}} & -\frac{\overline{w}}{\eta_{0}w_{0}} \\ 
-\frac{\overline{m} - m_{0}}{\eta_{0}w_{0}} & \lambda & -\frac{\overline{\eta}}{\eta_{0}w_{0}} \\
-\frac{\overline{w}}{\eta_{0}w_{0}} & -\frac{\overline{\eta}}{\eta_{0}w_{0}} & \lambda \\ 
\end{vmatrix} \\
&= \lambda \left(\lambda^{2} - \frac{1}{\xi_{0}^{2}}\right) 
\left( \lambda^{2} - \frac{1}{\eta_{0}^{2}w_{0}^{2}} 
\left\{ \overline{\eta}^{2} + \overline{w}^{2} + (\overline{m} - m_{0})^{2} \right\} \right) = 0.
\end{align*}
This implies that the eigenvalues of \( A \) are given by
\[
\left\{ 0, \pm \frac{1}{\xi_0}, \pm \frac{\sqrt{\overline{\eta}^2 + \overline{w}^2 + (\overline{m} - m_0)^2}}{\eta_0 w_0} \right\}.
\]
We observe that the maximum eigenvalue is bounded above, since
\begin{align*}
\frac{1}{\xi_{0}} \leq 1, 
\quad \frac{\sqrt{\overline{\eta}^{2} + \overline{w}^{2} + (\overline{m} - m_{0})^{2}}}{\eta_{0}w_{0}} 
&\leq \overline{\eta} + \overline{w} + | \overline{m} - m_{0} | 
\leq M_{\eta} + M_{w} + 2 := M_{\lambda}.
\end{align*}
Since the operator norm \( \| A \| \) is equal to the maximum eigenvalue in magnitude, it follows that \( \| A \| \leq M_{\lambda} \). Therefore, we conclude that
\begin{align*}
\left| 
\mathcal{D}_{g}^{2} \mathcal{D}_{\theta} L_{P}(\theta^{*}, \overline{g})
\left[ \theta - \theta^{*}, g - g_{0}, g - g_{0} \right] 
\right| 
&\leq 2 \, \E_{\mu} \!\left[ (\theta - \theta^{*})^{2} \right]^{1/2} 
\, \E_{Q^{0}_{X} \times \mu} 
\!\left[ \big( (g-g_{0})^{\top} A (g - g_{0}) \big)^{2} \right]^{1/2} \\
&\leq 2 \, \E_{\mu} \!\left[ (\theta - \theta^{*})^{2} \right]^{1/2}
\, \E_{Q^{0}_{X} \times \mu} 
\!\left[ \| A \|^{2} \, \| g-g_{0} \|_{2}^{4} \right]^{1/2} \\
&\leq 2 M_{\lambda} \, 
\| \theta - \theta^{*} \|_{L^{2}(\mu)} 
\| g - g_{0} \|_{L^{4}(Q^{0}_{X} \times \mu; \ell^{2})}^{2}.
\end{align*}
\end{proof}

Based on the properties established in the previous lemma, we now derive a risk difference inequality at the true nuisance, showing that it is controlled by the target estimation error at $\widehat{g}$ together with the nuisance estimation error.

\subsection{Oracle Excess Risk Bound in Terms of Target and Nuisance Errors (Proof of \Cref*{lem:general})}
\begin{proof}
Our argument adapts the proof of Theorem~1 in \citet{foster_orthogonal_2023}. We first directly compute the difference in the risk function with respect to $\theta$ 
around $\theta^*$, holding the nuisance function fixed at $\hat{g}$:
\begin{equation}\label{eq:taylor_around_g0}
L_{P^{0}}(\hat{\theta}, \hat{g}) - L_{P^{0}}(\theta^{*}, \hat{g}) = \mathcal{D}_{\theta}L_{P^{0}}(\theta^{*}, \hat{g})[\hat{\theta} - \theta^{*}] 
+ \| \hat{\theta} - \theta^{*} \|_{L^{2}(\mu)}^{2}.
\end{equation}
We next expand the first-order term \(\mathcal{D}_{\theta}L_{P^{0}}(\theta^{*}, \hat{g})[\hat{\theta} - \theta^{*}]\) around \(g_0\) via Taylor expansion in \(g\):
\begin{equation*}
\begin{aligned}
\mathcal{D}_{\theta}L_{P^{0}}(\theta^{*}, \hat{g})[\hat{\theta} - \theta^{*}]
&= \mathcal{D}_{\theta}L_{P^{0}}(\theta^{*}, g_{0})[\hat{\theta} - \theta^{*}] + \mathcal{D}_{g}\mathcal{D}_{\theta}L_{P^{0}}(\theta^{*}, g_{0})[\hat{\theta} - \theta^{*}, \hat{g} - g_{0}] \\
&\quad + \tfrac{1}{2} \mathcal{D}^{2}_{g}\mathcal{D}_{\theta}L_{P^{0}}(\theta^{*}, \tilde{g})[\hat{\theta} - \theta^{*}, \hat{g} - g_{0}, \hat{g} - g_{0}]
\end{aligned}
\end{equation*}
for some \(\tilde{g}\) on the line segment between \(g_0\) and \(\hat{g}\). Substituting this into \eqref{eq:taylor_around_g0} yields:
\begin{align*}
L_{P^{0}}(\hat{\theta}, \hat{g}) - L_{P^{0}}(\theta^{*}, \hat{g}) 
&= \mathcal{D}_{\theta}L_{P^{0}}(\theta^{*}, g_{0})[\hat{\theta} - \theta^{*}]
+ \mathcal{D}_{g}\mathcal{D}_{\theta}L_{P^{0}}(\theta^{*}, g_{0})[\hat{\theta} - \theta^{*}, \hat{g} - g_{0}] \\
&\quad + \frac{1}{2} \mathcal{D}^{2}_{g}\mathcal{D}_{\theta}L_{P^{0}}(\theta^{*}, \tilde{g})[\hat{\theta} - \theta^{*}, \hat{g} - g_{0}, \hat{g} - g_{0}]
+ \| \hat{\theta} - \theta^{*} \|_{L^{2}(\mu)}^{2}.
\end{align*}
By the Neyman orthogonality property, the second term above vanishes. Then applying the smoothness bound on the third derivative and using Hölder and Young’s inequality, we obtain:
\begin{align*}
&L_{P^{0}}(\hat{\theta}, \hat{g}) - L_{P^{0}}(\theta^{*}, \hat{g}) \\
&\quad\geq \mathcal{D}_{\theta}L_{P^{0}}(\theta^{*}, g_{0})[\hat{\theta} - \theta^{*}]
- M_{\lambda} \| \hat{\theta} - \theta^{*} \|_{L^{2}(\mu)} \cdot \| \hat{g} - g_{0} \|_{L^{4}(Q^{0}_{X} \times \mu; \ell^{2})}^{2}
+ \| \hat{\theta} - \theta^{*} \|_{L^{2}(\mu)}^{2} \\
&\quad\geq \mathcal{D}_{\theta}L_{P^{0}}(\theta^{*}, g_{0})[\hat{\theta} - \theta^{*}]
- \frac{1}{2} \| \hat{\theta} - \theta^{*} \|_{L^{2}(\mu)}^{2}
- \frac{M_{\lambda}^{2}}{2} \| \hat{g} - g_{0} \|_{L^{4}(Q^{0}_{X} \times \mu; \ell^{2})}^{4}
+ \| \hat{\theta} - \theta^{*} \|_{L^{2}(\mu)}^{2},
\end{align*}
which implies
\begin{align}
\| \hat{\theta} - \theta^{*} \|_{L^{2}(\mu)}^{2} 
\leq 2\left( L_{P^{0}}(\hat{\theta}, \hat{g}) - L_{P^{0}}(\theta^{*}, \hat{g}) \right)
- 2 \mathcal{D}_{\theta}L_{P^{0}}(\theta^{*}, g_{0})[\hat{\theta} - \theta^{*}]
+ M_{\lambda}^{2} \| \hat{g} - g_{0} \|_{L^{4}(Q^{0}_{X} \times \mu; \ell^{2})}^{4}.
\label{eq:ub_of_theta_norm}
\end{align}
Finally, consider the risk difference with respect to $\theta$ around $\theta^{*}$ under $g_0$:
\begin{align*}
L_{P^{0}}(\hat{\theta}, g_{0}) - L_{P^{0}}(\theta^{*}, g_{0}) 
&= \mathcal{D}_{\theta}L_{P^{0}}(\theta^{*}, g_{0})[\hat{\theta} - \theta^{*}]
+ \| \hat{\theta} - \theta^{*} \|_{L^{2}(\mu)}^{2}.
\end{align*}
Plugging \eqref{eq:ub_of_theta_norm} into the above yields
\begin{align*}
&L_{P^{0}}(\hat{\theta}, g_{0}) - L_{P^{0}}(\theta^{*}, g_{0}) \\
&\quad\leq 2\left( L_{P^{0}}(\hat{\theta}, \hat{g}) - L_{P^{0}}(\theta^{*}, \hat{g}) \right)
- \mathcal{D}_{\theta}L_{P^{0}}(\theta^{*}, g_{0})[\hat{\theta} - \theta^{*}]
+ M_{\lambda}^{2} \| \hat{g} - g_{0} \|_{L^{4}(Q^{0}_{X} \times \mu; \ell^{2})}^{4} \\
&\quad\leq 2\left( L_{P^{0}}(\hat{\theta}, \hat{g}) - L_{P^{0}}(\theta^{*}, \hat{g}) \right)
+ M_{\lambda}^{2} \| \hat{g} - g_{0} \|_{L^{4}(Q^{0}_{X} \times \mu; \ell^{2})}^{4},
\end{align*}
where the last inequality follows from the first-order optimality of \(\theta^*\).
\end{proof}

Based on \Cref{lem:general}, we can derive the regret bound stated in \Cref{thm:excess_risk}. 
To prove \Cref{thm:excess_risk}, we employ the Talagrand concentration inequality for Lipschitz loss functions, 
allowing for the target parameter to be vector-valued.

\subsection{Excess Risk Bound for \Cref*{alg:gen_est} (Proof of \Cref*{thm:excess_risk})}
\begin{proof}[Proof of \Cref{thm:excess_risk}]
In our setting, the loss depends on the target function evaluated at two possibly distinct inputs $a$ and $b$.
To directly leverage the framework in \citet{foster_orthogonal_2023}, 
we reformulate the scalar-valued target $\theta: \mathbb{R} \to \mathbb{R}$ 
as a vector-valued function 
\[
\boldsymbol{\theta} : \mathbb{R}^2 \to \mathbb{R}^2, 
\quad \boldsymbol{\theta}(a,b) = \big(\theta(a), \theta(b)\big).
\]
This representation allows us to treat $\theta(a)$ and $\theta(b)$ jointly. Note that estimating $\boldsymbol{\theta}$ is equivalent to estimating $\theta$, 
but the vector-valued formulation streamlines the notation and 
aligns our setting with existing results on Lipschitz losses of vector-valued target parameters. We define the loss at $\boldsymbol{\theta}$ as
\begin{align*}
&\ell_{P^0}(\boldsymbol{\theta}, g_0; z) := \ell_{P^0}(\theta, g_0; z), \\
&L_{P^0}(\boldsymbol{\theta}, g_0; z) := L_{P^0}(\theta, g_0; z),
\end{align*}
and similarly denote $\boldsymbol{\theta}^*(a,b) = (\theta^*(a), \theta^*(b))$.

It is straightforward to verify that the multidimensional analogues of all conditions in Lemma \ref{lem:neyman} 
hold for $\boldsymbol{\theta}$ as well. In particular:
\begin{enumcond}\label{cond:vector_loss_properties}
    \item \textbf{First-order optimality}\label{cond:vector_first_order}  
    For all $\theta \in \Theta \subset L^2(\mu)$,
    \begin{align*}
    \mathcal{D}_{\boldsymbol{\theta}} L_{P^{0}}(\boldsymbol{\theta}^{*}, g_{0})\big[\boldsymbol{\theta} - \boldsymbol{\theta}^{*} \big] 
    &= 2 \E_{\mu} \big[ \big( \theta(A) - \theta^{*}(A) \big) \big( \theta^{*}(A) - \theta_{0}(A) \big) \big] \geq 0.
    \end{align*}

    \item \textbf{Neyman orthogonality}\label{cond:vector_neyman}  
    For all $\theta \in L^2(\mu)$ and $g \in \mathcal{G}$,
    \begin{align*}
    \mathcal{D}_{g}\mathcal{D}_{\boldsymbol{\theta}} L_{P^{0}}(\boldsymbol{\theta}^{*}, g_{0})\big[ \boldsymbol{\theta} - \boldsymbol{\theta}^{*}, g - g_{0} \big] = 0.
    \end{align*}

    \item \textbf{Second-order smoothness and strong convexity}\label{cond:vector_second_order}  
    For all $\theta \in \Theta$, $\overline{\theta} \in L^{2}(\mu)$, and $g \in \mathcal{G}$,
    \begin{align*}
    \mathcal{D}^{2}_{\boldsymbol{\theta}} L_{P}(\overline{\boldsymbol{\theta}}, g)\big[ \boldsymbol{\theta} - \boldsymbol{\theta}^{*}, \boldsymbol{\theta} - \boldsymbol{\theta}^{*} \big] 
    &= 2\E_{\mu}\big[(\theta(A) - \theta^{*}(A))^{2}\big] \\
    &= 2\E_{\mu \times \mu}\big[(\theta(A) - \theta^{*}(A)) (\theta(B) - \theta^{*}(B))\big] \\
    &= \| \boldsymbol{\theta} - \boldsymbol{\theta}^{*} \|_{L^{2}(\mu \times \mu; \ell^{2})}^{2}.
    \end{align*}

    \item \textbf{Higher-order smoothness}\label{cond:vector_higher_order}  
    For all $\theta \in L^{2}(\mu)$, $g \in \mathcal{G}$, and $\overline{g} \in \mathrm{star}(\mathcal{G}, g_{0})$,
    \begin{align*}
    \Big| \mathcal{D}_{g}^{2} \mathcal{D}_{\boldsymbol{\theta}} L_{P}(\boldsymbol{\theta}^{*}, \overline{g})\big[ \boldsymbol{\theta} - \boldsymbol{\theta}^{*}, g - g_{0}, g - g_{0} \big] \Big| 
    &\leq \sqrt{2} M_{\lambda} \| \boldsymbol{\theta} - \boldsymbol{\theta}^{*} \|_{L^{2}(\mu \times \mu; \ell^{2})} \| g - g_{0} \|_{L^{4}(Q^{0}_{X} \times \mu; \ell^2)}^{2}.
    \end{align*}
\end{enumcond}
We next establish that the loss under data fusion is Lipschitz continuous in its first argument, when evaluated at the nuisance estimators $\widehat{g}$ from \Cref{alg:gen_est}. 
To see this, note that
\begin{align*}
&\big| \ell_{P^0}(\boldsymbol{\theta}_{1}, \widehat{g}; z) - \ell_{P^0}(\boldsymbol{\theta}_{2}, \widehat{g}; z) \big| = \big| \ell_{P^0}(\theta_{1}, \widehat{g}; z) - \ell_{P^0}(\theta_{2}, \widehat{g}; z) \big| \\[0.5em]
&= \left|
\big( \theta_{1}(b) - \theta_{2}(b) \big)
\Big(
    \theta_{1}(b) + \theta_{2}(b) - 2\widehat{\tau}(b) 
    + 2\widehat{\xi}\,\mathbf{1}( s \in \mathcal{S}_{X} )
      \big( \widehat{\tau}(b) - \widehat{m}(b, x) \big)
\Big)
\right. \\
&\quad\left.
+ \big( \theta_{1}(a) - \theta_{2}(a) \big)
\Big(
    2\widehat{\eta}\,\widehat{w}(x, a)\,
    \mathbf{1}( s \in \mathcal{S}_{Y} )
    \big( \widehat{m}(x, a) - y \big)
\Big)
\right|
\end{align*}
This expression can be written as an inner product and bounded by
\[
\big\| \boldsymbol{\theta}_{1} - \boldsymbol{\theta}_{2} \big\|_{2} \cdot B(z),
\]
where $B(z)$ is defined as
\begin{equation*}
B(z) := \left\| 
\begin{pmatrix} 
\theta_{1}(b) + \theta_{2}(b) - 2\widehat{\tau}(b) 
    + 2\widehat{\xi}\,\mathbf{1}(s \in \mathcal{S}_{X})(\widehat{\tau}(b) - \widehat{m}(b, x)) \\[0.25em]
2\widehat{\eta}\,\widehat{w}(x, a)\mathbf{1}(s \in \mathcal{S}_{Y})(\widehat{m}(x, a) - y)
\end{pmatrix} \right\|_{2}.
\end{equation*}
We next derive a uniform bound on $B(z)$. First, observe that
\begin{align*}
B(z) 
&\leq \left| \theta_{1}(b) + \theta_{2}(b) - 2\widehat{\tau}(b) 
    + 2\widehat{\xi} \mathbf{1}(s \in \mathcal{S}_{X}) \big( \widehat{\tau}(b) - \widehat{m}(b, x) \big) \right| \\
&\quad + \left| 2\widehat{\eta}\, \widehat{w}(x, a) \mathbf{1}(s \in \mathcal{S}_{Y}) \big( \widehat{m}(x, a) - y \big) \right| \\
&\le 2\|\theta\|_{\infty} + 2|\widehat{\tau}(b)| 
    + 2\widehat{\xi} \cdot \big| \widehat{\tau}(b) - \widehat{m}(x, b) \big| 
    + 2\widehat{\eta} \cdot \widehat{w}(x, a) \mathbf{1}(s \in \mathcal{S}_{Y}) \big( |\widehat{m}(x, a)| + |y| \big) \\
&\le 4
    + 4\widehat{\xi} 
    + 2\widehat{\eta} \cdot \widehat{w}(x, a) \mathbf{1}(s \in \mathcal{S}_{Y}) \big( 1 + |y| \big),
\quad Q^{0}_{X} \times \mu\text{-a.e.},
\end{align*}
by $\|\theta\|_{\infty} \le 1$ and \Cref{cond:bounded-w-m}, and, under the consistency assumptions on $\widehat{\xi}, \widehat{\eta},$ and $\widehat{w}$, the following hold with high probability:
\begin{enumerate}
    \item Since $| \widehat{\xi} - \xi_{0} |^2 \xrightarrow{P^{0}} 0$, there exists $N(\delta)$ such that, for all $n \geq N(\delta)$,
    \[
    P^{0} \Big( | \widehat{\xi} - \xi_{0} | \leq \delta \xi_{0} \Big) \geq 1 - \delta/12,
    \quad \text{and thus} \quad
    P^{0} \Big( \widehat{\xi} \leq (1+\delta) \xi_{0} \Big) \geq 1 - \delta/12.
    \]
    \item Similarly, since $| \widehat{\eta} - \eta_{0} |^2 \xrightarrow{P^{0}} 0$, there exists $N(\delta)$ such that, for all $n \geq N(\delta)$,
    \[
    P^{0} \Big( | \widehat{\eta} - \eta_{0} | \leq \delta \eta_{0} \Big) \geq 1 - \delta/12,
    \quad \text{and thus} \quad
    P^{0} \Big( \widehat{\eta} \leq (1+\delta) \eta_{0} \Big) \geq 1 - \delta/12.
    \]

    \item Since $\| \widehat{w} - w_{0} \|_{L^{\infty}(Q^{0}_{X} \times \mu)} \xrightarrow{P^{0}} 0$, there exists $N(\delta)$ such that, for all $n \geq N(\delta)$,
    \[
    P^{0} \left( |\widehat{w}(x, a)| \leq (1 + \delta) \cdot \|w_{0}\|_{\infty} \text{ for } {Q^0_{X} \times \mu} \text{-almost every } (x, a) \right) \geq 1 - \delta/12.
    \]
\end{enumerate}

Combining these results, we conclude that
\[
B(z) \leq 4 + 4(1 + \delta)\xi_{0} + 2(1 + \delta)^{2}\eta_{0}\,\|w_{0}\|_{\infty}\,\mathbf{1}(s \in \mathcal{S}_{Y})(1 + |y|),
\]
$Q^{0}_{X} \times \mu$-a.e. with probability at least $1 - \delta/4$.

We now use that $Y$ is sub-exponential, as stated in \Cref{cond:subexp-y}. 
This moment condition yields the following sub-exponential tail bound 
(see, e.g., \citet{vershynin_high-dimensional_2018}): for any $t > 0$,
\begin{align*}
&P^{0}\!\left(\mathbf{1}(S \in \mathcal{S}_{Y}) \ c_{w} |Y - \theta_{0}(a)| \ge t \mid X=x, A=a \right) \\
&= P^{0}\!\left(S \in \mathcal{S}_{Y},\, c_{w}\lvert Y - \theta_{0}(a)\rvert \ge t \,\middle|\, X=x,A=a\right) \\
&= P^{0}\!\left(S \in \mathcal{S}_{Y} \mid X=x,A=a \right) \,
   P^{0}\!\left( \lvert Y - \theta_{0}(a)\rvert \ge \tfrac{t}{c_{w}} \,\middle|\, X=x,A=a,S \in \mathcal{S}_{Y}\right) \\
&\le 2\exp\!\left(
   -\tfrac{1}{2}\min\!\left\{
      \tfrac{(t/c_{w})^{2}}{\sigma^{2}},\;
      \tfrac{t/c_{w}}{L}
   \right\}
\right),
\end{align*}
where $c_{w}:=\eta_{0}\,\|w_{0}\|_{L^{\infty}(Q^{0}_{X} \times \mu)}$. \\
To derive a high-probability bound, set the right-hand side less than or equal to $\delta/4$, i.e.,
\[
2\exp\!\left(
   -\tfrac{1}{2}\min\!\left\{
      \tfrac{(t/c_{w})^{2}}{\sigma^{2}},\;
      \tfrac{t/c_{w}}{L}
   \right\}
\right) \;\le\; \frac{\delta}{4}.
\]
This inequality holds whenever
\[
t \;\ge\; c_{w}\cdot \max\!\left\{
   \sigma\sqrt{2\log\!\tfrac{8}{\delta}},\;\;
   2L\log\!\tfrac{8}{\delta}
\right\}.
\]
Consequently, with probability at least $1-\delta/4$,
\[
\mathbf{1}\{s \in \mathcal{S}_{Y}\}\,\eta_{0} \|w_{0}\|_{\infty}\,\lvert y - \theta_{0}(a)\rvert
\;\le\; 
\eta_{0}\|w_{0}\|_{\infty}\cdot \max\!\left\{
   \sigma\sqrt{2\log\!\tfrac{8}{\delta}},\;\;
   2L\log\!\tfrac{8}{\delta}
\right\}.
\]
In turn, this bound implies that, with probability at least $1 - \delta/4$,
\begin{align*}
    \mathbf{1}\{s \in \mathcal{S}_{Y}\}\,\eta_{0} \|w_{0}\|_{\infty} \,\lvert y \rvert 
    &\leq \mathbf{1}\{s \in \mathcal{S}_{Y}\}\,\eta_{0}  \|w_{0}\|_{\infty} \,\lvert \theta_{0}(a) \rvert  + \mathbf{1}\{s \in \mathcal{S}_{Y}\}\,\eta_{0}  \|w_{0}\|_{\infty} \,\lvert y - \theta_{0}(a) \rvert \\
    &\leq C_{\sigma, \delta} \ \eta_{0} \|w_{0}\|_{\infty},
\end{align*}
where $C_{\sigma, \delta}:= 1 + \max \left\{ \sigma \sqrt{2 \log(8/\delta)},\ 2L \log(8/\delta) \right\}$.
Putting everything together, we obtain
\begin{align*}
B(z) &\leq 4 + 4(1+\delta)\xi_{0} + 2(1+\delta)^{2} \eta_{0} \|w_{0}\|_{\infty}(1 + C_{\sigma, \delta})\\ 
&\le 4(1+\delta)^2 \big( 1 + \xi_{0} + C_{\sigma, \delta} \ \eta_{0} \|w_{0}\|_{\infty}\big)  =: B_{P^{0}}(\delta),
\end{align*}
$Q^{0}_{X} \times \mu$-a.e. with probability at least $1 - \delta/2$ for sufficiently large $n$. Therefore,
\begin{align*}
| \ell_{P^0}(\boldsymbol{\theta}_{1}, \widehat{g}; z) - \ell_{P^0}(\boldsymbol{\theta}_{2}, \widehat{g}; z) | 
&\leq B_{P^{0}}(\delta) \|\boldsymbol{\theta}_{1} - \boldsymbol{\theta}_{2}\|_{2},
\end{align*}
with probability at least $1-\delta/2$ for sufficiently large $n$.
This further entails
\begin{align*}
&\| \ell_{P^0}(\boldsymbol{\theta}_{1}, \widehat{g}; z) - \ell_{P^0}(\boldsymbol{\theta}_{2}, \widehat{g}; z) \|_{L^{\infty}(Q^{0}_{X} \times \mu)} 
\leq B_{P^{0}}(\delta) \|\boldsymbol{\theta}_{1} - \boldsymbol{\theta}_{2}\|_{L^{\infty}(\mu; \ell^{2})},\\
&\| \ell_{P^0}(\boldsymbol{\theta}_{1}, \widehat{g}; z) - \ell_{P^0}(\boldsymbol{\theta}_{2}, \widehat{g}; z) \|_{L^{2}(Q^{0}_{X} \times \mu)} 
\leq B_{P^{0}}(\delta) \|\boldsymbol{\theta}_{1} - \boldsymbol{\theta}_{2}\|_{L^{2}(\mu; \ell^{2})},
\end{align*}
again with probability at least $1-\delta/2$ for sufficiently large $n$.

Next, we show that the excess risk bound can be expressed as the sum of the squared critical radius multiplied by the Lipschitz constant, and the fourth power of the nuisance estimation error. 
Applying Theorem~3 of \citet{foster_orthogonal_2023} with the parameters 
\[
K_{2} = 2, \quad R = \sqrt{2}, \quad \beta_{1} = 2, \quad \beta_{2} = 2M_{\lambda}, \quad 
r = 0, \quad \lambda = 2, \quad \kappa = 0, \quad B_{1} = B_{P^{0}}(\delta)^{2}, \quad B_{2} = M_{\lambda}^{2},
\]
we obtain
\begin{align*}
    L_{P^{0}}(\widehat{\theta}, g_{0}) - L_{P^{0}}(\theta^{*}, g_{0})
    &= L_{P^{0}}(\widehat{\boldsymbol{\theta}}, g_{0}) - L_{P^{0}}(\boldsymbol{\theta}^{*}, g_{0}) \\
    &\;\;\leq C \Biggl( 
        B_{P^{0}}(\delta)^{2} 
        \left( \delta_{n}^{2} + \frac{\log(\delta^{-1})}{n} \right)
        + M_{\lambda}^{2} 
        \bigl\| \widehat{g} - g_{0} \bigr\|_{L^{4}(Q^{0}_{X} \times \mu; \ell^2)}^{2}
    \Biggr).
\end{align*}
for some universal constant $C>0$ (as in the proof of Theorem~3 in \citet{foster_orthogonal_2023}).
This bound holds for finite samples $n$ with probability at least $1 - \delta/2$. 
Since we have shown that the loss is Lipschitz in its first argument at $\widehat{g}$ with probability at least $1 - \delta/2$ for large $n$, we can combine both results and apply a union bound to conclude that the excess risk bound holds with probability at least $1 - \delta$ for large $n$.
\end{proof}

We have established the excess risk bound for the estimators in \Cref{alg:gen_est}. 
We now turn to \Cref{alg:rkhs_est}, for which the bound can be derived in a more direct manner by combining a local Rademacher complexity bound with the approximation error of kernel ridge regression. 
Finally, by appropriately selecting the cross-validation parameter $c$, we obtain a simplified bound.

\subsection{Excess Risk Bound for \Cref*{alg:rkhs_est} (Proof of \Cref*{thm:rkhs_excess})}
\begin{proof}
First, we know that the approximation error is bounded as
\begin{align*}
L_{P^{0}}(\theta^{*},g_{0}) - L_{P^{0}}(\theta_{0}, g_{0}) = \inf_{\theta \in \Theta} \| \theta_{0} - \theta \|_{L^{2}(\mu)}^{2} \leq c^{2 - \frac{1}{\alpha}} \| T_{\mathcal{K}}^{\alpha} \theta_{0} \|_{\mathcal{H}}^{1/\alpha},
\end{align*}
by Theorem 1.1 of \citet{smale_estimating_2003}. Next, we establish an excess risk bound for our estimator relative to $\theta^{*}$. Before applying \Cref{thm:bound_fusion}, we revisit the higher-order smoothness condition~(d) in the proof of \Cref{thm:excess_risk} to adjust it for the kernel ridge regression setting. By Lemma 5.1 of \citet{mendelson_regularization_2010}, for all $\theta \in \Theta$,
\begin{align}\label{eq:infty_norm_bound}
\| \theta - \theta^{*} \|_{L^{\infty}(\mu)} 
\leq C_{p} \|\theta - \theta^{*}\|_{\mathcal{H}}^p \|\theta - \theta^{*}\|_{L^{2}(\mu)}^{1-p},
\end{align}
where $C_{p}$ is a constant depending on $A, p, G$ in \Cref{cond:evd}.
By H{\"o}lder's inequality and \eqref{eq:infty_norm_bound},
\begin{align*}
    \Big| \mathcal{D}_{g}^{2} \mathcal{D}_{\boldsymbol{\theta}} L_{P}(\boldsymbol{\theta}^{*}, \overline{g})\big[ \boldsymbol{\theta} - \boldsymbol{\theta}^{*}, g - g_{0}, g - g_{0} \big] \Big| &\leq
    2\E_{Q^{0}_{X} \times \mu}\Big[ \big| (\theta - \theta^{*}) (g-g_{0})^{T} A (g - g_{0}) \big| \Big] \\
        &\leq 2 \| (\theta - \theta^{*}) \|_{L^{\infty}(\mu)} \E_{Q^{0}_{X} \times \mu} \Big[ \big| (g-g_{0})^{T} A (g-g_{0}) \big| \Big] \\
        &\leq 2M_{\lambda} \| (\theta - \theta^{*}) \|_{L^{\infty}(\mu)} \|g - g_{0}\|_{L^{2}(Q^{0}_{X} \times \mu, \ell^{2})}^{2} \\
        &\leq 2M_{\lambda}C_{p}\|\theta - \theta^{*}\|_{\mathcal{H}}^p \|\theta - \theta^{*}\|_{L^{2}(\mu)}^{1-p}\|g - g_{0}\|_{L^{2}(Q^{0}_{X} \times \mu, \ell^{2})}^{2} \\
        &\leq 2M_{\lambda}C_{p}c^{p} \|\theta - \theta^{*}\|_{L^{2}(\mu)}^{1-p}\|g - g_{0}\|_{L^{2}(Q^{0}_{X} \times \mu, \ell^{2})}^{2} \\
        &\leq 2M_{\lambda}C_{p}c^{p} \|\boldsymbol{\theta} - \boldsymbol{\theta}^{*}\|_{L^{2}(\mu)}^{1-p}\|g - g_{0}\|_{L^{2}(Q^{0}_{X} \times \mu, \ell^{2})}^{2}.
\end{align*}
Hence, we apply Theorem 3 of \citet{foster_orthogonal_2023} with the parameters
\[
K_{2} = 2, \quad R = \sqrt{2}, \quad \beta_{1} = 2,\quad r = p, \quad
\lambda = 2, \quad \kappa = 0, \quad B_{1} = B_{P^{0}}(\delta)^{2}, \quad B_{2} = (M_{\lambda}C_{p}c^{p})^{\frac{2}{1+p}},
\]
resulting in the following bound of sample error:
\begin{align*}
L_{P^{0}}(\widehat{\theta}, g_{0}) - L_{P^{0}}(\theta^{*}, g_{0}) 
&\leq C\Bigg( B_{P^{0}}(\delta)^{2} \left( \delta_{n}^{2} + \frac{\log(\delta^{-1})}{n} \right) + (M_{\lambda}C_{p}c^{p})^{\frac{2}{1+p}}\|\widehat{g} - g_{0}\|_{L^{2}(\ell^{2}, P^{0})}^{\frac{4}{1+p}} \Bigg)\\
&\lesssim B_{P^{0}}(\delta)^{2} \left(\delta_{n}^{2} + \frac{\log(\delta^{-1})}{n} \right) + c^{\frac{2p}{1+p}}\|\widehat{g} - g_{0}\|_{L^{2}(\ell^{2}, P^{0})}^{\frac{4}{1+p}},
\end{align*}
with $C$ is an universal constant and probability at least $1-\delta$ and $\delta_{n}$ is the critical radius. We will propose a particular solution of the equation defining the critical radius using Lemma 5 of \citet{nie_quasi-oracle_2021}. Given that $\mathcal{Z}_{i} = \epsilon_{i}, w = 1, h = \theta$, the local Rademacher complexity is bounded as
\begin{align*}
R_{n, \mathcal{Z}}(\Theta, \delta) \leq K \log(n) \frac{c^{p}\delta^{1-p}}{\sqrt{n}},
\end{align*}
where K is a constant.
Hence, $\delta_{n} = K c^{\frac{p}{1+p}} \log(n)^{\frac{1}{1+p}} n^{-\frac{1}{2(1+p)}}$ is a valid critical radius and if we plug this critical radius into the sample bound above,
\begin{align*}
L_{P^{0}}(\widehat{\theta}, g_{0}) - L_{P^{0}}(\theta^{*}, g_{0}) &\lesssim B_{P^{0}}( \delta)^{2} \left(\delta_{n}^{2} + \frac{\log(\delta^{-1})}{n} \right) + c^{\frac{2p}{1+p}}\|\widehat{g} - g_{0}\|_{L^{2}(Q^{0}_{X} \times \mu; \ell^{2})}^{\frac{4}{1+p}} \\
&\lesssim B_{P^{0}}(\delta)^{2} \left( c^{\frac{2p}{1+p}}(n\log(n)^{-2})^{-\frac{1}{1+p}} + \frac{\log(\delta^{-1})}{n} \right) + c^{\frac{2p}{1+p}}\|\widehat{g} - g_{0}\|_{L^{2}(Q^{0}_{X} \times \mu; \ell^{2})}^{\frac{4}{1+p}}
,
\end{align*}
which is the first result of Theorem. If $\| \widehat{g} - g_{0} \|_{L^{2}(\ell^{2}, P^{0})} = O_{p}(B_{P^{0}}(\delta)^{2/(1+p)}(n\log(n)^{-2})^{-1/4})$, the third term is absorbed into the first and the second is absorbed into the first for large $n$, giving us
\begin{align*}
L_{P^{0}}(\widehat{\theta}, g_{0}) - L_{P^{0}}(\theta^{*}, g_{0}) \lesssim B_{P^{0}}(\delta)^{2}c^{\frac{2p}{1+p}}(n\log(n)^{-2})^{-\frac{1}{1+p}}.
\end{align*}
Then if we look at the excess risk which is the sum of sample error and approximation error, it becomes
\begin{align*}
L_{P^{0}}(\widehat{\theta}, g_{0}) - L_{P^{0}}(\theta_{0}, g_{0}) &= L_{P^{0}}(\widehat{\theta}, g_{0}) - L_{P^{0}}(\theta^{*}, g_{0}) + L_{P^{0}}(\theta^{*}, g_{0}) - L_{P^{0}}(\theta_{0}, g_{0}) \\
&\lesssim B_{P^{0}}(\delta)^{2}c^{\frac{2p}{1+p}}(n\log(n)^{-2})^{-\frac{1}{1+p}} + c^{2 - \frac{1}{\alpha}}.
\end{align*}
If we choose $c = (B_{P^{0}}(\delta)^{-2(1+p)}n\log(n)^{-2})^{\frac{\alpha}{p + (1-2\alpha)}}$, then the error bound becomes
\begin{align*}
L_{P^{0}}(\widehat{\theta}, g_{0}) - L_{P^{0}}(\theta_{0}, g_{0}) \lesssim (B_{P^{0}}(\delta)^{-2(1+p)}n\log(n)^{-2})^{-\frac{1-2\alpha}{p+(1-2\alpha)}},
\end{align*}
which we desired.
\end{proof}

We are now ready to compare the upper bounds on the excess risk with and without data fusion. Since the difference between the two upper bounds is only up to a multiplicative factor of the Lipschitz constant, it suffices to compare the corresponding Lipschitz constants. We denote the Lipschitz constant without data fusion, in analogy to that with data fusion, as
\begin{align*}
    \underline{B}_{P^{0}}(\delta)
    := 4(1+\delta)^{2} \Bigl( 1 + \underline{\xi}_{0} 
    + C_{\sigma, \delta}\,\eta_{0}\,\|\underline{w}_{0}\|_{\infty} \Bigr).
\end{align*}

\subsection{Lipschitz Constant Decreases in Data Fusion (\Cref*{lem:lipschitz_comparison})}
\begin{lemma}[Lipschitz constant decreases in data fusion]\label{lem:lipschitz_comparison}
Suppose the conditions of \Cref{thm:bound_fusion} hold. 
Then there exists $N(\delta) \in \mathbb{N}$ such that, for all $n \geq N(\delta)$, with probability at least $1 - \delta/2$, the following inequalities hold for every realization $z := (x, a, b, y, s)$ of $Z \sim \widetilde{P}^{0}$ and all $\theta_{1}, \theta_{2} \in \Theta$:
\begin{align*}
&\bigl| \ell_{P^{0}}(\theta_{1}, \widehat{g}; z) 
      - \ell_{P^{0}}(\theta_{2}, \widehat{g}; z) \bigr|
   \;\leq\; B_{P^{0}}(\delta) 
      \cdot \left\| \bigl(\theta_{1}(a) - \theta_{2}(a),\, \theta_{1}(b) - \theta_{2}(b)\bigr) \right\|_{2}, \\
&\bigl| \underline{\ell}_{P^{0}}(\theta_{1}, \widehat{g}; z) 
      - \underline{\ell}_{P^{0}}(\theta_{2}, \widehat{g}; z) \bigr|
   \;\leq\; \underline{B}_{P^{0}}(\delta) 
      \cdot \left\| \bigl(\theta_{1}(a) - \theta_{2}(a),\, \theta_{1}(b) - \theta_{2}(b)\bigr) \right\|_{2}.
\end{align*}
Moreover, the constants satisfy $B_{P^{0}}(\delta) \leq \underline{B}_{P^{0}}(\delta)$,
with strict inequality whenever
\[
P^{0}\!\big(S \in \mathcal{S}_{X} \setminus \mathcal{S}_{Y}\big)
\;+\;
\essinf_{(X,A)\sim P^0} 
   P^{0}\!\big(S \in \mathcal{S}_{Y}\setminus \mathcal{S}_{X} \,\big|\, A,X,S\in\mathcal{S}_{Y}\big)
\;>\;0.
\]
\end{lemma}
\begin{proof}[Proof of \Cref{lem:lipschitz_comparison}]
In the proof of \Cref{thm:excess_risk}, we showed that there exists $N(\delta)$ such that, for all $n \geq N(\delta)$, with probability at least $1-\delta/2$,
\begin{align*}
\bigl| \ell_{P^{0}}(\theta_{1}, \widehat{g}; z) 
     - \ell_{P^{0}}(\theta_{2}, \widehat{g}; z) \bigr|
&\;\leq\; B_{P^{0}}(\delta) \cdot 
   \left\| \bigl( \theta_{1}(a) - \theta_{2}(a),\, \theta_{1}(b) - \theta_{2}(b) \bigr) \right\|_{2}.
\end{align*}

By a similar argument as in the proof of \Cref{thm:excess_risk}, for all $n \geq N(\delta)$,
\begin{align*}
\bigl| \underline{\ell}_{P^{0}}(\theta_{1}, \widehat{g}; z) 
     - \underline{\ell}_{P^{0}}(\theta_{2}, \widehat{g}; z) \bigr|
&\;\leq\; \underline{B}_{P^{0}}(\delta) \cdot 
   \left\| \bigl( \theta_{1}(a) - \theta_{2}(a),\, \theta_{1}(b) - \theta_{2}(b) \bigr) \right\|_{2},
\end{align*}
with probability at least $1-\delta/2$. It remains to show $B_{P^{0}}(\delta)\le \underline{B}_{P^{0}}(\delta)$ and to characterize when the inequality is strict.
\paragraph{Step 1. Comparison of $\xi_0$ terms.}
Since $S_X\cap S_Y\subset S_X$,
\[
\xi_0
=\frac{1}{P^{0}(S\in S_X)}
\;\le\;\frac{1}{P^{0}(S\in S_X\cap S_Y)}
=\underline{\xi}_0,
\]
with strict inequality if $P^{0}(S\in S_X\setminus S_Y)>0$.

\paragraph{Step 2. Comparison of $\eta_0\|w_0\|_\infty$ terms.}
Let $T_1:=\mathcal{S}_X\cap\mathcal{S}_Y$ and $T_2:=\mathcal{S}_Y\setminus\mathcal{S}_X$.
Fix a common dominating measure $\nu\times\mu$ for $(X,A)$ and write all densities with respect to it.
For any measurable $E\times F\subset\mathcal{A}\times\mathcal{X}$,
\begin{align*}
P^{0}(A\in E,X\in F,S\in\mathcal{S}_Y)
&=P^{0}(A\in E,X\in F,S\in T_1)+P^{0}(A\in E,X\in F,S\in T_2)\\
&=\frac{1}{\underline{\eta}_0}\!\int_{E\times F} p^{0}(x\mid S\in T_1)\,p^{0}(a\mid x,S\in T_1)\,d(\nu\times\mu)\\
&\quad+\frac{1}{\eta_0}\!\int_{E\times F} p^{0}(a,x\mid S\in\mathcal{S}_Y)\,
P^{0}\!\bigl(S\in T_2\mid a,x,S\in\mathcal{S}_Y\bigr)\,d(\nu\times\mu),
\end{align*}
where we used
\[
P^{0}(A\in E,X\in F,S\in T_2)
= \frac{1}{\eta_0}\!\int_{E\times F} p^{0}(a,x\mid S\in\mathcal{S}_Y)\,
P^{0}\!\bigl(S\in T_2\mid a,x,S\in\mathcal{S}_Y\bigr)\,d(\nu\times\mu).
\]

Define
\[
\varepsilon_{P^{0}}
:=\operatorname*{ess\,inf}_{(A,X)\sim P^{0}}
P^{0}\!\bigl(S\in T_2\,\bigm|\,A,X,S\in\mathcal{S}_Y\bigr)\in[0,1].
\]
Then for $\nu\times\mu$-a.e.\ $(a,x)$ with $S\in\mathcal{S}_Y$,
\[
P^{0}\!\bigl(S\in T_2\mid a,x,S\in\mathcal{S}_Y\bigr)\;\ge\;\varepsilon_{P^{0}}.
\]
Hence, for all measurable $E\times F$,
\begin{align*}
P^{0}(A\in E,X\in F,S\in\mathcal{S}_Y)
&\ge \frac{1}{\underline{\eta}_0}\!\int_{E\times F} p^{0}(x\mid S\in T_1)\,p^{0}(a\mid x,S\in T_1)\,d(\nu\times\mu)\\
&\quad+\frac{\varepsilon_{P^0}}{\eta_0}\!\int_{E\times F} p^{0}(a,x\mid S\in\mathcal{S}_Y)\,d(\nu\times\mu).
\end{align*}
But also
\[
P^{0}(A\in E,X\in F,S\in\mathcal{S}_Y)
=\frac{1}{\eta_0}\!\int_{E\times F} p^{0}(x\mid S\in\mathcal{S}_Y)\,p^{0}(a\mid x,S\in\mathcal{S}_Y)\,d(\nu\times\mu).
\]
Combining and rearranging,
\begin{align*}
&(1-\varepsilon_{P^0})\,\frac{1}{\eta_0}\!\int_{E\times F} p^{0}(x\mid S\in\mathcal{S}_Y)\,p^{0}(a\mid x,S\in\mathcal{S}_Y)\,d(\nu\times\mu)
\; \\& \ge\;
\frac{1}{\underline{\eta}_0}\!\int_{E\times F} p^{0}(x\mid S\in T_1)\,p^{0}(a\mid x,S\in T_1)\,d(\nu\times\mu).
\tag{$\ast$}
\end{align*}

Since \((\ast)\) holds for all measurable $E\times F$, we obtain the \emph{pointwise} a.e.\ bound
\[
(1-\varepsilon_{P^0})\,\frac{p^{0}(x\mid S\in\mathcal{S}_Y)\,p^{0}(a\mid x,S\in\mathcal{S}_Y)}{\eta_0}
\;\ge\;
\frac{p^{0}(x\mid S\in T_1)\,p^{0}(a\mid x,S\in T_1)}{\underline{\eta}_0}
\quad \nu\times\mu\text{-a.e. }(x,a).
\tag{$\ast\ast$}
\]

Note that $p^{0}(x \mid S \in T_{1}) = p^{0}(x \mid S \in \mathcal{S}_{X})$ by \Cref{cond:common}.

Taking inverse of \((\ast\ast)\) and multiply by $(1-\varepsilon_{P^{0}}) p^{0}(x \mid S \in T_{1})$ yields
\[
(1-\varepsilon_{P^0})\,\underline{\eta}_0\,\underline{w}_0(x,a)
\;\ge\;
\eta_0\,w_0(x,a)
\qquad Q_X^{0}\times\mu\text{-a.e. }(x,a),
\]
and taking essential suprema over $(x,a)$ gives
\[
(1-\varepsilon_{P^0})\,\underline{\eta}_0\,\|\underline{w}_0\|_\infty
\;\ge\;
\eta_0\,\|w_0\|_\infty.
\]

\paragraph{Step 3. Conclusion.}
Combining the comparisons for $\xi_0$ and for $\eta_0\|w_0\|_\infty$ shows
\[
B_{P^0}(\delta)\;\le\; \underline{B}_{P^0}(\delta),
\]
with strict inequality if $P^{0}(S\in S_X\setminus S_Y)>0$ or $\varepsilon_{P^0}>0$.
This completes the proof.

\end{proof}

Given that the Lipschitz constant decreases under data fusion and that the remainder terms in the risk bounds coincide with and without data fusion, we can show that the risk is smaller under data fusion when applying the same estimation algorithm as \Cref{alg:rkhs_est}.

\subsection{Data Fusion Improves the Excess Risk Upper Bound (Proof of \Cref*{thm:bound_fusion})}
\begin{proof}
Under the assumptions of \Cref{thm:rkhs_excess}, there exists a universal constant $C>0$ such that
\[
L_{P^{0}}(\widehat{\theta},g_0)
\;\le\;
C\,\Big(B_{P^{0}}(\delta)^{-2(1+p)}\,n\,\log(n)^{-2}\Big)^{-\frac{1-2\alpha}{\,p+(1-2\alpha)\,}}.
\]
Similarly, since we apply the same estimation algorithm and the remainder terms coincide, the bound without data fusion satisfies
\[
L_{P^{0}}(\underline{\theta},g_0)
\;\le\;
C\,\Big(\underline{B}_{P^{0}}(\delta)^{-2(1+p)}\,n\,\log(n)^{-2}\Big)^{-\frac{1-2\alpha}{\,p+(1-2\alpha)\,}}.
\]
Taking the ratio of the upper bounds and cancelling common $n,\log n$ factors yields
\begin{align*}
\left(\frac{1+\xi_{0}+C_{\sigma, \delta}\,\eta_{0}\|w_{0}\|_{\infty}}{1+\underline{\xi}_{0}+C_{\sigma, \delta}\,\underline{\eta}_{0}\|\underline{w}_{0}\|_{\infty}}\right)^{\frac{2(1+p)(1-2\alpha)}{\,p+(1-2\alpha)\,}}
\;\le\;1,
\end{align*}
with strict inequality when \[
P^{0}\!\big(S\in \mathcal{S}_{X}\setminus \mathcal{S}_{Y}\big)
\;+\;
\essinf_{(X,A)\sim P^{0}}
P^{0}\!\big(S\in \mathcal{S}_{Y}\setminus \mathcal{S}_{X}\mid A,X,S\in\mathcal{S}_{Y}\big)
\;>\;0.\] This completes the proof.
\end{proof}

\subsection{Minimax Lower Bound without Data Fusion (Proof of \Cref*{lem:minimax_bound})}

Define
\begin{align}
N_\delta:= \frac{1}{4} \left(
    \sqrt{\,3\log\!\frac{2}{\delta} \;+\; 4\left(\frac{\delta}{32}\right)^{-1-\frac{(1 - 2\alpha)}{p}}}
    \;-\;
    \sqrt{\,3\log\!\frac{2}{\delta}\,}
\right)^{2}. \label{eq:Ndelta}
\end{align}
As $\delta\rightarrow 0$, $N_\delta=[1+o(1)]\left(\delta/32\right)^{-1 - (1 - 2\alpha)/p}$.
\begin{proof}[Proof of \Cref{lem:minimax_bound}]


Let $\tau=(\tau_{k})_{k=0}^\infty$ denote an estimator sequence, where, for each $k$, $\tau_{k}$ takes as input $k$ i.i.d draws from $Q$ and outputs an estimate of $\theta_Q$; if $k=0$, then $\theta_{0}$ takes as input an empty set and returns an arbitrary value, such as the zero function. Define $\underline{\theta}_\tau$ to be the estimator based on $n$ i.i.d draws from $P$ that first selects the $K:= \sum_{i=1}^n \mathbf{1}(S_i\in\mathcal{S}_{XY})$ observations from sources in $\mathcal{S}_{XY}$, and then applies $\tau_{K}$ to this subsample of observations; expressed in notation, $\underline{\theta}_\tau(\{(X_i,A_i,Y_i,S_i)\}_{i=1}^n)=\tau_{K}(\{(X_i,A_i,Y_i)\}_{i : S_i\in\mathcal{S}_{XY}})$. For a distribution $Q$ on $\mathcal{X} \times \mathcal{A} \times \mathcal{Y}$, 
we denote by $Q^{n_{XY}}$ the product measure on $(\mathcal{X} \times \mathcal{A} \times \mathcal{Y})^{n_{XY}}$ 
corresponding to $n_{XY}$ i.i.d.\ samples drawn from $Q$. 
Similarly, for $P \in \mathcal{P}$ we write $P^{n}$ for the product measure on 
$(\mathcal{X} \times \mathcal{A} \times \mathcal{Y} \times \mathcal{S})^{n}$ 
of $n$ i.i.d.\ copies from $P$. For any $Q$ in the model $\mathcal{Q}$ defined in \Cref{cond:src}, further define 
\begin{align*}
    \mathcal{P}(Q):=\left\{P\in\mathcal{P} : P(X | S \in \mathcal{S}_{X}) = Q(X),
P(Y | X, A, S \in \mathcal{S}_{Y}) = Q(Y | X, A), P(S\in\mathcal{S}_{XY})=\tfrac{n_{XY}}{n}\right\},
\end{align*}
where we have suppressed the dependence on $n_{XY}/n$ in the notation. We also let $\|\cdot\|$ denote the $L^2(\mu)$ norm.

\paragraph{Part 1: Lower bounding the no-data-fusion minimax risk under sampling from $P$ by the worst-case Bayes risk under sampling from $Q$.} Fix $Q\in\mathcal{Q}$, $P\in\mathcal{P}(Q)$ such that $P(A|X,S=s)=Q(A|X)$ for all $s$, and an estimator sequence $\tau$. For any $k$ and $t>0$,
\begin{align*}
P^n\{\|\underline{\theta}_\tau-\theta_Q\|^2\ge t\mid K=k\}&= Q^{k}\{\|\tau_{k}-\theta_Q\|^2\ge t\},
\end{align*}
where we use that $\{(X_i,A_i,Y_i)\}_{i : S_i\in\mathcal{S}_{XY}}| K=k$ is equal in distribution to an iid sample from $Q^k$. 
Above we use the convention that $Q^0\{\|\theta_{0}-\theta_Q\|^2\ge t\}=1\{\|\theta_{0}-\theta_Q\|^2\ge t\}$. Hence, for $r\ge 0$,
\begin{align*}
P^n\{\|\underline{\theta}_\tau-\theta_Q\|^2\ge t\}&\ge  P^n\{\|\underline{\theta}_\tau-\theta_Q\|^2\ge t,K\le r\} \\
&= \sum_{k=1}^{r} Q^{k}\{\|\tau_{k}-\theta_Q\|^2\ge t\} P^n\{K=k\} \\
&\ge \min_{k : 0\le k\le r} Q^{k}\{\|\tau_{k}-\theta_Q\|^2\ge t\} \ P^n\{K\le r\}.
\end{align*}
Taking a supremum over $P\in\mathcal{P}(Q)$, followed by a supremum over $Q\in\mathcal{Q}$ and an infimum over estimator sequences $\tau$ yields
\begin{align*}
&\inf_\tau \sup_{Q\in\mathcal{Q}} \sup_{P\in\mathcal{P}(Q) : P(A|X,S)=Q(A|X)} P^n\{\|\underline{\theta}_\tau-\theta_Q\|^2\ge t\} \\
&\quad\ge \inf_\tau \sup_{Q\in\mathcal{Q}}  \min_{k : 0\le k\le r} Q^{k}\{\|\tau_{k}-\theta_Q\|^2\ge t\} \sup_{P\in\mathcal{P}(Q)} P^n\{K\le r\}.
\end{align*}
Upper bounding the left-hand side shows that
\begin{align*}
\inf_\tau \sup_{Q\in\mathcal{Q}} \sup_{P\in\mathcal{P}(Q)} P^n\{\|\underline{\theta}_\tau-\theta_Q\|^2\ge t\}&\ge \inf_\tau \sup_{Q\in\mathcal{Q}}  \min_{k : 0\le k\le r} Q^{k}\{\|\tau_{k}-\theta_Q\|^2\ge t\} \sup_{P\in\mathcal{P}(Q)} P^n\{K\le r\}.
\end{align*}
We will show that the supremum over $P$ lower bounds by some function $f$ applied to $n$ and $n_{XY}$. Let $r:= C_{\delta, v}n_{XY}$, where $C_{\delta, v}$ is the smallest constant makes $r$ as a positive integer and $C_{\delta, v} \ge 1+\sqrt{3\log(2/\delta)/n_{XY}}$. By choosing $r$ in this way, a multiplicative Chernoff bound yields $P^n\{K\le r\}\ge 1-\delta/2$ for each $P\in\mathcal{P}(Q)$, and so the above shows that
\begin{align}
    &\inf_\tau \sup_{Q\in\mathcal{Q}} \sup_{P\in\mathcal{P}(Q)} P^n\{\|\underline{\theta}_\tau-\theta_Q\|^2\ge t\} \nonumber \\
    &\ge (1-\delta/2)\,\inf_\tau \sup_{Q\in\mathcal{Q}}  \min_{k : 0\le k\le r} Q^{k}\{\|\tau_{k}-\theta_Q\|^2\ge t\} \nonumber \\
    &\ge (1-\delta/2)\,\sup_{\Pi}\inf_\tau \min_{k : 0\le k\le r} \int Q^{k}\{\|\tau_{k}-\theta_Q\|^2\ge t\} \Pi(dQ), \label{eq:bayesRiskLB}
\end{align}
where the final supremum is over priors $\Pi$ on $\mathcal{Q}$; this inequality---which reflects that the minimax risk is lower bounded by the worst-case Bayes risk---is essentially a consequence of the maximin inequality \citep{wald_sequential_1945}. The right-hand side above can be simplified by using that, for any $\Pi$,
\begin{align*}
    \inf_\tau \min_{k : 0\le k\le r} \int Q^{k}\{\|\tau_{k}-\theta_Q\|^2\ge t\} \Pi(dQ)\ge \inf_\tau \int Q^{k}\{\|\tau_{k}-\theta_Q\|^2\ge t\} \Pi(dQ).
\end{align*}
The inequality also trivially holds the other way, but that fact will not be used in this proof. To see why the above holds, fix an estimator sequence $\tau$, and let $\overline{\tau}$ be the same sequence except with $\overline{\tau}_{r}$ an estimator that takes as input $r$ observations and applies $\tau_{j^\star}$ to the first $j^\star$ of them, where $j^\star:=\argmin_{j : 0\le j\le r} \int Q^{j}\{\|\tau_{j}-\theta_Q\|^2\ge t\} \Pi(dQ)$. This construction then implies that
\begin{align*}
    \min_{k : 0\le k\le r}\int Q^{k}\{\|\tau_{k}-\theta_Q\|^2\ge t\} \Pi(dQ) = \int Q^{r}\{\|\overline{\tau}_{r}-\theta_Q\|^2\ge t\} \Pi(dQ),
\end{align*}
and taking an infimum over all possible estimators of $r$ observation on the right, followed by over all estimator sequences on the left, implies the claim. Returning to \eqref{eq:bayesRiskLB}, this shows that
\begin{align}
    \inf_\tau \sup_{Q\in\mathcal{Q}} \sup_{P\in\mathcal{P}(Q)} P^n\{\|\underline{\theta}_\tau-\theta_Q\|^2\ge t\}\ge (1-\delta/2)\,\sup_{\Pi}\inf_\tau \int Q^{r}\{\|\tau_{r}-\theta_Q\|^2\ge t\} \Pi(dQ). \label{eq:bayesLB}
\end{align}
Hence, we have shown that any lower bound on the Bayes risk under a least favorable prior in the problem with $r$ observations drawn from $Q$ also implies a lower bound on the minimax risk under sampling a random number, $K$, of observations from the target distribution.

\paragraph{Part 2: Using Fano's method to lower bound the worst-case Bayes risk.} In Part~3, we will follow arguments from \citet{zhang_optimality_2023} to construct $\{Q_0,Q_1,\ldots,Q_M\}\subset \mathcal{Q}$ with $M\ge 2^{r^{p/[p + (1 - 2\alpha)]}/8}$ such that, for $\theta_j:=\theta_{Q_j}$ and to-be-specified constants $c_{\delta}>0$ and $\kappa:=\delta/4$, the
\begin{enumerate}[label=(\roman*)]
    \item\label{it:funsFar} \textit{dose-response functions are far apart:} $\|\theta_{j} - \theta_{k}\| \geq c_{\delta} \kappa \; r^{-\frac{1-2\alpha}{p + (1 - 2\alpha)}}$ for all $0 \leq j < k \leq M$.
    \item\label{it:distsClose} \textit{distributions are close together:} $\frac{1}{M+1} \sum_{i=1}^{M} D_{\mathrm{KL}}(Q^{r}_{i}, Q^{r}_{0}) \leq \kappa \log M$.
\end{enumerate}
Before giving this construction, we motivate why it will be useful. This motivation begins with Eq.~2.8 from \citet{tsybakov_introduction_2009}, which shows that, if we take $t= c_{\delta} \kappa \; r^{-\frac{1-2\alpha}{p + (1 - 2\alpha)}}/2$, then
\begin{align*}
    Q_j^{r}\{\|\tau_{r}-\theta_j\|^2\ge t\}\ge Q_j^{r}\{\psi_\tau^*\not= j\},\ \ j=0,1,\ldots,M,
\end{align*}
where $\psi_\tau^*:=\argmin_{j\in \{0,1,\ldots,M\}} \|\tau^{r}-\theta_{Q_j}\|$ is a test based on $\tau$ that returns the index of a distribution most likely to have generated the data. Combining the fact that $\Pi_M$ is a prior with the above bound and the fact that $\psi_\tau^*$ is a test yields that the worst-case Bayes risk from \eqref{eq:bayesLB} lower bounds as
\begin{align*}
    &\sup_{\Pi}\inf_\tau \int Q^{r}\{\|\tau_{r}-\theta_Q\|^2\ge t\}\, \Pi(dQ) \\
    &\quad\ge \inf_\tau \frac{1}{M+1}\sum_{j=0}^M Q_j^{r}\{\|\tau_{r}-\theta_j\|^2\ge t\}\ge \inf_\tau \frac{1}{M+1}\sum_{j=0}^M Q_j^{r}\{\psi_\tau^*\not= j\}\ge \inf_\psi \frac{1}{M+1}\sum_{j=0}^M Q_j^{r}\{\psi\not= j\}.
\end{align*}
The final infimum above is over all tests, which take as input $r$ observations and return an element of $\{0,1,\ldots,M\}$ indicating a guess of which product measure in $\{Q_0^r,Q_1^r,\ldots,Q_M^r\}$ they were generated by. \citet{tsybakov_introduction_2009} refers to the right-hand side as the average probability of error. It can be lower bounded via Fano's method (Corollary~2.6 from that work), which, when combined with the above, yields that
\begin{align*}
    \sup_{\Pi}\inf_\tau \int Q^{r}\{\|\tau_{r}-\theta_Q\|^2\ge t\} \Pi(dQ)&\ge \frac{\log(M+1)-\log 2}{\log M} - \kappa\ge 1  - \frac{\log 2}{\log M} - \kappa.
\end{align*}
Plugging in $\kappa=\delta/4$ and $M\ge 2^{n/8}$ yields
\begin{align*}
    \sup_{\Pi}\inf_\tau \int Q^{r}\{\|\tau_{r}-\theta_Q\|^2\ge t\}\, \Pi(dQ)&\ge 1 - \frac{\log 2}{\log 2^{r^{p/[p + (1 - 2\alpha)]}/8}} - \frac{\delta}{4}\ge 1 - 8r^{-p/[p + (1 - 2\alpha)]} - \frac{\delta}
    {4}.
\end{align*}
Since $n_{XY}\ge N_\delta$ as defined in \eqref{eq:Ndelta}, $r:= C_{\delta, v}n_{XY}\ge  (\delta/32)^{-[p + (1 - 2\alpha)]/p}$, and so the right-hand side is lower bounded by $1-\delta/2$. Plugging in our chosen $t$ and recalling \eqref{eq:bayesLB} yields that
\begin{align*}
    \inf_\tau \sup_{Q\in\mathcal{Q}} \sup_{P\in\mathcal{P}(Q)} P^n\left\{\|\underline{\theta}_\tau-\theta_Q\|^2\ge \frac{c_{\delta}\delta}{8} r^{-\frac{1-2\alpha}{p + (1 - 2\alpha)}}\right\}&\ge (1-\delta/2)(1-\delta/2)\ge 1-\delta.
\end{align*}
Plugging in the fact that $r:= C_{\delta,v} n_{XY}\ge (1+\sqrt{3\log(2/\delta)/n_{XY}})n_{XY}$ yields that
\begin{align*}
    &\inf_\tau \sup_{Q\in\mathcal{Q}} \sup_{P\in\mathcal{P}(Q)} P^n\left\{\|\underline{\theta}_\tau-\theta_Q\|^2\ge \frac{c_{\delta}\delta}{8} (1+\sqrt{3\log(2/\delta)/n_{XY}})^{-\frac{1-2\alpha}{p + (1 - 2\alpha)}} n_{XY}^{-\frac{1-2\alpha}{p + (1 - 2\alpha)}}\right\}\ge 1-\delta.
\end{align*}
Finally, plugging in the inequality $c_\delta/8\ge c_{\mathcal{Q}}$ from the forthcoming \eqref{eq:cprimeIneq} for the constant $c_{\mathcal{Q}}$ defined therein, we find that
\begin{align*}
    &\inf_\tau \sup_{Q\in\mathcal{Q}} \sup_{P\in\mathcal{P}(Q)} P^n\left\{\|\underline{\theta}_\tau-\theta_Q\|^2\ge c_{\mathcal{Q}}\delta (1+\sqrt{3\log(2/\delta)/n_{XY}})^{-\frac{1-2\alpha}{p + (1 - 2\alpha)}} n_{XY}^{-\frac{1-2\alpha}{p + (1 - 2\alpha)}}\right\}\ge 1-\delta.
\end{align*}
The desired conclusion is at hand. Let $\underline{\theta}$ be the estimator from the lemma statement. Since it does not use data fusion, there exists $\underline{\tau}$ such that $\underline{\theta}=\underline{\theta}_{\underline{\tau}}$. Hence, by the above, for any $\epsilon>0$ there exist $Q\in\mathcal{Q}$ and $P\in\mathcal{P}(Q)$ such that \eqref{eq:minimaxLB} holds with probability at least $1-\delta-\epsilon$ under sampling from $P^n$; in particular, this holds when $\epsilon=\delta$.


\paragraph{Part 3: Constructing $Q_0,Q_1,\ldots,Q_M$.} We follow arguments from \citet{zhang_optimality_2023}. Let $m := \lceil r^{p/[p + (1 - 2\alpha)]} \rceil$. By the Varshamov-Gilbert lemma \citep[Lemma~2.9][]{tsybakov_introduction_2009}, there exist binary vectors $w^{(0)}, \ldots, w^{(M)} \in \{0,1\}^m$ for some $M \geq 2^{m/8}$ such that $\sum_{l = 1}^{m} (w_{l}^{(i)} - w_{l}^{(j)})^2 \geq \frac{m}{8}$ for all $i \neq j$. Let $\xi := C_\kappa m^{-[p + (1 - 2\alpha)]/p}$ for
\begin{align*}
    C_\kappa := \min\left\{
        2^{-(1-2\alpha)/p} \frac{R}{G_1}, \;
        \frac{\kappa \overline{\sigma}^{2} \log 2}{4}, \;
        \frac{1}{\big( \sup_l \|e_l\|_{\infty}\big)^{2}}
    \right\},
\end{align*}
and define
\[
\theta_{0} \equiv 0, 
\quad \theta_{i} := \xi^{1/2} \sum_{l=1}^{m} w_{l}^{(i)} e_{m+l}, \quad i = 1, \ldots, M,
\]
where $\{e_{l}\}$ are the orthonormal eigenfunctions of the integral operator $T_{\mathcal{K}}$, indexed so the corresponding eigenvalues are nonincreasing. Let $Q_X$ be the marginal distribution of $X$ under an arbitrarily chosen distribution in $\mathcal{Q}$---the particular choice made will play no role in this proof. Let $Q_i$ be the distribution of $(X,A,Y)$ that can be sampled from via $X\sim Q_X$, $A\mid X\sim \mu$, and $Y\mid A = X, A = x \sim \mathcal{N}( \theta_i(a), \overline{\sigma}^{2})$ where $\overline{\sigma} := \min(\sigma, L)$, where $\sigma$ and $L$ are defined in \Cref{cond:subexp-y}. Note that the dose response function of $Q_i$, $\theta_{Q_i}$, is equal to $\theta_i$. These distributions belong to the model $\mathcal{Q}$ since (i) their Gaussian errors easily satisfy the sub-exponential condition in \Cref{cond:subexp-y}, (ii) $\E_{Q_i}[Y|A,X]$ is $Q^{0}_X\times \mu$-a.s. bounded by 1 since $p\le 1-2\alpha$ ensures $\xi^{1/2}\le C_\kappa^{1/2} m^{-1}\le 1/(m\sup_l \|e_l\|_{\infty})=1/(mM_{e})$, and (iii) the source condition \Cref{cond:src} is also satisfied since
\begin{align*}
\|T_{\mathcal{K}}^{\alpha} \theta_{i}\|_{\mathcal{H}}
= \xi \sum_{k=1}^{m} \sigma_{m+k}^{-(1-2\alpha)} \left(w_{k}^{(i)}\right)^{2}
\leq \xi \sum_{k=1}^{m} \sigma_{2m}^{-(1-2\alpha)}
\leq \xi \sum_{k=1}^{m} G_{1} \left( 2m \right)^{\frac{1-2\alpha}{p}}
= C_\kappa G_{1} 2^{\frac{1-2\alpha}{p}}\le R.
\end{align*}

Having constructed $Q_0,Q_1,\ldots,Q_M\in\mathcal{Q}$, it remains to show they satisfy \ref{it:funsFar} and \ref{it:distsClose} from Part~2 of this proof. For \ref{it:funsFar}, we use the orthonormality of $\{e_{l}\}$, the choice of Varshamov-Gilbert weights, and the fact that $\lceil r^{p/[p + (1 - 2\alpha)]} \rceil\le 2r^{p/[p + (1 - 2\alpha)]}$ since $r\ge 1$ to show that, for all $i\not=j$,
\begin{align*}
\|\theta_i-\theta_j\|_{L^2(\mu)}^2
& = \xi \sum_{l=1}^{m} \left( w_{l}^{(i)} - w_{l}^{(j)} \right)^{2} \geq \xi \, \frac{m}{8} = \frac{(C_\kappa/\kappa) \kappa \; m^{-(1-2\alpha)/p}}{8} \\
&\ge \frac{(C_\kappa/\kappa) \kappa \; r^{-\frac{1-2\alpha}{p + (1 - 2\alpha)}}}{8}2^{-(1-2\alpha)/p} = c_{\delta} \kappa \; r^{-\frac{1-2\alpha}{p + (1 - 2\alpha)}},
\end{align*}
where $c_{\delta}:=\frac{(C_\kappa/\kappa)}{8}2^{-(1-2\alpha)/p}$. The $\delta$ subscript indicates the dependence of $c_\delta$ on $\kappa=\delta/4$. When divided by $8$ as at the end of Part~2 of this proof, this quantity lower bounds as
\begin{align}
    c_{\delta}/8&=\min\left\{ 2^{-2(1-2\alpha)/p} R/(16\delta G_1), \; (\overline{\sigma}^{2} \log 2)2^{-(1-2\alpha)/p}/256, \frac{2^{-(1-2\alpha)/p}}{16\delta \big( \sup_l \|e_l\|_{\infty}\big)^{2}} \right\} \nonumber \\
    &\ge \min\left\{ 2^{-2(1-2\alpha)/p} R/(16 G_1), \; (\overline{\sigma}^{2} \log 2)2^{-(1-2\alpha)/p}/256, \frac{2^{-(1-2\alpha)/p}}{16 \big( \sup_l \|e_l\|_{\infty}\big)^{2}} \right\} =: c_{\mathcal{Q}}. \label{eq:cprimeIneq}
\end{align}
It remains to show \ref{it:distsClose}. For this, we use that the errors are Gaussian, which yields that, for any $i\not=0$,
\begin{align*}
    D_{\mathrm{KL}}(Q_{i}^r, Q_{0}^r) &= \frac{r}{2\overline{\sigma}^{2}} \| \theta_{i} \|_{L^{2}(\mu)}^{2} = \frac{r\xi}{2\overline{\sigma}^{2}} \sum_{k=1}^{m} \left(w_{k}^{(i)}\right)^{2} \leq \frac{r\xi m}{2\overline{\sigma}^{2}} \\
    &= \frac{rC_{\kappa}m^{-\frac{1-2\alpha}{p}}}{2\overline{\sigma}^{2}} \leq \frac{C_{\kappa}m}{2\overline{\sigma}^{2}} \leq \frac{4C_{\kappa}\log{M}}{\overline{\sigma}^{2} \log2} \leq \kappa \log M.
\end{align*}
As $i\not=0$ was arbitrary, \ref{it:distsClose} holds.

\end{proof}

\subsection{Sufficient Conditions for Data Fusion to Improve Worst-Case Performance (Proof of \Cref*{thm:minimax_bound})}
\begin{proof}
We will show that, for all $n,n_{XY}$ large enough, the first term in \eqref{eq:minimaxDominance} is at least $1-\delta$ and the second term is at most $\delta$. For brevity, let $\rho := (1-2\alpha)/[p+(1-2\alpha)]\in(0,1)$.
 
\paragraph{Step 1. Show $\sup_{P\in\mathcal P_{XY}} P^{n}\!\left( L_{P}(\underline{\theta},g_{0}) \ge \tfrac{c_{\mathcal{Q}}}{4}\,\delta\, n_{XY}^{-\rho} \right) \ge 1-\delta.$}
By \Cref{lem:minimax_bound}, for all $(n,n_{XY})$ with $n \ge n_{XY} \ge N_1(\delta)$, there exists $P \in \mathcal P_{XY}$ with $P(S \in \mathcal S_X \cap \mathcal S_Y) = n_{XY}/n$ such that
\begin{equation*}
P^{n}\!\left(
  L_{P}(\underline{\theta},g_{0}) 
  \ \ge\ 
  \frac{c_{\mathcal{Q}}}{2}\,\delta\,
  \Bigl(1+\sqrt{3\log(2/\delta)/n_{XY}}\Bigr)^{-\rho}\,
  n_{XY}^{-\rho}
\right)\ \ge\ 1-\delta,
\end{equation*}
where $c_{\mathcal{Q}}>0$ is a constant depending on $\mathcal{Q}$. Using that $(1+x)^{-\rho}\ge 1-\rho x$ for $x\ge 0$, we have
\begin{equation*}
\Bigl(1+\sqrt{3\log(2/\delta)/n_{XY}}\Bigr)^{-\rho} 
\ \ge\ 
1-\rho\sqrt{3\log(2/\delta)/n_{XY}}.
\end{equation*}
Hence, if
\[
n_{XY}\ \ge\ N_2(\delta)
\ :=\ \left\lceil \frac{3\rho^{2}\,\log(2/\delta)}{\delta^{2}} \right\rceil,
\]
then $\Bigl(1+\sqrt{3\log(2/\delta)/n_{XY}}\Bigr)^{-\rho} \ge 1-\delta$.
Therefore, for all $n \ge n_{XY} \ge \max\{N_1(\delta),N_2(\delta)\}$,
\begin{equation*}
P^{n}\!\left( L_{P}(\underline{\theta},g_{0}) \ \ge\ \frac{c}{2}\,\delta(1-\delta)\, n_{XY}^{-\rho} \right)\ \ge\ 1-\delta.    
\end{equation*}
This shows that, for all $n\ge n_{XY}\ge \max\{N_1(\delta),N_2(\delta)\}$, 
\begin{equation*}
\sup_{P\in\mathcal P_{XY}} P^{n}\!\left( L_{P}(\underline{\theta},g_{0}) \ \ge\ \frac{c_{\mathcal{Q}}}{2}\,\delta(1-\delta)\, n_{XY}^{-\rho} \right)\ \ge\ 1-\delta,   
\end{equation*}
and combining this with the fact that $\delta(1-\delta)\ge \delta/2$ for $\delta\in(0,1/2]$ yields that, for all such $n,n_{XY}$,
\begin{equation*}
\sup_{P\in\mathcal P_{XY}} P^{n}\!\left( L_{P}(\underline{\theta},g_{0}) \ \ge\ \tfrac{c_{\mathcal{Q}}}{4}\,\delta\, n_{XY}^{-\rho} \right)\ \ge\ 1-\delta.
\end{equation*}

\paragraph{Step 2. Show $\sup_{P\in\mathcal P_{XY}} P^{n}\!\left( L_{P}(\widehat{\theta},g_{0}) > r(\delta, n)/\log(n)^{h\rho/2} \right) < \delta.$}

To prove the result, we will apply Theorem~3 of \citet{foster_orthogonal_2023} as we did in the proof of \Cref{thm:rkhs_excess}. We first establish that the loss is Lipschitz in its first argument with a Lipschitz constant uniform over $P \in \mathcal P_{XY}$. To see this, for any $P \in \mathcal P_{XY}$, $\theta_{1}, \theta_{2} \in \Theta$, and realization $z:= (x,a,b,y,s)$ of $Z \sim P \times \mu$, 
\begin{align*}
\big|\ell_{P}(\theta_{1}, \widehat{g}; z) - \ell_{P}(\theta_{2}, \widehat{g}; z)\big|
= \big\| \left( \theta_{1}(a) - \theta_{2}(a),\; \theta_{1}(b) - \theta_{2}(b) \right) \big\|_{2}\cdot B(z),
\end{align*}
where $B(z)$ is defined as
\begin{equation*}
B(z) := \left\|
\begin{pmatrix}
\theta_{1}(b) + \theta_{2}(b) - 2\widehat{\tau}(b)
+ 2\widehat{\xi}\,\mathbf{1}(s \in \mathcal{S}_{X})\big(\widehat{\tau}(b) - \widehat{m}(b,x)\big) \\[0.25em]
2\,\widehat{\eta}\,\widehat{w}(x,a)\,\mathbf{1}(s \in \mathcal{S}_{Y})\big(\widehat{m}(x,a) - y\big)
\end{pmatrix}
\right\|_{2}.
\end{equation*}
We next derive a uniform bound on $B(z)$. First, observe that
\begin{align*}
B(z)
&\le \Big| \theta_{1}(b)+\theta_{2}(b)-2\widehat{\tau}(b)
+ 2\widehat{\xi}\,\mathbf{1}(s \in \mathcal{S}_{X})\big(\widehat{\tau}(b)-\widehat{m}(b,x)\big) \Big| \\
&\quad + \Big| 2\,\widehat{\eta}\,\widehat{w}(x,a)\,\mathbf{1}(s \in \mathcal{S}_{Y})\big(\widehat{m}(x,a)-y\big) \Big| \\
&\le 2\|\theta\|_{\infty} + 2|\widehat{\tau}(b)|
+ 2\widehat{\xi}\,\big| \widehat{\tau}(b)-\widehat{m}(b,x)\big|
+ 2\widehat{\eta}\,\widehat{w}(x,a)\,\mathbf{1}(s \in \mathcal{S}_{Y})\big(|\widehat{m}(x,a)|+|y|\big) \\
&\le 4 + 4\widehat{\xi} + 2\widehat{\eta}\,\widehat{w}(x,a)\,\mathbf{1}(s \in \mathcal{S}_{Y})\big(1+|y|\big) \\
&\le 4 + 4M_{\xi} + 2M_{\eta}M_{w}\,\mathbf{1}(s \in \mathcal{S}_{Y})\,(1+|y|)
\qquad Q^{0}_{X}\times\mu\text{-a.e.},
\end{align*}
by $\|\theta\|_{\infty}\le 1$ and \Cref{cond:nuisances}. By \Cref{cond:subexp-y} and the same argument as in the proof of \Cref{thm:excess_risk}, with probability at least $1-\delta/4$,
\[
\mathbf{1}(s \in \mathcal{S}_{Y})\,\eta_{0}\,\|w_{0}\|_{\infty}\,|y|
\;\le\; C_{\sigma, \delta}\,\eta_{0}\,\|w_{0}\|_{\infty}.
\]
Since $\eta_{0}\,\|w_{0}\|_{\infty}\ge (M_{\eta}M_{w})^{-1}>0$, it follows that
$\mathbf{1}(s \in \mathcal{S}_{Y})\,|y|\le C_{\sigma, \delta}$, which further implies
\begin{align*}
B(z)
&\le 4\big(1 + M_{\xi} + C_{\delta}\,M_{\eta}M_{w}\big)
=: B_{\mathcal{P}}
\qquad Q^{0}_{X}\times\mu\text{-a.e.}
\end{align*}
Then we apply Theorem~3 of \citet{foster_orthogonal_2023} with the parameters
\begin{align*}
K_{2} = 2, \quad R = \sqrt{2}, \quad \beta_{1} = 2, \quad r = p, \quad \lambda = 2, \quad \kappa = 0, \quad B_{1} = B_{\mathcal{P}}^{2}, \quad B_{2} = (M_{\lambda}C_{p}c^{p})^{\frac{2}{1+p}},
\end{align*}
and plug in the valid critical radius $\delta_{n} \propto c^{\frac{p}{1+p}} (n\log(n)^{-2})^{-\frac{1}{2+p}}$ resulting in the following bound:
\begin{align*}
L_{P}(\widehat{\theta}, g_{0}) - L_{P}(\theta^{*}, g_{0})
&\lesssim  c^{\frac{2p}{1+p}} (n\log(n)^{-2})^{-\frac{1}{1+p}} 
   + \frac{\log(\delta^{-1})}{n} 
   + c^{\frac{2p}{1+p}} 
     \|\widehat{g} - g_{0} \|^{\frac{4}{1+p}}_{L^{2}(Q_{X} \times \mu, \ell^{2})}
\end{align*}
with probability at least $1-3/4\delta$. By \eqref{eq:ghatUniform}, there exists $N_3(\delta)<\infty$ such that, for all $n\ge N_3(\delta)$, for all $P \in \mathcal P_{XY}$,
\begin{align*}
\|\widehat g-g_0\|_{L^2(Q_X^0\times\mu;\ell^2)} \le M_g \ (n\log(n)^{-2})^{-1/4},
\end{align*}
with probability at least $1-\delta/4$. Then, the nuisance estimation error can be absorbed into the target estimation error, yielding, for all $n \ge N_{3}(\delta)$,
\begin{align*}
L_{P}(\widehat{\theta}, g_{0}) - L_{P}(\theta^{*}, g_{0}) \lesssim c^{\frac{2p}{1+p}} (n\log(n)^{-2})^{-\frac{1}{1+p}}
\end{align*}
with probability at least $1-\delta$.
Then for any $P \in \mathcal P_{XY}$, the excess risk with respect to the true CDRF $\theta_{0}$ is bounded by
\begin{align*}
L_{P}(\widehat{\theta}, g_{0}) \le C' (n\log(n)^{-2})^{-\rho},
\end{align*}
by plugging in $c \propto (n\log(n)^{-2})^{\rho}$ and $C'$ is an universal constant. Hence, for $n \ge N_{3}(\delta)$,
\begin{align*}
    \sup_{P \in \mathcal P_{XY}}P^{n}(L_{P}(\widehat{\theta}, g_{0}) \le C' (n\log(n)^{-2})^{-\rho}) \ge 1 - \delta.
\end{align*}

Next, using $n_{XY}\le n/(\log n)^{2+h}$, we can rewrite the right-hand side in terms of $n_{XY}$: for some $N_4(\delta)<\infty$ and all $n\ge N_4(\delta)$,
\begin{align*}
C'\,(n\log(n)^{-2})^{-\rho} \le \tfrac{c_{\mathcal{Q}}}{4}\,\delta\,n_{XY}^{-\rho}\,(\log n)^{-h\rho/2} = r(\delta,n)\,(\log n)^{-h\rho/2}.
\end{align*}

Letting 
\begin{align}\label{eq:large_n_minimax}
N_{\delta}:= \max\{N_{1}(\delta), N_{2}(\delta), N_{3}(\delta), N_{4}(\delta)\},
\end{align}
we conclude
\begin{align*}
    &\sup_{P\in\mathcal{P}_{XY}} P^n\!\Bigl\{ L_P(\underline{\theta}, g_0) \ge r(\delta,n)\Bigr\}- \sup_{P\in\mathcal{P}_{XY}} P^n\!\Bigl\{ L_P(\widehat{\theta}, g_0) > r(\delta,n)/\log(n)^\frac{h(1-2\alpha)}{2[p+(1-2\alpha)]} \Bigr\}
    \;\ge\; 1 - 2\delta.
\end{align*}
\end{proof}

\end{document}